\documentclass[aos,preprint]{imsart}

\RequirePackage{amsthm,amsmath,amsfonts,amssymb,float, bm, mathtools,multirow}
\RequirePackage[numbers]{natbib}
\RequirePackage[colorlinks,citecolor=blue,urlcolor=blue]{hyperref}
\RequirePackage{graphicx}
\usepackage{xcolor}
\usepackage{comment}

\DeclareMathOperator*{\argmin}{arg\,min}
\startlocaldefs

\newtheorem{theorem}{Theorem}[subsection]
\newtheorem{lemma}[theorem]{Lemma}

\newtheorem{proposition}{Proposition}
\newtheorem{definition}[theorem]{definition}
\newcommand{\bs}[1] {\bm{#1}}

\DeclareMathOperator*{\Cov}{\rm Cov}

\endlocaldefs

\newcommand{\wang}{\textcolor{orange}}
\begin{document}

\begin{frontmatter}
\title{Optimal  Classification for Functional Data}
\runtitle{Optimal Functional Data Classification}

\begin{aug}
\author[A]{\snm{Shuoyang} \fnms{Wang} \ead[label=e1]{szw0100@auburn.edu}},
\author[B]{\snm{Zuofeng} \fnms{Shang} \ead[label=e2]{zshang@njit.edu}}
\author[A]{\snm{Guanqun} \fnms{Cao} \ead[label=e3]{gzc0009@auburn.edu}}
\and
\author[C]{\snm{Jun S.} \fnms{Liu} \ead[label=e4]{jliu@stat.harvard.edu}}
\address[A]{Department Mathematics and Statistics, Auburn University}

\address[B]{Department of Mathematical Sciences, New Jersey Institute of Technology}

\address[C]{Department of Statistics, Harvard University}
\end{aug}

\begin{abstract}
A central topic in functional data analysis is how to design an optimal decision rule, based on training samples, to classify a data function.
We exploit the optimal   classification problem when data functions are Gaussian processes.
 Sharp nonasymptotic convergence rates for minimax excess misclassification risk are derived in both settings that data functions are fully observed and discretely observed.
 We explore two easily implementable classifiers based on discriminant analysis and deep neural network, respectively, which are both proven to achieve optimality in Gaussian setting. 
 Our deep neural network classifier is new in literature which demonstrates outstanding performance even when data functions are non-Gaussian.
 In case of discretely observed data, we discover a novel critical sampling frequency that governs the sharp convergence rates.
The proposed classifiers  perform favorably in finite-sample applications, as we demonstrate through comparisons with other functional classifiers in
simulations and  one real data application.

\end{abstract}



\end{frontmatter}
\section{Introduction}
In many applications, data are collected in the form of functions such as curves or images. Such data are nowadays commonly
referred to as functional data. 
A fundamental problem in functional data analysis is to 
classify a data function based on training samples.
For instance, in the speech recognition data extracted from the TIMIT database \cite{Ferraty:Vieu:03}, 
the training samples are digitized speech curves of American English speakers from different phoneme groups, and the task is to predict the
phoneme of a new speech curve. 
Classic multivariate analysis techniques such as logistic regression or discriminant analysis are not directly applicable,
since functional data are intrinsically infinite-dimensional \cite{Wang:etal:16}. A common strategy is to adapt multivariate analysis to functional settings such as functional logistic regression \cite{Araki:etal:09} and functional discriminant analysis \cite{Shin:08,Delaigle:etal:12,Delaigle:Hall:12, Delaigle:Hall:13,Galeano:etal:15,Dai:etal:17,Berrendero:etal:18,Park:etal:20}, among others. Despite their impressive performances, one is often interested in knowing whether and which of these approaches are statistically optimal, and if not, how to construct an optimal functional classifier with better performances.

Optimal classification has been investigated in multivariate settings 
\cite{Mammen:etal:99,Tsybakov:04,Lecue:08,Farnia:Tse:16,Cai:Zhang:19,Cai:Zhang:19b,Mazuelas:etal:20}.
The term ``optimality'' refers to minimizing the excess misclassification risk relative to the oracle Bayes rule,
which provides a theoretical understanding on the nature of the problem as well as a benchmark
to measure the performance of various classifiers.
Optimal classification in functional setting is more challenging due to the  
infinite-dimensional characteristic of the data.
Existing work such as \cite{Delaigle:Hall:12} focuses on the special case that the Bayes risk vanishes, called as perfect classification. 
As revealed in \cite{Berrendero:etal:18}, the Bayes risk vanishes when the probability measures of the populations are mutually singular.
If  the two populations have equivalent probability measures, i.e., the singularity fails,  the Bayes risk does not vanish. The latter scenario is more challenging since the two populations are much ``closer'' to each other
in the sense that the differences of the population means and covariances are sufficiently smooth, and
there is a lack of literature on how to design an optimal functional classifier.

In this paper, we  investigate the optimal classification problem under the Gaussian setting,
i.e., the observed data are Gaussian processes. In the nonvanishing Bayes risk setting, we
derive sharp nonasymptotic rates for the Minimax Excess Misclassification Risk (MEMR) which provides a theoretical understanding on how well one can approximate the Bayes risk based on training samples. Our results cover both fully observed data and discretely observed data.
We also show that functional quadratic discriminant analysis (FQDA) and Functional Deep Neural Network (FDNN) are both able to achieve sharp rates of MEMR, hence, are minimax optimal.

Functional discriminant analysis is a popular technique in classifying Gaussian data (\cite{Galeano:etal:15, Dai:etal:17}), whereas its optimality remains open. Our work provides the first rigorous analysis to fill this gap. 
Specifically, we derive an upper bound for the excess misclassification risk of FQDA in Gaussian setting which matches the sharp rate of MEMR.
In conventional settings such as low- or high-dimensional data classification, the optimality of the discriminant analysis 
has been established by \cite{anderson:2003,Cai:Zhang:19,Cai:Zhang:19b}.
Our work can be viewed as a nontrivial extension of their results to functional data.
In practice, FQDA is known to perform poorly when data are non-Gaussian,
so it is desirable to design a classifier that is robust to the violation of the Gaussian assumption.
We propose a novel FDNN classifier based on deep neural network (DNN) to address this issue. FDNN is proven to achieve the same optimality as FQDA in the Gaussian setting, as well as exhibits better classification accuracy when data are non-Gaussian. 
DNN has been recently applied in various nonparametric problems; see \cite{Schmidt:19,Bauer:Kohler:19,Kim:etal:21,Liu:19,Liu:etal:2021,hu2020arxiv}. The present work provides the first application of DNN in functional data classification with provable guarantees.

In the setting of discretely observed data, the rate of convergence for MEMR
demonstrates an interesting phase transition phenomenon jointly characterized by
the number of data curves and the sampling frequency. The discretely observed data scenario is practically meaningful since,
in real-world problems, functional data can only be observed at discrete sampling points.
Our analysis reveals that when the sampling frequency is relatively small,
the number of data curves has little effect on the rate of MEMR.
When sampling frequency is relatively large, the rate of MEMR
depends more on the  number of data curves.
In other words, there exists a critical sampling frequency that
governs the performance of the minimax optimal classifier.
In functional regression, the existence of a critical sampling frequency 
that governs the optimal estimation has been discovered by \cite{CY:11}.
The present work has made a relevant and new discovery in functional classification.

The rest of the paper is organized as follows. Section \ref{SEC:prelim} provides some background
on the functional Bayes classifier and optimal functional classification. Section \ref{SEC:theory} establishes sharp nonasymptotic rates for MEMR 
in both scenarios of fully observed data and discretely observed data.
Sections \ref{SEC:opt:classifier} and \ref{SEC:FDNN} propose FQDA and FDNN classiers, both proven optimal. Section \ref{SEC:simulation} compares FQDA and FDNN with existing functional classification methods through simulation. 
Section 7 illustrates an application of our method  to analyze speech recognition dataset. 
Section \ref{SEC:discussion} concludes the paper with a brief summary. 
Major technical details for the proofs of main results are  deferred to the  Appendix.

{\it Notation and Terminologies.} We introduce some basic notations and definitions
that will be used throughout the rest of the paper. Vectors and matrices are denoted by boldface letters.
For
a matrix $\mathbf{A} \in \mathbb{R}_{p\times p}$,   $|\mathbf{A} |$ is the determinant of $\mathbf{A} $, and $\mathbf{I}_p$ is the $p\times p$ identity matrix.
For two sequences of
positive numbers $a_n$ and $b_n$, $a_n \lesssim b_n$ means that for some constant $c > 0$,
$a_n \leq c b_n$ for all $n$,  $a_n \asymp b_n$ means $a_n \lesssim b_n$ and $b_n \lesssim a_n$, and $a_n \ll b_n$ means $\lim_{n\rightarrow \infty}|a_n|/|b_n|=0$. We
also use $c,c_0,c_1,\ldots,C,C_0,C_1,\ldots$ to denote absolute constants whose values may vary from place to place.




\section{Preliminaries}\label{SEC:prelim}
In this section, we review some background on functional Bayes classifier and optimal classification in Gaussian setting.

Let $Z(t),t\in\mathcal{T}:=[0,1]$ be a random process.
We say that $Z$ belongs to class $k$ if
$Z\sim\mathcal{GP}({\eta}_k, {\Omega}_k)$ for $k=1,2$,
where $\mathcal{GP}({\eta}_k, {\Omega}_k)$ is a Gaussian process
with unknown mean function $\eta_k$ and unknown covariance function $\Omega_k$. For $k=1,2$, let $\pi_k\in(0,1)$ be the unknown probability of $Z$ belonging to class $k$ which satisfy $\pi_1+\pi_2=1$.
Suppose that $\Omega_k$ satisfies eigen-decomposition:
\begin{equation}\label{eqn:Omegak}
\Omega_k(s,t)=\sum_{j=1}^\infty\lambda_j^{(k)}\psi_j(s)\psi_j(t), s,t\in\mathcal{T},
\end{equation}
where $\psi_j, j\ge1$ is an orthonormal basis of $L^2(\mathcal{T})$ w.r.t. the usual $L^2$ inner product $\langle\cdot,\cdot\rangle$,
and $\lambda_j^{(k)}$ are positive eigenvalues.  
Note that (\ref{eqn:Omegak}) requires the covariance functions possessing the same eigenfunctions which is a common assumption in literature for technical convenience; see \cite{Delaigle:Hall:12, Dai:etal:17}.
Write $\eta_k(t)=\sum_{j=1}^{\infty} \mu_{kj} \psi_j(t)$
 and $Z(t)=\sum_{j=1}^\infty z_j \psi_j(t)$,
 where $\mu_{kj}$ represent the projection scores of 
 $\eta_k$ and $z_j$ represent the projection scores of $Z$.
It is easy to see that, when $Z$ belongs to class $k$,
$z_j$'s are pairwise uncorrelated with mean $\mu_{kj}$
and variance $\lambda_j^{(k)}$. 

Define
$\bs{\theta}=\left(\pi_1,\pi_2,\bs{\mu}_1,\bs{\mu}_2,\bs{\Sigma}_1,\bs{\Sigma}_2\right)$,
in which $\bs{\mu}_k=(\mu_{k1},\mu_{k2},\ldots)$ is the infinite sequence of mean projection scores
and $\bs{\Sigma}_k$ is a diagonal linear operator from $L^2(\mathcal{T})$ to $L^2(\mathcal{T})$
satisfying $\bs{\Sigma}_k\psi_j=\lambda_j^{(k)}\psi_j$, for $j\ge1$ and $k=1,2$.
Given $\bs{\theta}$, it follows by \cite{Berrendero:etal:18} and \cite{Torrecilla:etal:2020} that the functional Bayes rule
for classifying a new data function $Z\in L^2(\mathcal{T})$ has an expression
\begin{eqnarray}\label{DEF:FBAYES}
G^\ast_{\bs{\theta}}(Z)=
\begin{cases}
      1, & Q^\ast(Z, \bs\theta)\geq 0,\\
      2, & Q^\ast(Z, \bs\theta)< 0,
   \end{cases}
\end{eqnarray}
where
\[
Q^\ast(Z,\bs\theta)=\langle\mathbf{D}(Z-\eta_1),Z-\eta_1\rangle-2\langle\bs{\Sigma}_2^{-1}(\eta_2-\eta_1),Z-\bar{\eta}\rangle-\log\text{det}\left(\bs{\Sigma}_2^{-1}\bs{\Sigma}_1\right)+\log\left(\frac{\pi_1}{\pi_2}\right)
\]
is the quadratic discriminant functional
in which $\bar{\eta}=(\eta_1+\eta_2)/2$, $\mathbf{D}=\bs{\Sigma}_2^{-1}-\bs{\Sigma}_1^{-1}$
(difference of inverse operators),
and
\[
\text{det}\left(\bs{\Sigma}_2^{-1}\bs{\Sigma}_1\right)=\exp\left\{\sum_{s=1}^\infty\frac{(-1)^{s-1}}{s}\text{Tr}\left(\left[\bs{\Sigma}_2^{-1}\bs{\Sigma}_1-\text{id}\right]^s\right)\right\}
\]
is the infinite determinant of $\bs{\Sigma}_2^{-1}\bs{\Sigma}_1$ (see Plemelj’s formula in \cite{Simon:77}).

In practice, $G^\ast_{\bs{\theta}}$ is unobservable since $\bs{\theta}$ is unknown.
Suppose we observe a training sample $\{X_i^{(k)}(t): 1\le i\le n_k, k=1,2, t\in\mathcal{T}\}$, 
where $n_k$ is the sample size for class $k$,
$X_i^{(k)}\sim\mathcal{GP}({\eta}_k, {\Omega}_k)$, all $X_i^{(k)}$'s are independent, and are independent of $Z$ to be classified.
For a generic classifier $\widehat{G}$ constructed using the training samples, its performance is measured by
the misclassification risk $R_{\bs{\theta}}(\widehat{G})=E_{\bs{\theta}}[\mathbb{I}\{\widehat{G}(Z)\neq Y(Z)\}]$
under the true parameter $\bs{\theta}$, where $Y(Z)$ denotes the unknown label of $Z$.

Following \cite{Delaigle:Hall:12,Dai:etal:17}, if
\begin{equation}\label{imperf:clas:cond}
\textrm{both $\sum_{j=1}^\infty {\left(\mu_{1j}-\mu_{2j} \right)^2}/{\lambda_{j}^{(2)}}$
and $\sum_{j=1}^\infty\left(\lambda^{(1)}_j/\lambda_{j}^{(2)}-1\right)^2$ are convergent,}
\end{equation}
then $R_{\bs{\theta}}(G^\ast_{\bs{\theta}})>0$. 
Classification under (\ref{imperf:clas:cond}) is challenging since
the two Gaussian measures are asymptotically equivalent; see  \cite{Berrendero:etal:18}.
Since $G^\ast_{\bs{\theta}}$ achieves the smallest risk,
it is impossible to design a classifier  {with}  zero risk.
Instead, we aim to construct a classifier $\widehat{G}$ based on training samples
that performs similarly as $G^\ast_{\bs{\theta}}$,
which motivates the study of  MEMR:
\[
\inf_{\widehat{G}}\sup_{\theta\in\Theta}E[ R_{\bs\theta}(\widehat{G}) - R_{\bs\theta}(G^{\ast}_{\bs\theta})],
\]
where the infimum is taken over all functional classifiers constructed using the training samples and $\Theta$ is a parameter space to be described in the following section.

\section{Sharp nonasymptotic rates for MEMR}\label{SEC:theory}
We derive sharp  nonasymptotic rates for MEMR in both scenarios of 
fully observed data and discretely observed data. To the best of our knowledge, these are the first results exploring MEMR in functional setting. 
\subsection{Parameter space}
Our MEMR results rely on an explicit parameter space for $\bs{\theta}$.
We shall first introduce the concepts of hyperrectangles and Sobolev balls.

\begin{definition}
A hyperrectangle of order $\omega>0$ and length $A>0$ is defined as
\begin{equation}\label{Hyper only}
    H^\omega(A) = \left\{ \bs{a}=(a_1,a_2,\ldots): \sup_{j\geq 1} |a_j| j^{1+\omega}  \leq A\right\}.
\end{equation}
\end{definition}
An implication of $\bs{a}\in H^\omega(A)$ is that $|a_k|\le A k^{-(1+\omega)}$ for any $k\ge1$, in which $\omega$ governs the decay rate of the coordinates.

\begin{definition}
A $\ell_1$-Sobolev ball of order $\omega>0$ and radius $A>0$ is defined as
\begin{equation}\label{Sobolev only}
    S^\omega(A) = \left\{ \bs{a}=(a_1,a_2,\ldots): \sum_{j=1}^\infty |a_j| j^{\omega}  \leq A\right\}.
\end{equation}
\end{definition}
An implication of $\bs{a}\in S^\omega(A)$ is that $\sum_{k=L}^{\infty}|a_k|\le  A L^{-\omega}$ for any $L\ge1$, in which $\omega$ governs the decay rate of the tail sum.

Hyperrectangles and Sobolev balls depict different perspectives on a real sequence: the former controls a sequence element-wisely, while the latter controls its tail sum.  Although  overlapping, hyperrectangles and Sobolev balls do not include each other. 

In the rest of this article, consider the following two parameter spaces for $\bs{\theta}$. For $\nu_1,\nu_2>0$,
\begin{eqnarray} \label{DEF:Theta hyper}
{\Theta}_H(\nu_1,\nu_2) &:=& \left\{ \right.\bs{\theta}: \left\{\mu_{1j}^2 \vee \mu_{2j}^2\right\}_{j\geq 1} \in H^{\nu_1}, \left\{{\lambda_{j}^{(1)}}\vee {\lambda_{j}^{(2)}}\right\}_{j\geq 1} \in H^{\nu_1}, \nonumber\\
&& \left\{ \left(\mu_{1j}-\mu_{2j}\right)^2 /{\lambda_{j}^{(2)}}\right\}_{j\geq 1} \in H^{\nu_2}, \left\{ (\lambda^{(1)}_j/\lambda_{j}^{(2)}-1)^2\right\}_{j\geq 1}\in H^{\nu_2}, \nonumber\\
&&C_0 \leq \pi_1, \pi_2 \leq 1-C_0 \left. \right\},
\end{eqnarray}
and 
\begin{eqnarray} \label{DEF:Theta soblev}
{\Theta}_S(\nu_1,\nu_2) &:=& \left\{ \right.\bs{\theta}: \left\{\mu_{1j}^2 \vee \mu_{2j}^2\right\}_{j\geq 1} \in S^{\nu_1}, \left\{{\lambda_{j}^{(1)}}\vee {\lambda_{j}^{(2)}}\right\}_{j\geq 1} \in S^{\nu_1}, \nonumber\\
&& \left\{ \left(\mu_{1j}-\mu_{2j}\right)^2 /{\lambda_{j}^{(2)}}\right\}_{j\geq 1} \in S^{\nu_2}, \left\{ (\lambda^{(1)}_j/\lambda_{j}^{(2)}-1)^2\right\}_{j\geq 1}\in S^{\nu_2}, \nonumber\\
&& C_0 \leq \pi_1, \pi_2 \leq 1-C_0 \left. \right\},
\end{eqnarray}
where $C_0\in \left( 0, 1/2\right)$
is constant, $H^\omega=H^\omega(A)$,
and $S^\omega=S^\omega(A)$. For notation simplicity, $A$ is omitted. 
Specifically, $\bs{\theta}\in {\Theta}_H(\nu_1,\nu_2)$ 
implies that $\mu_{kj}^2$, $\lambda_j^{(k)}$ belong to $H^{\nu_1}$, and $(\mu_{1j}-\mu_{2j})^2/\lambda_{j}^{(2)}$, $(\lambda^{(1)}_j/\lambda_{j}^{(2)}-1)^2$
belong to $H^{\nu_2}$; $\nu_1$ governs the smoothness of the mean functions and covariance functions, and $\nu_2$ governs the separation of the two populations.
Moreover, the series $\sum_{j=1}^\infty {\left(\mu_{1j}-\mu_{2j} \right)^2}/{\lambda_{j}^{(2)}}$ and
$\sum_{j=1}^\infty(\lambda^{(1)}_j/\lambda_{j}^{(2)}-1)^2$ are both convergent which implies that Bayes risk is nonvanishing; see (\ref{imperf:clas:cond}). One can interpret $\bs{\theta}\in \Theta_S(\nu_1,\nu_2)$ similarly. In the subsequent subsections, we shall derive nonasymptotic rate of MEMR under both parameter spaces (\ref{DEF:Theta hyper}) and (\ref{DEF:Theta soblev}), in both scenarios of 
fully observed data and discretely observed data.



\subsection{Sharp nonasymptotic rate of MEMR under fully observed data}\label{sec:full:data:theory}
Suppose that the data functions $X_i^{(k)}(t)$, $i=1,\ldots,n_k,k=1,2$ are fully observed for arbitrary $t\in\mathcal{T}$. Throughout, let $n = n_1\wedge n_2$. 
\begin{theorem}\label{THM:full} 
For both $\Theta=\Theta_H(\nu_1,\nu_2)$ and $\Theta=\Theta_S(\nu_1,\nu_2)$, the following holds:
$$\inf_{\widehat{G}}\sup_{\bs{\theta} \in {\Theta}} E\left[ R_{\bs\theta}(\widehat{G}) - R_{\bs\theta}(G^{\ast}_{\bs\theta})\right]\asymp
\left(\frac{\log n}{n}\right)^{\frac{\nu_2}{1+\nu_2}}, $$
where the infimum is taken over all functional classifiers. 
\end{theorem}
Theorem \ref{THM:full} provides a sharp nonasymptotic rate for MEMR under both parameter spaces (\ref{DEF:Theta hyper}) and (\ref{DEF:Theta soblev}).
Interestingly, the rate relies on $\nu_2$ rather than $\nu_1$, implying that the smoothness of the population mean and covariance differences plays a more crucial role than the
smoothness of mean and covariance functions regarding the performance of the optimal functional classifier. Specifically, the sharp rate for MEMR becomes faster when $\nu_2$ increases. 
This intriguing phenomenon might be due to the fully observed data. In fact, 
as revealed in Section \ref{sec:partial:data:theory}, when data are discretely observed, this phenomenon may not hold. Moreover, the optimal rate appears to depend only on the smoothness, rather than the size, of the difference of the two populations. This means that the optimal rate does not change if the size of the population difference changes while its smoothness remains the same.

  \subsection{Sharp nonasymptotic rate of MEMR under discretely observed data}\label{sec:partial:data:theory}
  Suppose we observe ${X}_i^{(k)}(t_1),\ldots,{X}_i^{(k)}(t_M)$, $i=1,\ldots,n_k, k=1,2$ on evenly spaced $t_1,\ldots,t_M\in\mathcal{T}$;
  i.e., the data functions are observed over $M$ evenly spaced sampling points. For technical convenience, we make an additional assumption that $\psi_j$'s in (\ref{eqn:Omegak}) are Fourier basis of $L^2(\mathcal{T})$, i.e., $
  \psi_1(t)=1, \psi_{2j}(t)=\sqrt{2}\cos\left(2j\pi t\right), \psi_{2j+1}(t)=\sqrt{2}\sin\left(2j\pi t\right)$, for $j\geq 1$, $t\in\mathcal{T}$. 
\begin{theorem}\label{THM:sampling}
Let $\nu_1,\nu_2>0$ with $\nu_1\leq 1+\nu_2$.
For both $\Theta=\Theta_H(\nu_1,\nu_2)$ and $\Theta=\Theta_S(\nu_1,\nu_2)$, the following holds:
$$\inf_{\widehat{G}}\sup_{\bs{\theta} \in  {{\Theta}}} E\left[ R_{\bs\theta}(\widehat{G}) - R_{\bs\theta}(G^{\ast}_{\bs\theta})\right]\asymp\left(\frac{\log n}{n} + \frac{1}{M^{\nu_1}}\right)^{\frac{\nu_2}{1+\nu_2}}, $$
 {where the infimum is taken over all functional classifiers.} 
\end{theorem}
Theorem \ref{THM:sampling} reveals that $M^\ast=\left({n}/{\log n}\right)^{1/\nu_1} $ is a critical sampling frequency for the rate of MEMR over the parameter space ${\Theta}_H(\nu_1,\nu_2)$ and ${\Theta}_S(\nu_1,\nu_2)$.
When $M\geq M^\ast$, the MEMR is of rate $\left({\log n}/{n}\right)^{\nu_2/(1+\nu_2)}$ which is free of $M$ and is consistent
with the rate derived in Theorem \ref{THM:full}.
In other words, when $M\geq M^\ast$, the optimal classifier performs as well as
the one based on fully observed data.
When $M< M^\ast$, the MEMR is of rate $M^{-\nu_1 \nu_2/(1+\nu_2)}$ which solely relies on $M$.
Another interesting finding is that, when $M< M^\ast$, 
the rate of MEMR relies on
both $\nu_1$ and $\nu_2$, i.e., the smoothness of the mean and covariance functions as well as the separation between the populations.  
This differs from the estimation or testing problems in which the minimax optimal rate only relies on the smoothness of the mean function (see \cite{CY:11,Cai:Yuan:12:JASA, Hilgert:Mas:Verzelen:13,Shang:Cheng:15:AOS}). 

\section{Functional quadratic discriminant analysis} \label{SEC:opt:classifier}
In this section, we establish an optimal functional classifier  based on FQDA which requires accurately estimating the functional Bayes classifier through estimating the principle mean projection scores and principle eigenvalues. 
FQDA has been a major technique in functional classification literature
\cite{Galeano:etal:15,Dai:etal:17}. The basic idea is to 
first project the data functions onto 
an orthonormal basis and extract the principle projection scores, and then perform conventional QDA over the extracted scores.
FQDA performs well when data are Gaussian processes. Whereas, there is a lack of rigorous proof on the optimality of FQDA. We will construct FQDA classifier and prove its optimality in both fully observed data and discretely observed data.

\subsection{FQDA for fully observed data}\label{sec:full:obs}
Consider the ideal case that the data functions are fully observed as in Section \ref{sec:full:data:theory}. 
Write $X_i^{(k)}(t)=\sum_{j=1}^\infty{\xi}_{ij}^{(k)}\psi_j(t)$ for $i=1,\ldots,n_k$, $k=1,2$, where ${\xi}_{ij}^{(k)}$ are observed projection scores.
For $J\ge1$, let
\begin{equation}\label{FQDA:est}
\widehat{\bs{\mu}}_k=(\bar{\xi}_{\cdot 1}^{(k)},\ldots, \bar{\xi}_{\cdot J}^{(k)})^{\top},\,\,\,\,
\widehat{\mathbf{D}}=\widehat{\bs{\Sigma}}_2^{-1} - \widehat{\bs{\Sigma}}_1^{-1},\,\,\,\,
\widehat{\bs{\beta}}=\widehat{\bs{\Sigma}}_2^{-1} (\widehat{\bs\mu}_1 - \widehat{\bs\mu}_2),
\end{equation}
where $\bar{\xi}_{\cdot j}^{(k)}={n_k}^{-1}\sum_{i=1}^{n_k}{\xi}_{ij}^{(k)}$
is the estimation of mean projection score,
$\widehat\lambda_{j}^{(k)}=n_k^{-1}\sum_{i=1}^{n_k} (\xi_{ij}^{(k)}-\bar{\xi}_{\cdot j}^{(k)})^2$ is the estimation of eigenvalue, and $\widehat{\bs{\Sigma}}_k=\text{diag}\left(\widehat\lambda_{1}^{(k)},\ldots,\widehat\lambda_{J}^{(k)}\right)$ is the estimation of covariance operator.
The FQDA classifier is designed as follows:
\begin{eqnarray}\label{DEF:classifier}
 {\widehat{G}_J^{FQDA}(Z)}=
\left\{\begin{array}{cc}
1, & \widehat{Q}(\bs{z})\ge0,\\
2, & \widehat{Q}(\bs{z})<0,
\end{array}\right.
\end{eqnarray}
where 
\[
\widehat{Q}(\bs{z}):=(\bs{z}-\widehat{\bs{\mu}}_{1})^\top \widehat{\mathbf{D}}(\bs{z}-\widehat{\bs{\mu}}_{1}) - 2\widehat{\bs{\beta}}^\top (\bs{z}-\widehat{\bar{\bs{\mu}}}) - \log\left( |\widehat{\mathbf{D}}\widehat{\mathbf{\Sigma}}_1+\mathbf{I}_J|\right)+2\log\left(  {\widehat\pi_1}/{\widehat\pi_2}\right),
\]
$\bs{z}=(z_1,\ldots,z_J)^\top$ includes the first $J$ projection scores of $Z$ (see Section \ref{SEC:prelim}),
$\widehat{\bar{\bs{\mu}}}=(\widehat{\bs{\mu}}_1+\widehat{\bs{\mu}}_2)/2$ and $\widehat{\pi}_k=n_k/(n_1+n_2)$ is the sample proportion of class $k$.
Heuristically, when $J$ is suitably large, (\ref{DEF:classifier}) shall perform similar as the functional Bayes classifier (\ref{DEF:FBAYES}).

\begin{theorem}\label{THM:full:qda}
{
For both $\Theta=\Theta_H(\nu_1,\nu_2)$ and $\Theta=\Theta_S(\nu_1,\nu_2)$, the proposed FQDA classifier (\ref{DEF:classifier}) satisfies}
$$ \sup_{\bs{\theta} \in \Theta} E\left[ R_{\bs\theta}( {\widehat{G}_{J^\ast}^{FQDA})} - R_{\bs\theta}(G^{\ast}_{\bs\theta})\right]  \lesssim \left(\frac{\log n}{n}\right)^{\frac{\nu_2}{1+\nu_2}},$$
where $J^\ast\asymp (n/\log n)^{1/(1+\nu_2)}$.
\end{theorem}
Theorem \ref{THM:full:qda} provides an upper bound for the excess misclassification risk of (\ref{DEF:classifier}) with $J=J^\ast$. 
Since the upper bound matches Theorem \ref{THM:full},
we claim that FQDA attains minimax optimality if the leading $J^\ast$ basis functions are used for constructing the classifier. 

\subsection{FQDA for discretely observed data}\label{sec:part:obs}
Consider the more realistic case that the data functions are discretely observed as in Section \ref{sec:partial:data:theory}.
For $1\le J\le M$, define
\[
\mathbf{B}=\left(\begin{matrix}
  \psi_1(t_1) &  \psi_2(t_1) & \cdots & \psi_J(t_1) \\
  \psi_1(t_2) &  \psi_2(t_2) & \cdots & \psi_J(t_2) \\
  \vdots & \vdots & & \vdots \\
  \psi_1(t_M) &  \psi_2(t_M) & \cdots & \psi_J(t_M) \\
   \end{matrix}\right).
\]
Heuristically, when $J$ is suitably large, the data vector $\mathbf{X}_i^{(k)}=({X}_i^{(k)}(t_1),\ldots,{X}_i^{(k)}(t_M))^\top$
has an approximate expression
$\mathbf{X}_i^{(k)}\approx\mathbf{B}\bs{\mu}_k$ for $i=1,\ldots,n_k$,
where $\bs{\mu}_k=(\mu_{k1},\ldots,\mu_{kJ})^\top$ is the vector of $J$ principle mean projection scores.
When $\psi_j$'s are Fourier basis, it holds that $\mathbf{B}^\top\mathbf{B}=M\mathbf{I}_J$, which leads to
\begin{equation}\label{EQ:zeta}
    \bs{\mu}_k \approx \bs{\zeta}_i^{(k)}:= M^{-1}\mathbf{B}^\top\mathbf{X}_i^{(k)}.
\end{equation}
For $k=1,2$, let
\begin{equation}\label{sFQDA:est}
\widehat{\bs{\mu}}_{sk}=\frac{1}{n_k} \sum_{i=1}^{n_k} {\bs{\zeta}}_i^{(k)},\,\,
\widehat{\mathbf{D}}_s=  \widehat{\bs{\Sigma}}_{s2}^{-1} - \widehat{\bs{\Sigma}}_{s1}^{-1},\,\,
\widehat{\bs\beta}_s= \widehat{\mathbf{\Sigma}}_{s2}^{-1}\left(\widehat{\bs{\mu}}_{s2}-\widehat{\bs{\mu}}_{s1}\right),
\end{equation}
where $ \widehat{\mathbf{\Sigma}}_{sk}=\text{diag}\left(\widehat{\lambda}_{s1}^{(k)},\ldots, \widehat{\lambda}_{sJ}^{(k)}\right)$ with $\widehat{\lambda}_{sj} ^{(k)}=  n_k^{-1}\sum_{i=1}^{n_k} \left(\zeta_{ij}^{(k)} - \bar\zeta_{\cdot j}^{(k)} \right)^2$,
$\bar\zeta_{\cdot j}^{(k)}=n^{-1}_k \sum_{i=1}^{n_k}\zeta_{ij}^{(k)}$, and
$\zeta_{ij}^{(k)}$'s are components of ${\bs{\zeta}}_i^{(k)}$.
We then propose the following classification rule, called as sampling FQDA (sFQDA):
\begin{eqnarray}\label{DEF:classifierS}
\widehat{G}_J^{sFQDA}(Z)=\left\{
\begin{array}{cc}
1, & \widehat{Q}_s(\bs{z})\ge0,\\
2, & \widehat{Q}_s(\bs{z})<0,
\end{array} \right.
\end{eqnarray}
where
\[
\widehat{Q}_s(\bs{z}):=(\bs{z}-\widehat{\bs{\mu}}_{s1})^\top \widehat{\mathbf{D}}_s(\bs{z}-\widehat{\bs{\mu}}_{s1}) - 2\widehat{\bs{\beta}}_s^\top (\bs{z}-\widehat{\bar{\bs{\mu}}}_s) - \log\left( |\widehat{\mathbf{D}}_s\widehat{\mathbf{\Sigma}}_{s1}+\mathbf{I}_J|\right)+2\log\left(  {\widehat\pi_1}/{\widehat\pi_2}\right),
\]
with $\widehat{\bar{\bs{\mu}}}_s=(\widehat{\bs{\mu}}_{s1}+\widehat{\bs{\mu}}_{s2})/2$.

\begin{theorem}\label{THM:sampling:qda}
Let $\nu_1,\nu_2>0$ with $\nu_1\leq 1+\nu_2$.
For both $\Theta=\Theta_H(\nu_1,\nu_2)$ and $\Theta=\Theta_S(\nu_1,\nu_2)$, the sFQDA in (\ref{DEF:classifierS}) satisfies 
$$ \sup_{\bs{\theta} \in {\Theta}} E\left[ R_{\bs\theta}(\widehat{G}_{J^\ast}^{sFQDA}) - R_{\bs\theta}(G^{\ast}_{\bs\theta})\right]  \lesssim \left(\frac{\log n}{n} + \frac{1}{M^{\nu_1}}\right)^{\frac{\nu_2}{1+\nu_2}},$$
where $J^\ast \asymp M^{\nu_1/(1+\nu_2)} \mathbb{I}(M < M^\ast) + \left({n}/{\log n}\right)^{1/(1+\nu_2)}\mathbb{I}(M \geq M^\ast)$, $M^\ast = \left({n}/{\log n}\right)^{1/\nu_1}$ and $\mathbb{I}(\cdot)$ is the indicator function.
\end{theorem}
Theorem \ref{THM:sampling:qda} provides an upper bound for the excess misclassification risk of (\ref{DEF:classifierS}) with $J=J^\ast$, which matches Theorem \ref{THM:sampling}.
Therefore, we claim that sFQDA attains minimax optimality if the leading $J^\ast$ basis functions are used for constructing the classifier. 

Although FQDA is optimal in Gaussian setting, in general, it performs poorly when data are non-Gaussian.
 Hence, it is desirable to design a more accurate classifier in dealing with non-Gaussian data
 as well as preserving the same optimality in Gaussian case.
In the next section, we propose a novel approach based on DNN to address this issue.


\section{Functional deep neural network}\label{SEC:FDNN}

DNN has recently been used in various nonparametric regression and classification problems; see \cite{Schmidt:19, Bauer:Kohler:19, Kim:etal:21, Liu:19, Liu:etal:2021, wang2021stat, hu2020arxiv}.
As far as we know, the present section provides the first application of DNN in functional data classification. 
The basic idea is to train a DNN classifier using the observed principle projection scores.
Intuitively, when the network architectures are well selected, DNN should have a high expressive power so that the functional Bayes classifier can be well approximated even when its explicit form is not known.
Hence, FDNN is expected to be more resistant than FQDA to handle non-Gaussian data. We first define sparse DNN, and then construct FDNN classifiers in both fully observed data and discretely observed data, and prove their optimality. 
\subsection{Sparse deep neural network}
DNN tends to overfit the training data due to too much capacity of the network class. A common practice is to sparsify network parameters by methods such as dropout \cite{Goodfellow:etal:16}. Our approach is to train a functional classifier using sparse DNN which can effectively address the overfitting issue.

Let $\sigma$ denote the rectifier linear unit (ReLU) activation function,  i.e.,  $\sigma(x)=(x)_+$ for $x\in\mathbb{R}$.
For any real vectors $\bs{V}=(v_1,\ldots,v_r)^\top$ and $\bs{y}=(y_1,\ldots,y_r)^\top$, define the shift activation function $\sigma_{\bs{V}}(\bs{y})=(\sigma(y_1-v_1),\ldots,\sigma(y_r-v_r))^\top$.
For $L, J\ge1$ and $\bs{p}=(p_0,p_1,\ldots,p_{L}, p_{L+1})\in\mathbb{N}^{L+2}$, let $\mathcal{F}(L,J,\bs{p})$ denote the class of DNN over $J$ inputs, with $L$ hidden layers and $p_l$ nodes on hidden layer $l$, for $l=1,\ldots,L$. Let $p_0=J$ and $p_{L+1}=1$.
Any $f\in\mathcal{F}(L,J,\bs{p})$ has an expression
 \begin{equation} \label{EQ:f}
f(\mathbf{x}) = \mathbf{W}_L\sigma_{\bs{V}_L} \mathbf{W}_{L-1}\sigma_{\bs{V}_{L-1}}\ldots \mathbf{W}_1\sigma_{\bs{V}_1} \mathbf{W}_0\mathbf{x}, \,\,\,\,\mathbf{x}\in\mathbb{R}^J,
\end{equation}
where $\mathbf{W}_l\in\mathbb{R}^{p_{l+1}\times p_{l}}$, for $l=0,\ldots,L$, are weight matrices, $\bs{V}_l\in\mathbb{R}^{p_l}$, for $l=1,\ldots,L$,
are shift vectors. 
The sparse DNN class is defined as
\begin{eqnarray}\label{EQ:class}
&&\hspace{3mm} \mathcal{F}(L, J, \bs{p}, s, B)\\
 &=& \left\{ f\in \mathcal{F}(L, J, \bm{p}) :  \max_{l=0,\ldots,L}\| \mathbf{W}_l\|_{\infty} + \|\mathbf{v}_l\|_{\infty} \leq B, \sum_{l=0}^L\| \mathbf{W}_l\|_0 + \|\mathbf{v}_l\|_0 \leq s, \|f\|_\infty\le 1 \right\}, \nonumber
\end{eqnarray}
where $ \| \cdot\|_{\infty}$ denotes the maximum-entry norm of a matrix/vector
or supnorm of a function, $\| \cdot \|_0$  denotes the number of non-zero entries of a matrix or vector, $s>0$ controls the number of nonzero weights and shifts, $B>0$ controls the largest weights and shifts.
For notation convenience, we have assumed that the supnorm of $f$ has a unit upper bound, which can be replaced by arbitrary positive constant.

\subsection{FDNN classifier for fully observed data}\label{sec:full:dnn}
Let $\phi: \mathbb{R}\to[0,\infty)$ denote a surrogate loss such as the hinge loss $\phi(x)=(1-x)_+$. For $k=1,2$ and $i=1,\ldots,n_k$, recall 
$X_i^{(k)}(t)=\sum_{j=1}^\infty\xi_{ij}^{(k)}\psi_j(t)$ (see Section \ref{sec:full:obs}), and for
$J\ge1$, let $\bs\xi_i^{(k)} = (\xi_{i1}^{(k)}, \xi_{i2}^{(k)}, \ldots, \xi_{iJ}^{(k)})$ be the vector of $J$ principle projection scores corresponding to $X_i^{(k)}$.
Define the decision function
\begin{equation*}\label{hatf:FDNN}
\widehat{f}_{\phi}(\cdot) = \argmin_{f\in \mathcal{F}(L, J, \bs{p}, s, B)}\sum_{k=1}^2\sum_{i=1}^{n_k}\phi((2k-3)f(\bs\xi_i^{(k)})). 
\end{equation*}
Specifically, $\widehat{f}_{\phi}$ is the best network in $\mathcal{F}(L, J, \bs{p}, s, B)$ minimizing the empirical surrogate loss.
In practice, we suggest to use an R package ``\texttt{Keras}" to find $\widehat{f}_{\phi}$.

We then propose the following FDNN classifier:
\begin{eqnarray}\label{DEF:classifier:dnn}
 {\widehat{G}^{FDNN}(Z)}=
\left\{\begin{array}{cc}
1, & \widehat{f}_{\phi}(\bs z)\ge 0,\\
2, & \widehat{f}_{\phi}(\bs z)< 0.
\end{array}\right.
\end{eqnarray}

\begin{theorem}\label{THM:full:dnn}
Suppose the network class $\mathcal{F}(L, J, \bs{p}, s, B)$ satisfies 
\begin{description}
  \item (i)  $L\asymp \log n$;
  \item (ii) $J\asymp  n^{\frac{1}{1+\nu_2}}\left(\log n\right)^{\frac{-4}{1+\nu_2}}$;
  \item (iii) $ \max_{0\leq \ell \leq L} p_\ell \asymp n^{\frac{1}{1+\nu_2}}(\log n)^{\frac{\nu_2-3}{1+\nu_2}}$;
  \item (iv)  $s\asymp n^{\frac{1}{1+\nu_2}}(\log n)^{\frac{2\nu_2-2}{1+\nu_2}}$;
  \item (v)   $B\asymp n^{\frac{\nu_2}{2+2\nu_2}}(\log n)^{\frac{2-2\nu_2}{1+\nu_2}}$.
\end{description}
For both $\Theta=\Theta_H(\nu_1,\nu_2)$ and $\Theta=\Theta_S(\nu_1,\nu_2)$,  the FDNN classifier (\ref{DEF:classifier:dnn}) satisfies
$$ \sup_{\bs{\theta} \in \Theta} E\left[ R_{\bs\theta}( {\widehat{G} ^{FDNN})} - R_{\bs\theta}(G^{\ast}_{\bs\theta})\right]  \lesssim \left(\frac{\log^4 n}{n}\right)^{\frac{\nu_2}{1+\nu_2}}.$$
\end{theorem}
{Theorem \ref{THM:full:dnn} provides an upper bound for the excess misclassification risk of (\ref{DEF:classifier:dnn})}.
When the neural network architectures $(L,J, \bs{p}, s, B)$ are properly selected, the upper bound matches Theorem \ref{THM:full} up to a log factor.
Therefore, FDNN is proven minimax optimal. 

\subsection{FDNN classifier for discretely observed   data}\label{sec:part:dnn}
For $i=1,\ldots,n_k$, $k=1,2$, let $\bs\zeta_i^{(k)}$ be given in (\ref{EQ:zeta}).
Define the decision function
\begin{equation*}
\widehat{f}_{\phi}^{(s)}(\cdot) = \argmin_{f\in \mathcal{F}(L, J, \bs{p}, s, B)}\sum_{k=1}^2\sum_{i=1}^{n_k}\phi((2k-3)f(\bs\zeta_i^{(k)})).
\end{equation*}
We then propose the following  sampling FDNN (sFDNN) classifier:
\begin{eqnarray}\label{DEF:classifierS:dnn}
 {\widehat{G}^{sFDNN}(Z)}=
\left\{\begin{array}{cc}
1, & \widehat{f}_{\phi}^{(s)}(\bs z)\ge 0,\\
2, & \widehat{f}_{\phi}^{(s)}(\bs z)<0.
\end{array}\right.
\end{eqnarray}


\begin{theorem}\label{THM:sampling:dnn}
Suppose the network class $\mathcal{F}(L, J, \bs{p}, s, B)$ satisfies
\begin{description}
  \item (i) $L\asymp (\log M) \mathbb{I}(M \leq M^\ast)  +(\log n) \mathbb{I}(M \geq M^\ast) $; 
  \item (ii) $J\asymp M^{\frac{\nu_1}{1+\nu_2}} \mathbb{I}(M \leq M^\ast) + n^{\frac{1}{1+\nu_2}}\left(\log n\right)^{\frac{-4}{1+\nu_2}}\mathbb{I}(M \geq M^\ast)$;
    \item (iii) $ \max_{0\leq \ell \leq L} p_\ell \asymp  M^{\frac{\nu_1}{1+\nu_2}}(\log M) \mathbb{I}(M \leq M^\ast)  +n^{\frac{1}{1+\nu_2}}(\log n)^{\frac{\nu_2-3}{1+\nu_2}}\mathbb{I}(M \geq M^\ast) $;
    \item (iv) $s\asymp M^{\frac{\nu_1}{1+\nu_2}}(\log^2{M}) \mathbb{I}(M \leq M^\ast)  +n^{\frac{1}{1+\nu_2}}(\log n)^{\frac{2\nu_2-2}{1+\nu_2}}\mathbb{I}(M \geq M^\ast) $;
    \item (v) $B\asymp M^{\frac{\nu_1\nu_2}{2+2\nu_2}}\mathbb{I}(M \leq M^\ast)  +n^{\frac{\nu_2}{2+2\nu_2}}(\log n)^{\frac{2-2\nu_2}{1+\nu_2}}\mathbb{I}(M \geq M^\ast) $,
\end{description}
where $M^\ast = \left({n}/{\log^4 n}\right)^{1/\nu_1}$. 
Let $\nu_1,\nu_2>0$ with $\nu_1\leq 1+\nu_2$.
For both $\Theta=\Theta_H(\nu_1,\nu_2)$ and $\Theta=\Theta_S(\nu_1,\nu_2)$, the sFDNN classifier in (\ref{DEF:classifierS:dnn}) satisfies 
$$ \sup_{\bs{\theta} \in {\Theta}} E\left[ R_{\bs\theta}(\widehat{G}^{sFDNN}) - R_{\bs\theta}(G^{\ast}_{\bs\theta})\right]  \lesssim \left(\frac{\log^4 n}{n} + \frac{1}{M^{\nu_1}}\right)^{\frac{\nu_2}{1+\nu_2}}.$$
\end{theorem}
Theorem \ref{THM:sampling:dnn} provides an upper bound for the excess misclassification risk of (\ref{DEF:classifier:dnn}).
When the architectures $(L, J, \bs{p}, s, B)$ are properly selected, the upper bound matches the result in Theorem \ref{THM:sampling} up to a log factor.
Therefore, sFDNN is able to attain minimax optimality. The critical sampling  frequency $M^\ast = \left({n}/{\log^4 n}\right)^{1/\nu_1}$ differs from the one in Theorem \ref{THM:sampling} by a log factor as well.

\section{Simulation}\label{SEC:simulation}
Performances of FQDA and FDNN are examined through extensive simulations.
\subsection{Gaussian  setting}\label{SEC:Gaussian}
 We provide   numerical evidences to demonstrate the superior performance of FQDA and FDNN compared with two popular functional classifiers: quadratic discriminant method (QD)  proposed in \cite{Delaigle:Hall:13} and the nonparametric Bayes classifier (NB) proposed in \cite{Dai:etal:17}. 
 We will not compare with functional logistic regression since it performs poorly than NB when covariance differences in the populations are present \cite{Dai:etal:17}. 
 The difference between FQDA and QD is on how to estimate principle projection scores. Specifically, FQDA estimates projection scores by projecting the functional data onto Fourier basis, while QD applies functional principal component analysis to estimate principle projection scores in which eigenfunctions are data-driven. We evaluated all methods via four synthetic datasets. In  all simulations, we generated $n=n_1=n_2=50,100$ training samples for each class, which indicates $\pi_1=\pi_2=0.5$.  
We generated functional data $X_{i}^{(k)}(t)= \sum_{j=1}^{J} \xi_{ij}^{(k)} \psi_j(t)$ where 
$\xi_{ij}^{(k)}\sim N(\mu_{kj}, \lambda_j^{(k)})$, $i=1,\ldots,n_k$, $k=1,2$. In the following, $\bs\mu_k=(\mu_{k1},\ldots, \mu_{kJ})^\top$, $\bs\Sigma_k=\text{diag}(\lambda_1^{(k)},\ldots,\lambda_J^{(k)})$ and $\psi_j(t)$'s are specified in different models. 

\textit{Model 1}:  Let $J=3$,  $\bs\mu_1= \left(-1, 2, -3 \right)^\top$, $\bs\Sigma_1= \text{diag}\left( \frac{3}{5}, \frac{2}{5}, \frac{1}{5}\right) $,   $\bs\mu_2= \left(-\frac{1}{2}, \frac{5}{2}, -\frac{5}{2} \right)^\top$,  $\bs\Sigma_2=\text{diag}\left( \frac{9}{10}, \frac{1}{2}, \frac{3}{10}\right) $, $\psi_1(t)=\log(t+2)$, $\psi_2(t)=t$ and $\psi_3(t)=t^3$.

\textit{Model 2}:    Let $J=3$, $\bs\mu_1= \left(-6, 12, -18 \right)^\top$, $\bs\Sigma_1= \text{diag}\left( 3, 2, 1\right) $,   $\bs\mu_2= \left(-3, 9, -15 \right)^\top$,  $\bs\Sigma_2=\text{diag}\left( \frac{9}{2}, \frac{5}{2}, \frac{3}{2}\right) $, $\psi_1(t)=\log(t+2)$, $\psi_2(t)=t$ and $\psi_3(t)=t^3$.

\textit{Model 3}:   Let $J=4$,  $\bs\mu_1= \left(1, -1, 2, -3 \right)^\top$, $\bs\Sigma_1= \text{diag}\left( \frac{4}{5}, \frac{3}{5}, \frac{2}{5}, \frac{1}{5}\right) $,   $\bs\mu_2= \left(\frac{1}{2}, -\frac{1}{2}, \frac{5}{2}, -\frac{5}{2} \right)^\top$,  $\bs\Sigma_2=\text{diag}\left( 1, 1, \frac{1}{2}, \frac{3}{10}\right) $, $\psi_1(t)=\sin{2\pi t}$, $\psi_2(t)=\log(t+2)$, $\psi_3(t)=t$ and $\psi_4(t)=t^3$.

\textit{Model 4}:  Let $J=4$,  $\bs\mu_1= \left(6, -6, 12, -18 \right)^\top$, $\bs\Sigma_1= \text{diag}\left( 4, 3, 2, 1\right) $,   $\bs\mu_2= \left(3, -3, 9, -15\right)^\top$,  $\bs\Sigma_2=\text{diag}\left(5, 5, \frac{5}{2}, \frac{3}{2}\right) $, $\psi_1(t)=\sin{2\pi t}$, $\psi_2(t)=\log(t+2)$, $\psi_3(t)=t$ and $\psi_4(t)=t^3$.

In each above model, it is easy to see that the parameter $\bs{\theta}$
belongs to $\Theta_H(1,1)$ or $\Theta_S(1,1)$, i.e.,
$\nu_1=\nu_2=1$ in (\ref{DEF:Theta hyper}) or (\ref{DEF:Theta soblev}). The random functions were sampled at $M$ equally spaced sampling points from $0$ to $1$. We chose $M$ from $\left\{ 10, 20, 30, 40, 50\right\}$ to detect how   sampling frequency affects the classification error, where we regarded $M=50$ as the full observation.
In each scenario, the number of repetitions was set to be $100$, and the classification errors were evaluated with  $500$  samples.

As for tuning parameter selections: for FQDA, we selected $J$ through cross-validation proposed by \cite{Delaigle:Hall:12} and \cite{Delaigle:Hall:13};
for FDNN, we chose $L=\lceil \log M\rceil\vee\lceil \log n\rceil$, $J = c \lceil  M^{1/2}\rceil\vee  \lceil{n}^{1/2}\rceil$ for $1\le c\le 4$ depending on different settings,  $p_\ell = 20\lceil  M^{1/2} \rceil\vee \lceil n^{1/2}\rceil$, $B =5\lceil M^{1/4}\rceil\vee \lceil n^{1/4}\rceil$ and  $s=20 \lceil  M^{1/2} \rceil\vee \lceil n^{1/2}\rceil$. 
Note that the above selection of architecture parameters was based on 
Theorem \ref{THM:sampling:dnn}.


Tables \ref{TAB:S5}  to \ref{TAB:S8} summarize the misclassification rates for four classifiers given the combinations of  different mean and covariance models.  Given the  explicit definition of $X_i^{(k)}$, it is not surprising that the performances of FQDA and FDNN significantly dominate those of QD and NB   which  require the two series in (\ref{imperf:clas:cond}) are divergent. The discrepancy is increasing with the number of the number of observation per subject. Especially,    under the fully observed cases, the classification risks of our FQDA and FDNN classifiers are less than half of the risks generated by QD and NB. When the data are sparsely sampled ($M=10$), all classifiers have larger misclassification risks because of the less available information. But, the proposed FQDA and FDNN still outperform  two counterparts.

\subsection{Non-Gaussian setting}
To evaluate the performance of the proposed classifiers under non-Gaussian process situations,
we generated functional data $X_{i}^{(k)}(t)= \sum_{j=1}^{3} \xi_{ij}^{(k)} \psi_j(t)$ where $\xi_{ij}^{(1)}\sim N(\mu_{1j}, \lambda_j^{(2)})$, $i=1,\ldots,n_1$, and 
$\xi_{ij}^{(2)}\sim t_{7-2j}$, $i=1,\ldots,n_2$. We specify   $\bs\mu_1$, $\bs\Sigma_1$ and $\psi_1(t)$'s as follows.

\textit{Model 5}:  $\bs\mu_1= \left(-1, 2, -3 \right)^\top$, $\bs\Sigma_1= \text{diag}\left(3, 2, 1\right) $, $\psi_1(t)=\log(t+2)$, $\psi_2(t)=t$ and $\psi_3(t)=t^3$.

It is easy to see that $\bs{\theta}$ in Model 5 also belongs to $\Theta_H(1,1)$ or $\Theta_S(1,1)$.
The tuning parameter selections for FQDA and FDNN are the same as Section \ref{SEC:Gaussian}. 
Table \ref{TAB:S9} reports the misclassification rates for the four classifiers when the functional data of one of the classes are non-Gaussian. As he three competitors are designed only for the Gaussian process, FDNN dominates  in performance for both  sparsely and densely sampled functional data cases.  In most scenarios, the misclassification rates of FDNN are about only one third of those of QD and NB. FQDA incurred larger risks than FDNN in both cases, but  is still superior to QD and NB.

\begin{table}[tbh] \centering
  \caption{Misclassification rates ($\%$) with standard errors in brackets for Model 1}
  \label{TAB:S5}
\begin{tabular}{@{\extracolsep{0.1pt}} cccccc}
\hline
\hline
$M$ & $n$ & \multicolumn{1}{c}{FQDA} & \multicolumn{1}{c}{FDNN} & \multicolumn{1}{c}{QD} & \multicolumn{1}{c}{NB}\\
\hline
50 & 50 &18.75(0.02)  & 19.46(0.08)  & 39.15(0.02)  & 42.09(0.02) \\
& 100 &18.54(0.01)  & 16.86(0.09)  & 38.53(0.02)  & 40.96(0.02)  \\
\hline
40& 50&19.97(0.02) &19.91(0.08)   &39.12(0.02)  & 42.10(0.02)  \\
& 100 &19.85(0.02) &18.58(0.10)  & 38.49(0.02)  & 40.91(0.02) \\
\hline
30& 50 &22.17(0.02) &24.82(0.12)  & 39.14(0.02)  & 42.04(0.02) \\
& 100  &22.00(0.02) &18.70(0.10)  & 38.48(0.02)  & 40.87(0.02)  \\
\hline
20& 50&25.99(0.02) &26.04(0.12) & 39.00(0.02)  & 41.97(0.02)  \\
& 100  &26.04(0.02) & 24.27(0.01)  & 38.47(0.02)  & 40.75(0.05)   \\
\hline
10& 50  &32.10(0.02) &28.59(0.10) & 38.98(0.02)  & 41.79(0.02)   \\
& 100  &31.91(0.02)  &25.24(0.09) & 38.28(0.02)  & 40.70(0.02)    \\
\hline
\hline
\end{tabular}
\end{table}

\begin{table}[tbh] \centering
  \caption{Misclassification rates ($\%$) with standard errors in brackets for Model 2}
  \label{TAB:S6}
\begin{tabular}{@{\extracolsep{0.1pt}} cccccc}
\hline
\hline
$M$ & $n$ & \multicolumn{1}{c}{FQDA} & \multicolumn{1}{c}{FDNN} & \multicolumn{1}{c}{QD} & \multicolumn{1}{c}{NB}\\
\hline
50 & 50 &14.77(0.02)  & 18.82(0.10) & 37.91(0.02) & 41.03(0.02) \\
& 100 &14.58(0.01)  & 13.19(0.10)   & 37.35(0.02) & 39.92(0.02) \\
\hline
40& 50& 15.99(0.02) &18.52(0.10) & 37.85(0.02)  & 40.99(0.02) \\
& 100 &15.92(0.01)  &12.92(0.02)  & 37.32(0.02)  & 40.07(0.02) \\
\hline
30& 50 &18.29(0.02) &21.71(0.12) & 37.86(0.02) & 40.89(0.02)\\
& 100  &18.37(0.02) &12.95(0.09)  & 37.33(0.02)   &39.91(0.02)   \\
\hline
20& 50 &22.27(0.02) &24.01(0.14)&37.83(0.02)& 40.90(0.02)  \\
& 100  &22.39(0.02)  &21.70(0.11) & 37.28(0.02)& 39.81(0.02)  \\
\hline
10& 50  &29.12(0.02) &27.74(0.13)&37.66(0.02)& 40.72 (0.02)  \\
& 100  &29.16(0.02)  & 27.33(0.12)&37.18(0.02)& 39.57(0.02) \\
\hline
\hline
\end{tabular}
\end{table}

\begin{table}[tbh] \centering
  \caption{Misclassification rates ($\%$) with standard errors in brackets for Model 3}
  \label{TAB:S7}
\begin{tabular}{@{\extracolsep{0.1pt}} cccccc}
\hline
\hline
$M$ & $n$ & \multicolumn{1}{c}{FQDA} & \multicolumn{1}{c}{FDNN} & \multicolumn{1}{c}{QD} & \multicolumn{1}{c}{NB}\\
\hline
50 & 50 &18.63(0.02) & 20.02(0.04) &34.95(0.03)  & 40.26(0.03) \\
& 100 & 18.06(0.02)& 19.96(0.06) & 34.69(0.02) & 38.89(0.02)  \\
\hline
40& 50&19.85(0.02) &22.46(0.07) & 34.96(0.03) & 40.41(0.03)  \\
& 100 & 19.31(0.02)& 19.34(0.09) & 34.67(0.02) & 38.95(0.02) \\
\hline
30& 50 & 21.79(0.02)& 24.35(0.07) & 34.96(0.03) & 40.42(0.03) \\
& 100  & 21.33(0.02)& 20.05(0.08) & 34.70(0.02) & 39.05(0.02) \\
\hline
20& 50& 25.36(0.02) & 26.07(0.09)& 34.92(0.03) & 40.42(0.03)  \\
& 100  & 24.16(0.02) & 21.22(0.08)& 34.60(0.02) & 38.98(0.03)  \\
\hline
10& 50  & 30.25(0.02) & 26.03(0.08) & 34.72(0.03) & 40.35(0.03) \\
& 100  & 30.00(0.02) & 24.13(0.08) & 34.15(0.03) & 38.83(0.09)  \\
\hline
\hline
\end{tabular}
\end{table}

\begin{table}[tbh] \centering
  \caption{Misclassification rates ($\%$) with standard errors in brackets for Model 4}
  \label{TAB:S8}
\begin{tabular}{@{\extracolsep{0.1pt}} cccccc}
\hline
\hline
$M$ & $n$ & \multicolumn{1}{c}{FQDA} & \multicolumn{1}{c}{FDNN} & \multicolumn{1}{c}{QD} & \multicolumn{1}{c}{NB}\\
\hline
50 & 50 &14.56(0.02) &21.16(0.10) & 32.76(0.02) & 38.76(0.03)\\
& 100 & 14.26(0.02) &16.85(0.10)   & 32.64(0.02) & 36.77(0.03)  \\
\hline
40& 50& 15.89(0.02) & 20.42(0.10) & 32.78(0.02)  & 38.65(0.03) \\
& 100 & 19.31(0.02)& 20.18(0.09)  & 34.67(0.02)  & 38.95(0.02) \\
\hline
30& 50 & 18.26(0.02) &  22.75(0.10) & 32.72(0.02) & 38.58(0.03) \\
& 100  & 17.81(0.02) & 16.29(0.10)  & 32.60(0.02)  &36.74(0.03) \\
\hline
20& 50& 21.93(0.02) & 22.76(0.11)& 32.72(0.03)& 38.36(0.03) \\
& 100  & 21.54(0.02) & 21.29(0.09) & 32.59(0.02)& 36.88(0.03) \\
\hline
10& 50  & 27.46(0.02) &27.73(0.10)&32.52(0.02)& 38.78(0.03)\\
& 100  & 27.08(0.02) & 24.85(0.10)&32.34(0.02)& 37.00(0.02)  \\
\hline
\hline
\end{tabular}
\end{table}

\begin{table}[tbh] \centering
  \caption{Misclassification rates ($\%$) with standard errors in brackets for Model 5}
  \label{TAB:S9}
\begin{tabular}{@{\extracolsep{0.1pt}} cccccc}
\hline
\hline
$M$ & $n$ & \multicolumn{1}{c}{FQDA} & \multicolumn{1}{c}{FDNN} & \multicolumn{1}{c}{QD} & \multicolumn{1}{c}{NB}\\
\hline
50 & 50 &18.11(0.04) &13.20(0.01) & 42.63(0.02) & 40.27(0.03)\\
& 100 & 17.11(0.04) &12.29(0.01)   & 38.42(0.09) & 39.84(0.04)  \\
\hline
40& 50& 19.47(0.04) & 13.40(0.02) & 42.61(0.10)  & 40.38(0.04) \\
& 100 & 18.62(0.04) &12.35(0.01)  & 38.38(0.09)  & 39.79(0.04) \\
\hline
30& 50 & 22.14(0.05) &  12.89(0.01) & 42.73(0.01) & 40.50(0.03) \\
& 100  & 24.19(0.05) & 12.21(0.01)  & 38.30(0.09)   &40.11(0.04) \\
\hline
20& 50& 27.00(0.08) & 13.00(0.01)& 42.77(0.10)& 40.69(0.04) \\
& 100  & 22.75(0.07) & 12.21(0.01) & 38.17(0.09)& 40.26(0.04) \\
\hline
10& 50  & 36.75(0.08) & 23.01(0.16)& 43.16(0.04) & 41.38(0.04)\\
& 100  & 32.14(0.09) & 19.52(0.15)& 37.87(0.09)& 40.90(0.04)  \\
\hline
\hline
\end{tabular}
\end{table}
\section{Real Data Illustrations}\label{SEC:realdata}

This benchmark data example  was extracted from the TIMIT database (TIMIT Acoustic-Phonetic
Continuous Speech Corpus, NTIS, US Dept of Commerce), which is a widely used
resource for research in speech recognition and functional data classification \cite{Ferraty:Vieu:03}. The data set we used was
constructed by selecting five phonemes for classification based on digitized
speech from this database.  From each speech
frame, a log-periodogram transformation is applied so as to cast the speech data
in a form suitable for speech recognition. The five phonemes in this data
set are transcribed as follows: ``sh'' as in ``she'', ``dcl'' as in
``dark'', ``iy'' as the vowel in ``she'', ``aa'' as the vowel in ``dark'',
and ``ao'' as the first vowel in ``water''. For illustration purpose, we
focus on the ``aa'', ``ao'', ``iy'' and ``dcl''  phoneme classes. Each speech frame is represented by $n=400$ samples at a 16-kHz sampling rate;   the first $M=150$ frequencies from each subject are retained.  Figure \ref{FIG:data} displays 10 log-periodograms for each class phoneme.

We randomly select training sample size $n_1=n_2=300$ to train the classifiers of three methods and the rest of $100$
 samples remained as the test samples.
The tuning parameter selections for FQDA and FDNN are the same as Section \ref{SEC:Gaussian}. Table \ref{TAB:Speech}  reports the mean percentage
(averaged over the $100$ repetitions) of misclassified test curves. It can be seen that Both FQDA and FDNN   outperformed QD and NB   in all the three classification tasks. For ``ao'' vs``iy'', the misclassification rates of FQDA and FDNN are less than one third of that of QD; For ``ao'' vs``dcl'', the misclassification rates of FQDA and FDNN are around half of that of NB. The most difficult  task is to distinguish ``aa" and ``ao" and all the three   classifiers have much larger risks. However, the proposed classifiers FQDA and FDNN  still provide smaller risks and smaller standard errors compared with QD and NB classifiers. 

\begin{figure}
\begin{center}
\includegraphics[width=6.5cm, height=5cm]{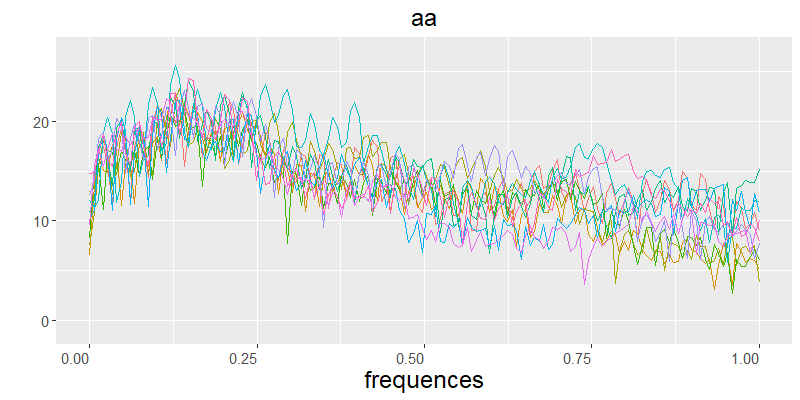}
\includegraphics[width=6.5cm, height=5cm]{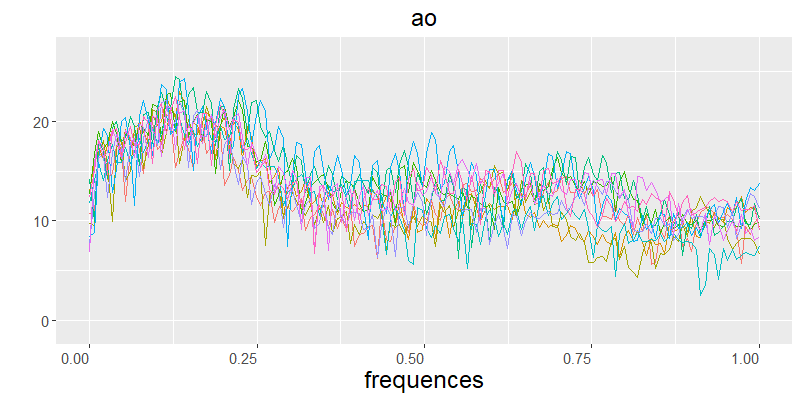}
\vskip 0.5cm
\includegraphics[width=6.5cm, height=5cm]{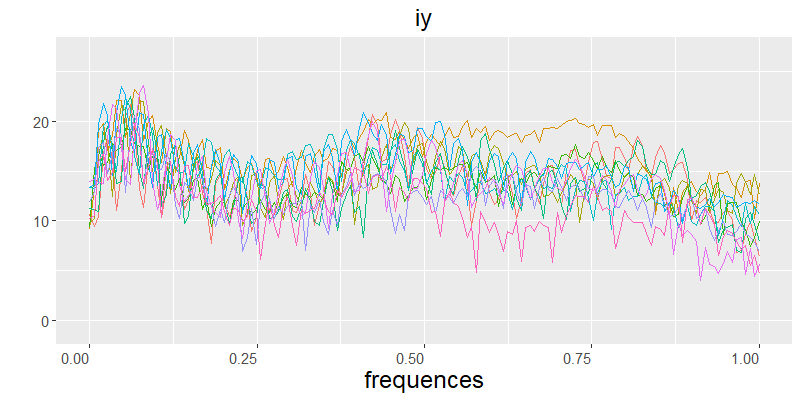}
\includegraphics[width=6.5cm, height=5cm]{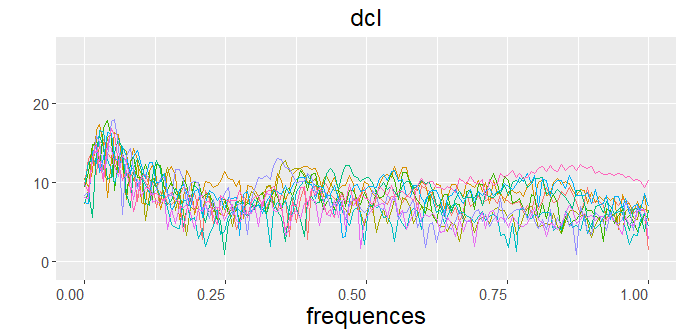}
\end{center}
\caption{A sample of 10 log-periodograms per class }
\label{FIG:data}
\end{figure}



\begin{table} \centering
  \caption{Misclassification rates ($\%$) with standard errors in brackets for Speech Recognition data.}
  \label{TAB:Speech}
\begin{tabular}{@{\extracolsep{0.1pt}} crrrrr}
\hline
\hline
Classes &  \multicolumn{1}{c}{FQDA} & \multicolumn{1}{c}{FDNN} & \multicolumn{1}{c}{QD} & \multicolumn{1}{c}{NB}\\
\hline
``aa'' vs ``ao''  & 20.278(0.014)& 20.744(0.016)& 25.402(0.026) & 25.378(0.021)    \\ \hline
``aa'' vs ``iy''  & 0.196(0.001)& 0.193(0.002) & 0.288(0.005) &  0.273(0.006)   \\ \hline
``ao'' vs ``iy''  &   0.153(0.004)& 0.183(0.004) & 0.578(0.005) &  0.232(0.005)   \\ \hline
``ao'' vs ``dcl'' &   0.270(0.003)& 0.229(0.002) & 0.391(0.005) &  0.472(0.006)   \\ \hline
\end{tabular}
\end{table}





\section{Conclusion}\label{SEC:discussion}
 We present a new minimax optimality viewpoint for   solving functional classification problems. In comparison with the existing literature,
 our results are able to deal with the more practical scenarios where the two populations are relatively ``close'' so that the optimal Bayes risk is asymptotically non-vanishing.
Our contributions are threefold. First,
we provide sharp convergence rates for MEMR when data are either fully or discretely observed,
 as well as a critical sampling frequency that governs the rate in the latter case.
Second, we  propose novel classifiers based on FQDA and FDNN which are proven to achieve minimax optimality. Third, we demonstrate via extensive simulations and real-data examples that  the proposed FDNN classifier has outstanding performances even when the Gaussian assumption is invalid.


\section*{Acknowledgement}
Wang's and  Cao's  research was partially supported by  NSF award DMS 1736470.
Shang's research was supported in part by NSF DMS 1764280 and 1821157.

\section*{Appendix}
  We introduce additional notation and definitions
that will be used throughout the rest of the paper.
 $\mathbf{1}_p$ is a p-dimensional
vector with elements being $1$. For a vector $\mathbf{u}$, $\|\mathbf{u} \|_2$, $\|\mathbf{u} \|_{\infty}$ denote the $L_2$
norm and $L_{\infty}$ norm respectively.  For
a matrix $\mathbf{A} \in \mathbb{R}_{p\times p}$,  $\|\mathbf{A} \|_2$,  $\|\mathbf{A} \|_F$  denote the   spectral  norm,
 Frobenius norm  respectively.
 
 We only give the proofs over parameter space $\Theta_S(\nu_1,\nu_2)$. The proofs over ${\Theta}_H(\nu_1,\nu_2)$ are the same, since they all imply the same upper bound of sequences with the same order and radius. Specifically, given $\omega$, $A>0$, the tail sum of $\bs{a}\in H^\omega(A)$  has an upper bound
\[
\sum_{k=L}^{\infty}|a_k|\le  A\sum_{k=L}^{\infty}k^{-(1+\omega)}\asymp AL^{-\omega},
\]
which is asymptotically equivalent to $\bs{a}\in S^\omega(A)$.

Throughout, let $\Theta\equiv \Theta_S(\nu_1,\nu_2)$, where $S=S^\omega(A)$ is the Sobolev ball of order $\omega$ and radius $B$. 
We also give some notation regarding the tail series. For $M\ge0$, define
\[
f_1(M; \Theta)=\sup_{\bs{\theta}\in\Theta}\left[\sum_{j=M+1}^\infty (\mu_{1j}^2\vee \mu_{2j}^2)\right]^{1/2}\,\,\,\,\text{and}\,\,\,\,
f_2(M; \Theta)=\sup_{\bs{\theta}\in\Theta}\sum_{j=M+1}^\infty (\lambda_j^{(1)}\vee \lambda_j^{(2)}).
\]
Without loss of generality, assume $f_1(M; \Theta) \leq 1$ and $f_2(M; \Theta) \leq 1$ for any $M\ge0$; otherwise one can scale them by $f_1(0; \Theta)$ and $f_2(0; \Theta)$, respectively. Both $f_1(M; \Theta)$ and $f_2(M; \Theta)$ are  decreasing in $M$ which depict the decay rate of $ \mu_{kj}$ and $\lambda_{j}^{(k)}$. Let $f(M; \Theta)=f_1(M; \Theta)^2 + f_2(M; \Theta)^2$.  

For $J\ge0$, define
$${g(J; \Theta) = \sup_{\bs{\theta} \in \Theta}\left\{\sum_{j=J+1}^\infty\frac{\left(\mu_{1j}-\mu_{2j} \right)^2}{\lambda_{j}^{(2)}} + \sum_{j=J+1}^\infty\left(\lambda_{j}^{(1)}/\lambda_{j}^{(2)}-1\right)^2 \right\}. }$$
Without {loss} of generality, assume $g(J; \Theta) \leq 1$ for any $J\ge0$;
otherwise one can scale $g(J; \Theta)$ by $g(0; \Theta)$. Note that $f(\cdot; \Theta)$, $g(\cdot; \Theta)$ are both monotone decreasing functions. 
\subsection{Technical lemmas}
We begin by collecting a few important technical lemmas that will be
used in the proofs of the minimax lower bounds. Define an alternative risk function $L_{\bs\theta}(\widehat{G})$ as follows,
\begin{equation*}\label{EQ:LG}
  L_{\bs\theta}(\widehat{G})= P\left( \widehat{G}(Z)\neq G^{\ast}_{\bs\theta}(Z) \right).
\end{equation*}
This loss function $L_{\bs\theta}(\widehat{G})$ is essentially the probability that $\widehat{G}$ produces
a different label than $G^{\ast}_{\bs\theta}$, and satisfies the triangle inequality. The connection between $R_{\bs\theta}(\widehat{G}) - R_{\bs\theta}(G^{\ast}_{\bs\theta})$ and $L_{\bs\theta}(\widehat{G})$ is presented
by the following lemma, which shows that it’s sufficient to provide a lower
bound for $L_{\bs\theta}(\widehat{G})$ to prove Theorem \ref{THM:full}.

\begin{lemma} (\cite{Azizyan:etal:13}) \label{LEM:LL}
For any $\bs\theta$, $\tilde{\bs\theta} \in \Theta$ and any classification rule $\widehat{G}$, recall that $G^{\ast}_{\tilde{\bs\theta}}$ is the optimal rule w.r.t. $\tilde{\bs\theta}$. If
\begin{equation*}
  L_{\bs{\theta}}(G^{\ast}_{\tilde{\bs\theta}})+ L_{\bs{\theta}}(\widehat{G})+\sqrt{\frac{KL(P_{\bs\theta},P_{\tilde{\bs\theta}} )}{2}} \leq 1/2,
\end{equation*}
then
\begin{equation*}
  L_{\bs{\theta}}(G^{\ast}_{\tilde{\bs\theta}})- L_{\bs{\theta}}(\widehat{G})-\sqrt{\frac{KL(P_{\bs\theta},P_{\tilde{\bs\theta}} )}{2}} \leq  L_{\tilde{\bs{\theta}}}(\widehat{G}) \leq   L_{\bs{\theta}}(G^{\ast}_{\tilde{\bs\theta}})+ L_{\bs{\theta}}(\widehat{G})+\sqrt{\frac{KL(P_{\bs\theta},P_{\tilde{\bs\theta}} )}{2}},
\end{equation*}
where the Kullback–Leibler (KL) divergence of two probability density functions $P_{\bs\theta_1}$ and $P_{\tilde{\bs\theta}_2}$
is defined by
$$
KL(P_{\bs\theta_1},P_{\tilde{\bs\theta}_2} )=\int P_{\bs\theta_1}(z)\log\frac{P_{\bs\theta_1}(z)}{P_{\tilde{\bs\theta}_2}(z)}dz.
$$
\end{lemma}
The following lemma gives a  Fano’s type minimax lower bound.
\begin{lemma}(Fano's Lemma in \cite{Tsybakov:09}) \label{LEM:Etheta}
Let $N\geq 0$ and  $\bs\theta_1,\ldots, \bs\theta_N$, $\tilde{\bs\theta} \in \Theta$. For some constants $ \varrho \in (0,1/8)$, $c>0$, and any classification rule $\widehat{G}$, if $KL(P_{\bs\theta_i},P_{\tilde{\bs\theta}} )\leq \varrho \log N/n$ for all $1\leq i \leq N$, and $L_{\bs\theta_i} (\widehat{G}) < c$ implies  $L_{\bs\theta_j} (\widehat{G}) \geq c$ for all $0\leq i\neq j\leq N$, then $\inf_{\widehat{G}}\sup_{i=1,\ldots,N} E_{\bs{\theta}_i}[L_{\bs{\theta}_i}(\widehat{G})] \gtrsim c$.
\end{lemma}
 We need a covering number argument, which is provided by the following lemma.
\begin{lemma}(\cite{Tsybakov:09}) \label{LEM:A}
 Define $\mathcal{A}_{J, J^{\ast}} = \left\{\bs{u}: \bs{u} \in \left\{ 0, 1\right\}^J, \| \bs{u}\|_0=J^\ast \right\}$, where $\| \cdot \|_0$  denotes the number of non-zero entries. If $J > 4J^{\ast}$, then there exists a subset $\{\bs{u}_0,\bs{u}_1, \ldots, \bs{u}_N \} \subset \mathcal{A}_{J,J^{\ast}}$, such that $\bs{u}_0=(0,\ldots,0)^{\top}$, $\rho_H (\bs{u}_i,\bs{u}_j)\geq J^{\ast}/2$ and $\log(N+1) \geq \frac{J^{\ast}}{5} \log(\frac{J}{J^{\ast}})$, where $\rho_H$ is the Hamming distance.
\end{lemma}

\begin{lemma}(Lemma 4.1 in \cite{Cai:Zhang:19b})\label{LEM:Ltheta}
Suppose $\bs\theta \in \Theta$. There exists a constant $c>0$, which doesn't depend on $n$, such that for  any classification rule $G$, if $L_{\bs\theta}(G)<c$, then $L^2_{\bs\theta}(G)\lesssim P_{\bs\theta}(G(Z)\neq Y(Z))-P_{\bs\theta}(G_{\bs\theta}^\ast(Z)\neq Y(Z))$, where $G_{\bs\theta}^\ast$ is the optimal rule.
\end{lemma}
Based on Lemma \ref{LEM:Ltheta}, we use Fano's inequality on a carefully designed
least favorable multivariate normal distributions to complete the proof of Theorem \ref{THM:full}.

{

Next, we collect lemmas that will be used in the proofs of the minimax upper bounds. 

The following two lemmas demonstrate the existence and uniqueness of $J^\ast$ and $M^\ast$ in both fully observed case and discretely observed case. Note that it's easy to see that $f(M; \Theta) \asymp f^\ast(M;\Theta):= M^{-\nu_1} $ and $  g(J; \Theta) \asymp g^\ast(J;\Theta):= J^{-\nu_2}. $

\begin{lemma}\label{LEM:self-similar full}
If  $ g(J; \Theta) \asymp  J^{-\nu_2}$, then for any $n\in \mathbb{N}^+$, there exists a unique $J^\ast$ such that $J^\ast \log n/n = g^\ast(J^\ast; \Theta)$. 
\end{lemma}
\begin{proof}
Solution to $J^\ast \log n/n = g^\ast(J^\ast; \Theta)$ over all $J\in \mathbb{N}^+$ is $J^\ast = (\log n/n)^{-1/(\nu_2+1)}$. The uniqueness is satisfied by the monotonicity of  $g^\ast(J^\ast; \Theta)$.
\end{proof}

\begin{lemma}\label{LEM:self-similar}
If  $f(M; \Theta) \asymp M^{-\nu_1} $ and $ g(J; \Theta) \asymp  J^{-\nu_2}$, such that $\nu_2+1 \geq \nu_1$, then for any $n\in \mathbb{N}^+$, there exists an unique $M^\ast\in \mathbb{N}^+$ such that $\log n/n = f^\ast(M^\ast; \Theta)$. Furthermore, for any $M\leq M^\ast$, there exists an unique positive integer $J_1^\ast \leq M$ such that $J_1^\ast f^\ast(M) = g^\ast(J_1^\ast; \Theta)$; For any $M \geq M^\ast$, there exists an unique positive integer $J_2^\ast \leq M$ such that $J_2^\ast \log n/n = g^\ast(J_2^\ast; \Theta)$. 
\end{lemma}

\begin{proof}
First of all, we have $f^\ast(M; \Theta) = M^{-\nu_1}$ and $g^\ast(J; \Theta)=J^{-\nu_2}$.

In the following, we can justify our statement. 
Set the equation $f^\ast(M; \Theta) = \log n/n $, by the monotonicity of $f^\ast(M; \Theta)$, there is an unique $M^\ast=(\log n/n)^{-1/\nu_1}$. 

When $M < M^\ast$, set the equation $M^{-\nu_1} = J^{-\nu_2-1}$, then there exists an unique $J_1^\ast=M^{\nu_1/(\nu_2+1)}$, since $\nu_1 \leq \nu_2+1$, we have $J_1^\ast \leq M$. 

When $M \geq M^\ast$, set the equation $\log n/n = J^{-\nu_2-1}$, then there exists a unique $J_2^\ast=(\log n/n)^{-1/(\nu_2+1)}$. Since $M^{-\nu_1} \leq \log n/n$, we have $J_2^\ast \leq M^{\nu_1/(\nu_2+1)}$, which also implies $J_2^\ast \leq M$. 
\end{proof}

}

The following two lemmas show the consistency of the differential direction and graphical direction.
With a slight abuse of notation, let ${\mathbf{\Sigma}_k}=\text{diag}\left( {\lambda}_1^{(k)}, \ldots, {\lambda}_J^{(k)}\right)$, $\mathbf{D} = \Sigma_2^{-1}-\Sigma_1^{-1}$.
\begin{lemma}\label{LEM:Dandbeta}
The proposed estimators  in  equation (\ref{FQDA:est}) satisfy that, with probability
at least $1 - 3/n$, $\| \widehat{\mathbf{D}} - \mathbf{D}\|_F \lesssim  \sqrt{\frac{J\log n}{n}} (\| \mathbf{D}\bs{\Sigma}_1\|_F + \| \mathbf{D}\bs{\Sigma}_2\|_F )$, $\| \widehat{\bs{\beta}} - \bs{\beta}\|_2 \lesssim  \sqrt{\frac{J\log n}{n}} \| \bs{\beta}\|_2$.
\end{lemma}

\begin{proof} Note that  matrices $\bs\Sigma_1$, $\bs\Sigma_2$, $\mathbf{D}_1$ and  $\mathbf{D}_2$ are diagonal matrices. Hence,  $\| (\widehat{\mathbf{D}} - \mathbf{D})\left(\bs{\Sigma}_1\bs{\Sigma}_2 \right)(\widehat{\mathbf{D}} - \mathbf{D})\|_F \asymp \| \widehat{\mathbf{D}} - \mathbf{D}\|_F^2 $. Since we have the following decomposition
\begin{eqnarray*}
&&\bs{\Sigma}_1\bs{\Sigma}_2 (\widehat{\mathbf{D}} - \mathbf{D}) \\
&=& \bs{\Sigma}_1\bs{\Sigma}_2 \widehat{\mathbf{D}} - (\bs{\Sigma}_1 - \bs{\Sigma}_2) \\
&=& \bs{\Sigma}_1\bs{\Sigma}_2\widehat{\mathbf{D}} -  \widehat{\bs{\Sigma}}_1\widehat{\bs{\Sigma}}_2\widehat{\mathbf{D}}+ \widehat{\bs{\Sigma}}_1\widehat{\bs{\Sigma}}_2\widehat{\mathbf{D}} -  (\bs{\Sigma}_1 - \bs{\Sigma}_2) \\
&=& (\bs{\Sigma}_1\bs{\Sigma}_2-  \widehat{\bs{\Sigma}}_1\widehat{\bs{\Sigma}}_2)\widehat{\mathbf{D}}+(\widehat{\bs{\Sigma}}_1 - {\bs{\Sigma}}_1) + (\widehat{\bs{\Sigma}}_2 - {\bs{\Sigma}}_2) \\
&=& (\bs{\Sigma}_1\bs{\Sigma}_2- {\bs{\Sigma}}_1\widehat{\bs{\Sigma}}_2 + {\bs{\Sigma}}_1\widehat{\bs{\Sigma}}_2 - \widehat{\bs{\Sigma}}_1\widehat{\bs{\Sigma}}_2)\widehat{\mathbf{D}}+(\widehat{\bs{\Sigma}}_1 - {\bs{\Sigma}}_1) + (\widehat{\bs{\Sigma}}_2 - {\bs{\Sigma}}_2) \\
&=& \left\{\bs{\Sigma}_1(\bs{\Sigma}_2 -\widehat{\bs{\Sigma}}_2) + \widehat{\bs{\Sigma}}_2(\bs{\Sigma}_1 -\widehat{\bs{\Sigma}}_1)\right\}\widehat{\mathbf{D}}+(\widehat{\bs{\Sigma}}_1 - {\bs{\Sigma}}_1) + (\widehat{\bs{\Sigma}}_2 - {\bs{\Sigma}}_2) \\
&=& \left\{\bs{\Sigma}_1(\bs{\Sigma}_2 -\widehat{\bs{\Sigma}}_2) + \bs{\Sigma}_2(\bs{\Sigma}_1 -\widehat{\bs{\Sigma}}_1) + (\widehat{\bs{\Sigma}}_2-\bs{\Sigma}_2 )(\bs{\Sigma}_1 -\widehat{\bs{\Sigma}}_1)\right\}\widehat{\mathbf{D}}+(\widehat{\bs{\Sigma}}_1 - {\bs{\Sigma}}_1) + (\widehat{\bs{\Sigma}}_2 - {\bs{\Sigma}}_2), \\
\end{eqnarray*}
thus we have
\begin{eqnarray*}
&& \|(\widehat{\mathbf{D}} - \mathbf{D})^{\top}\left(\bs{\Sigma}_1\bs{\Sigma}_2 \right)(\widehat{\mathbf{D}} - \mathbf{D}) \|_F \\
& =&\left \|(\widehat{\mathbf{D}} - \mathbf{D})^{\top}\left\{ (\widehat{\bs{\Sigma}}_1 - \widehat{\bs{\Sigma}}_2) - (\bs{\Sigma}_1 - \bs{\Sigma}_2) + (\bs{\Sigma}_1-\widehat{\bs{\Sigma}}_1)(\widehat{\mathbf{D}} - \mathbf{D})\widehat{\bs{\Sigma}}_2 \right. \right.\\
&&+ \left.\left. (\bs{\Sigma}_2-\widehat{\bs{\Sigma}}_2)(\widehat{\mathbf{D}} - \mathbf{D})\bs{\Sigma}_1  + (\bs{\Sigma}_1-\widehat{\bs{\Sigma}}_1)\mathbf{D}\widehat{\bs{\Sigma}}_2 +  (\bs{\Sigma}_2-\widehat{\bs{\Sigma}}_2)\mathbf{D}\bs{\Sigma}_1 \right\} \right \|_F\\
& =& \left \|(\widehat{\mathbf{D}} - \mathbf{D})^{\top}\left\{   (\widehat{\bs{\Sigma}}_1 - \widehat{\bs{\Sigma}}_2) - (\bs{\Sigma}_1 - \bs{\Sigma}_2) + (\bs{\Sigma}_1-\widehat{\bs{\Sigma}}_1)(\widehat{\mathbf{D}} - \mathbf{D})(\widehat{\bs{\Sigma}}_2 - \bs{\Sigma}_2 ) \right. \right.\\
&&+  (\bs{\Sigma}_2-\widehat{\bs{\Sigma}}_2)(\widehat{\mathbf{D}} - \mathbf{D})\bs{\Sigma}_1  + (\bs{\Sigma}_1-\widehat{\bs{\Sigma}}_1)\mathbf{D}(\widehat{\bs{\Sigma}}_2 - \bs{\Sigma}_2 )  +  (\bs{\Sigma}_2-\widehat{\bs{\Sigma}}_2)\mathbf{D}\bs{\Sigma}_1 \\
&&+ \left.\left. (\bs{\Sigma}_1-\widehat{\bs{\Sigma}}_1)(\widehat{\mathbf{D}} - \mathbf{D})\bs{\Sigma}_2  +  (\bs{\Sigma}_1-\widehat{\bs{\Sigma}}_1)\mathbf{D}\bs{\Sigma}_2 \right\} \right \|_F\\
& \leq & 2\sqrt{\frac{J\log n}{n}}\| \widehat{\mathbf{D}} - \mathbf{D}\|_F +
\frac{J\log n}{n}\| \widehat{\mathbf{D}} - \mathbf{D}\|_F^2 +  2\sqrt{\frac{J\log n}{n}}\| \widehat{\mathbf{D}} - \mathbf{D}\|_F\\
&& +  \sqrt{\frac{J\log n}{n}}\| (\mathbf{D}\bs{\Sigma}_1\|_F + \mathbf{D}\bs{\Sigma}_2\|_F )\| \widehat{\mathbf{D}} - \mathbf{D}\|_F + \frac{J\log n}{n}\| \mathbf{D}\bs{\Sigma}_1\|_F \| \widehat{\mathbf{D}} - \mathbf{D}\|_F\\
& \lesssim & \sqrt{\frac{J\log n}{n}}\| \widehat{\mathbf{D}} - \mathbf{D}\|_F^2 + \sqrt{\frac{J\log n}{n}}\| \widehat{\mathbf{D}} - \mathbf{D}\|_F  +
  \sqrt{\frac{J\log n}{n}}(\| \mathbf{D}\bs{\Sigma}_1\|_F + \|\mathbf{D}\bs{\Sigma}_2\|_F)\| \widehat{\mathbf{D}} - \mathbf{D}\|_F,
\end{eqnarray*}
where the first inequality is derived by Lemma 8.5 in \cite{Cai:Zhang:19b}.
In the above inequality, we divide $\| \widehat{\mathbf{D}} - \mathbf{D}\|_F$ on both sides, and we have
$\| \widehat{\mathbf{D}} - \mathbf{D}\|_F \lesssim  \sqrt{\frac{J\log n}{n}} (\| \mathbf{D}\bs{\Sigma}_1\|_F + \| \mathbf{D}\bs{\Sigma}_2\|_F )$.

Similarly, we have $| (\widehat{\bs{\beta}} - \bs{\beta})^{\top}\bs{\Sigma}_2(\widehat{\bs{\beta}} - \bs{\beta}) | \asymp \| \widehat{\bs{\beta}} - \bs{\beta}\|_2^2$.
With probability at least $1 -{O(1/n)}$, we have
\begin{eqnarray*}
&& | (\widehat{\bs{\beta}} - \bs{\beta})^{\top}\bs{\Sigma}_2(\widehat{\bs{\beta}} - \bs{\beta}) | \\
& =& |(\widehat{\bs{\beta}} - \bs{\beta})^{\top} (\bs{\Sigma}_2 - \widehat{\bs{\Sigma}}_2)\widehat{\bs{\beta}} +  (\widehat{\bs{\beta}} - \bs{\beta})^{\top} (\widehat{\bs{\delta}} - \bs{\delta})| \\
& =& | (\widehat{\bs{\beta}} - \bs{\beta})^{\top} (\bs{\Sigma}_2-\widehat{\bs{\Sigma}}_2)\bs{\beta} + (\widehat{\bs{\beta}} - \bs{\beta})^{\top} (\bs{\Sigma}_2 -\widehat{\bs{\Sigma}}_2)(\widehat{\bs{\beta}} - \bs{\beta})+ (\widehat{\bs{\beta}} - \bs{\beta})^{\top} (\widehat{\bs{\delta}} - \bs{\delta})| \\
& {\lesssim} & \sqrt{\frac{J\log n}{n}}\| \widehat{\bs{\beta}} - \bs{\beta}\|_2 \| \bs{\beta}\|_2 +
\sqrt{\frac{J\log n}{n}}\| \widehat{\bs{\beta}} - \bs{\beta}\|_2^2 + {\sqrt{\frac{\log n}{n}}\| \widehat{\bs{\beta}} - \bs{\beta}\|_2}.
\end{eqnarray*}
In the above inequality, we divide $\| \widehat{\bs{\beta}} - \bs{\beta}\|_2$ on both sides, and we have $\| \widehat{\bs{\beta}} - \bs{\beta}\|_2 \lesssim  \sqrt{\frac{J\log n}{n}} \| \bs{\beta}\|_2$.
\end{proof}

\begin{lemma} \label{LEM:sample}
With probability at least $1-O(1/n)$, the proposed estimators in equation (\ref{sFQDA:est}) satisfy that, 
\begin{eqnarray*}
\| \widehat{\mathbf{D}}_s - \mathbf{D}\|_F &\lesssim&  \left(\sqrt{\frac{J\log n}{n}}+ \sqrt{J} f_2(M)\right) \left( \| \mathbf{D}\bs{\Sigma}_1\|_F + \| \mathbf{D}\bs{\Sigma}_2\|_F \right), \\
\|\widehat{\bs{\beta}}_s - \bs{\beta}\|_2 &\lesssim&  \left(\sqrt{\frac{J\log n}{n}}+ \sqrt{J} f_2(M)\right) \| \bs{\beta}\|_2 + \sqrt{J} f_1(M).
\end{eqnarray*}
\end{lemma}

\begin{proof}
{For simplicity, let $\widetilde{\mathbf{B}} = M^{-1/2}\mathbf{B}$ and $\widetilde\psi_j = M^{-1/2}\psi_j$ for all $j\geq 1$. In the following we omit $k$ for simplicity. Consider the parameter space $\Theta$, when take the first $M$ basis, such that for $\widetilde{\mathbf{B}}^\top\widetilde{\mathbf{B}}=\mathbf{I}_M$, we have
$
\widehat{\bs{\mu}} - \bs{\mu}= \bar{\bs{\xi}} - \bs{\mu} + \widetilde{\mathbf{B}}^\top\bar{\bs{\epsilon}}$, where $\bar{\bs{\xi}}= \left( \bar{\xi}_{\cdot 1}, \bar{\xi}_{\cdot 2}, \ldots, \bar{\xi}_{\cdot M}\right)^{\top}$, $\bar{\bs{\epsilon}} = \frac{1}{n}\sum_{i=1}^{n_k} \bs{\epsilon}_i$, $\bs{\epsilon}_i= \left( \epsilon_{i1}, \epsilon_{i2}, \ldots, \epsilon_{iM}\right)^{\top}$ and $\epsilon_{im}= \sum_{j=J+1}^{\infty}\xi_{ij}\widetilde\psi_{j}(t_m)$, $m=1,\ldots,M$ . For any $M$ dimensional vector $\bs{a}$ and $J \leq M$, define $\bs{a}_J$ as the vector of first $J$ elements in $\bs{a}$.

Since $\bar{\bs{\epsilon}} = \left(\sum_{k=M+1}^{\infty}\bar{\xi}_j\widetilde\psi_j(t_1), \ldots, \sum_{j=M+1}^{\infty}\bar{\xi}_j\widetilde\psi_j(t_M)\right)^\top$, we have normally distributed random vector $$\widetilde{\mathbf{B}}^\top\bar{\bs{\epsilon}} = \left(\sum_{k=M+1}^{\infty}a_{1k} \bar{\xi}_{\cdot k}, \ldots, \sum_{k=M+1}^{\infty}a_{Mk} \bar{\xi}_{\cdot k}\right)^\top, $$
where $a_{jk}=\sum_{m=1}^M\widetilde\psi_j(t_m)\widetilde\psi_k(t_m)=O(1)$,  and 
$$E\left(\sum_{k=M+1}^{\infty} \bar{\xi}_{\cdot k}\right)^2 = f_1^2(M)+n^{-1}\sum_{k=M+1}^{\infty}\lambda_k = f_1^2(M)+O(1/n).$$ 
Hence for the first $J\leq M$ elements, we have 
\begin{eqnarray*}
&&\left\|\left(\widetilde{\mathbf{B}}^\top\bar{\bs{\epsilon}} \right)_J \right\|_2^2 \lesssim J\left( \sum_{k=M+1}^{\infty}\bar{\xi}_{\cdot k}\right)^2 \lesssim  Jf_1^2(M) + \frac{J\log n}{n} \\
\end{eqnarray*}
in probability $1 - O(1/n)$, the last inequality is obtained by Delta method and the fact that $\Phi^{-1}(1/n)\approx \sqrt{2\log n}$, where $\Phi(\cdot)$ is cumulative density function of standard normal distribution. 
therefore, 
\begin{eqnarray*}
 \| \widehat{\bs{\mu}}_J - \bs{\mu}_J\|_2 \leq   \| \bar{\bs{\xi}}_J - \bs{\mu}_J\|_2  + \|\left(\widetilde{\mathbf{B}}^{\top}\bar{\bs{\epsilon}} \right)_J \|_2
 \lesssim   \sqrt{J}f_1(M)+\sqrt{\frac{J\log n}{n}} 
\end{eqnarray*}
with probability $1-O(1/n)$.  

Let $\bs e_{ij}$ be the $j$-th element of $\widetilde{\mathbf{B}}^\top {\bs{\epsilon_{i}}}$ and $\bs \bar e_{j}$ be the $j$-th element of $\widetilde{\mathbf{B}}^\top\bar{\bs{\epsilon}}$ for simplicity. The variance of the first $J$ scores is estimated as $\widehat{\lambda}_j = \frac{1}{n}\sum_{i=1}^n \left(\widehat{\mu}_{ij} - \widehat{\mu}_j \right)^2$, i.e.
$\widehat{\mathbf{\Sigma}}=\text{diag}\left( \widehat{\lambda}_1, \ldots, \widehat{\lambda}_J\right)$.
For $\widehat{\lambda}_j$, such that $1\leq j\leq J$,  we have
\begin{eqnarray*}
\widehat{\lambda}_j = \frac{1}{n}\sum_{i=1}^n\left(\xi_{ij}+ e_{ij}-\bar{\xi_{j}}- \bar e_j \right)^2 = \frac{1}{n}\sum_{i=1}^n\left(\zeta_{ij}-\bar\zeta_{j}\right)^2. 
\end{eqnarray*}
where $\zeta_{ij}$ is normally distributed with variance $\lambda_j + f_2(M)$. Therefore, by Lemma 8.5 in \cite{Cai:Zhang:19b}, with probability $1-O(1/n)$, we have $\sup_j \| \widehat{\lambda}_j - \lambda_j + f_2(M)\|_\infty \lesssim \frac{\log n}{n}$,  which is equivalent to $\|\widehat{\mathbf{\Sigma}}-\mathbf{\Sigma}\|_2 \lesssim f_2(M) +  \frac{\log n}{n}$. As a result, $\|\widehat{\mathbf{\Sigma}}-\mathbf{\Sigma}\|_F \lesssim \sqrt{J}f_2(M) + \sqrt{\frac{J\log n}{n}}$ with probability $1-O(1/n)$. 
Hence, the results can be easily derived from the proof of Lemma \ref{LEM:Dandbeta}. }

\end{proof}
With a little abuse of notation, let
$${Q}(\bs{z}; \bs\theta):=(\bs{z}-{\bs{\mu}}_{1})^\top {\mathbf{D}}(\bs{z}-{\bs{\mu}}_{1}) - 2{\bs{\beta}}^\top (\bs{z}-{\bar{\bs{\mu}}}) - \log\left( |{\mathbf{D}}{\mathbf{\Sigma}}_1+\mathbf{I}_J|\right)+2\log\left(  {\pi_1}/{\pi_2}\right),$$ $ M(\bs{z})=Q(\bs{z};\bs\theta) - \widehat{Q}(\bs{z})$ and $\Delta_{M(\bs{z})}=\sqrt{\frac{J\log n}{n}} (\| \mathbf{D}\|_F + \|\bs{\beta} \|_2)$, where $\widehat{Q}(\bs{z})$ is defined in Section \ref{sec:full:obs}.


\begin{lemma}\label{LEM:Mz}
We have $P\left(M(\bs{z})\lesssim  \Delta_{M(\bs{z})}\right)\geq 1 -O(1/n)$.
\end{lemma}

\begin{proof}
By the definition, we have
\begin{eqnarray*}
&&M(\bs{z})\\
&=& \left\{ (\bs{z} - \bs{\mu}_1)^{\top}\mathbf{D}(\bs{z} - \bs{\mu}_1)  -  (\bs{z} - \widehat{\bs{\mu}}_1)^{\top}\widehat{\mathbf{D}}(\bs{z} - \widehat{\bs{\mu}}_1) \right\} + \left\{ 2\widehat{\bs{\beta}}^{\top}(\bs{z} - \widehat{\bar{\bs{\mu}}}) - 2\bs{\beta}^{\top} (\bs{z} - \bar{\bs{\mu}}) \right\}\\
&& + \log|\widehat{\mathbf{D}}\widehat{\bs{\Sigma}}_1 + \mathbf{I}_J | -   \log|\mathbf{D}\bs{\Sigma}_1 + \mathbf{I}_J |) + 2\log \left( {\pi_1}/{\pi_2}\right) - 2\log \left( {\widehat{\pi}_1}/{\widehat{\pi}_2}\right) \\
&=& \left\{  (\bs{z} - \bs{\mu}_1)^{\top}\mathbf{D}(\bs{z} - \bs{\mu}_1)  -  (\bs{z} - \widehat{\bs{\mu}}_1)^{\top}\widehat{\mathbf{D}}(\bs{z} - \widehat{\bs{\mu}}_1) \right\} - tr \left( \bs{\Sigma}_1^{1/2} (\widehat{\mathbf{D}} - \mathbf{D})\bs{\Sigma}_1^{1/2}\right) \\
&&+ \left\{2\widehat{\bs{\beta}}^{\top}(\bs{z} - \widehat{\bar{\bs{\mu}}}) - 2\bs{\beta}^{\top} (\bs{z} - \bar{\bs{\mu}}) \right\} \\
&&+  \log|\widehat{\mathbf{D}}\widehat{\bs{\Sigma}}_1 + \mathbf{I}_J | -   \log|\mathbf{D}\bs{\Sigma}_1 + \mathbf{I}_J |) + tr \left( \bs{\Sigma}_1^{1/2} (\widehat{\mathbf{D}} - \mathbf{D})\bs{\Sigma}_1^{1/2}\right)\\
&&+ 2\log \left( {\pi_1}/{\pi_2}\right) - 2\log \left( {\widehat{\pi}_1}/{\widehat{\pi}_2}\right)\\
&=& \Xi_1+\Xi_2+\Xi_3+\Xi_4
\end{eqnarray*}
Without loss of generality, assuming $\bs{z}\sim {N}(\bs{\mu}_1, \bs{\Sigma}_1)$.
First, we shall bound $\Xi_1$. Note that
\begin{eqnarray*}
&&  (\bs{z} - \bs{\mu}_1)^{\top}\mathbf{D}(\bs{z} - \bs{\mu}_1)  -  (\bs{z} - \widehat{\bs{\mu}}_1)^{\top}\widehat{\mathbf{D}}(\bs{z} - \widehat{\bs{\mu}}_1)\\
&=& (\bs{z} - \bs{\mu}_1)^{\top}\mathbf{D}(\bs{z} - \bs{\mu}_1)   - (\bs{z} - \widehat{\bs{\mu}}_1)^{\top}\mathbf{D}(\bs{z} - \widehat{\bs{\mu}}_1) + (\bs{z} - \widehat{\bs{\mu}}_1)^{\top}(\widehat{\mathbf{D}} - \mathbf{D})(\bs{z} - \widehat{\bs{\mu}}_1)\\
&=& 2(\bs{\mu}_1 - \widehat{\bs{\mu}}_1)\mathbf{D}\bs{z} + (\bs{\mu}_1 - \widehat{\bs{\mu}}_1)^{\top}\mathbf{D}(\bs{\mu}_1 - \widehat{\bs{\mu}}_1) + (\bs{z} - \widehat{\bs{\mu}}_1)^{\top}(\widehat{\mathbf{D}} - \mathbf{D})(\bs{z} - \widehat{\bs{\mu}}_1)\\
& \leq & 2\sqrt{ (\bs{\mu}_1 - \widehat{\bs{\mu}}_1)^{\top}\mathbf{D}(\bs{\mu}_1 - \widehat{\bs{\mu}}_1)} \sqrt{\bs{z}^{\top}\mathbf{D}\bs{z}}  + (\bs{\mu}_1 - \widehat{\bs{\mu}}_1)^{\top}\mathbf{D}(\bs{\mu}_1 - \widehat{\bs{\mu}}_1) \\
&&+(\bs{z} - \widehat{\bs{\mu}}_1)^{\top}(\widehat{\mathbf{D}} - \mathbf{D})(\bs{z} - \widehat{\bs{\mu}}_1)\\
&=& I_1+I_2+I_3
\end{eqnarray*}
According to Gaussian Chaos, since $\widehat{\bs{\mu}}_1 - \bs{\mu}_1 \sim {N}(\bs{0}, \frac{1}{n}\mathbf{\Sigma}_1)$, and $E\left\{  (\bs{\mu}_1 - \widehat{\bs{\mu}}_1)^{\top}\mathbf{D}(\bs{\mu}_1 - \widehat{\bs{\mu}}_1) \right\} = n^{-1}\sum_{j=1}^J(\epsilon_j -1) $ and $\epsilon_j =\lambda_j^{(1)}/\lambda_j^{(2)}$, we have
\begin{eqnarray*}
 && P\left((\bs{\mu}_1 - \widehat{\bs{\mu}}_1)^{\top}\mathbf{D}(\bs{\mu}_1 - \widehat{\bs{\mu}}_1)  >2\left\| \frac{1}{n} (\text{diag}(\epsilon_1, \ldots, \epsilon_J) - \mathbf{I}_{J})\right\|_F\sqrt{x} \right.\\
&&\left. + 2\left\| \frac{1}{n} (\text{diag}(\epsilon_1, \ldots, \epsilon_J)  - \mathbf{I}_{J})\right\|_2 x+ \sum_{j=1}^J \frac{\epsilon_j -1}{n}\right)\leq  \exp(-x)
\end{eqnarray*}
Since $\| \frac{1}{n} (\text{diag}(\epsilon_1, \ldots, \epsilon_J)  - \mathbf{I}_{J})\|_F = \frac{1}{n}\sqrt{\sum_{j=1}^J(\epsilon_j-1)^2}$ and  $\| \frac{1}{n} (\text{diag}(\epsilon_1, \ldots, \epsilon_J)  - \mathbf{I}_{J})\|_2 = \frac{1}{n}\max_j{(\epsilon_j-1)}$. Therefore, take $x = \log n$, we have
\begin{eqnarray}\label{EQ:I2}
&& I_2 \\
&=&(\bs{\mu}_1 - \widehat{\bs{\mu}}_1)^{\top}\mathbf{D}(\bs{\mu}_1 - \widehat{\bs{\mu}}_1)
 \lesssim  \frac{\sqrt{\log n}}{n}\|\mathbf{D}\bs{\Sigma}_1 \|_F + \frac{\log n}{n} + \frac{1}{n} tr\left(\mathbf{D}\bs{\Sigma}_1 \right)
 \lesssim  \frac{\sqrt{J}}{n} \|\mathbf{D}\bs{\Sigma}_1 \|_F   \nonumber
\end{eqnarray}
with probability at least $1-1/n$.
Note that
$$\bs{z}^{\top}\mathbf{D}\bs{z}=(\bs{z} - \bs{\mu}_1)^{\top}\mathbf{D}(\bs{z} - \bs{\mu}_1) + 2\bs{\mu}_1^{\top}\mathbf{D}(\bs{z} - \bs{\mu}_1) + \bs{\mu}_1^{\top}\mathbf{D}\bs{\mu}_1$$
Since $\bs{z} - \bs{\mu}_1 \sim {N}(\bs{0}, \mathbf{\Sigma}_1)$, and we have $E\left\{  (\bs{z} - \bs{\mu}_1)^{\top}\mathbf{D}(\bs{z} - \bs{\mu}_1 ) \right\} = \sum_{j=1}^J(\epsilon_j -1)$, Thus
\begin{eqnarray*}
(\bs{z} - \bs{\mu}_1)^{\top}\mathbf{D}(\bs{z} - \bs{\mu}_1)
  \lesssim   \sqrt{\log n}\| \mathbf{D} \bs{\Sigma}_1\|_F+ \log n + tr\left( \mathbf{D} \bs{\Sigma}_1\right)  \lesssim   \sqrt{J}\| \mathbf{D}\bs{\Sigma}_1 \|_F
\end{eqnarray*}
with probability at least $1-1/n$.
Note that
 $ 2\bs{\mu}_1^{\top}\mathbf{D}(\bs{z} - \bs{\mu}_1) \sim {N}({0}, 4\sum_{j=1}^J \frac{\mu_{1j}^2}{\lambda^{(1)}_j}(\epsilon_j-1)^2)$, we have  $\bs{\mu}_1^{\top}\mathbf{D}(\bs{z} - \bs{\mu}_1) \lesssim\| \mathbf{D}\bs{\Sigma}_1\|_F\sqrt{\log n} $ with probability at least $1-O(1/n)$.  Since $\bs{\mu}_1^{\top}\mathbf{D}\bs{\mu}_1 \leq \| \bs{\Sigma}_1^{-1/2}\bs{\mu}\|_4^2 \| \mathbf{D}\bs{\Sigma}_1\|_F$,   we have $\bs{z}^{\top}\mathbf{D}\bs{z} \lesssim \sqrt{J}\| \mathbf{D} \bs{\Sigma}_1\|_F$ with probability at least $1-O(1/n)$. Thus
\begin{equation}\label{EQ:muD}
 I_1=2\sqrt{ (\bs{\mu}_1 - \widehat{\bs{\mu}}_1)^{\top}\mathbf{D}(\bs{\mu}_1 - \widehat{\bs{\mu}}_1)} \sqrt{\bs{z}^{\top}\mathbf{D}\bs{z}}  \lesssim \sqrt{\frac{J}{n}} \| \mathbf{D} \bs{\Sigma}_1\|_F.
\end{equation}

Denote the diagonal matrix $\bs{\Sigma}_1^{1/2} (\widehat{\mathbf{D}} - \mathbf{D})\bs{\Sigma}_1^{1/2}= \text{diag}(\rho_1,\ldots, \rho_J)$, and $\bs{z}_{0}=(z_{01}, \ldots,z_{0J})^{\top}$ $ \sim {N}(\mathbf{0},\mathbf{I}_{J})$ are independent standard normally distributed random variables, then we have
\begin{eqnarray}
&& I_3 - tr \left( \bs{\Sigma}_1^{1/2} (\widehat{\mathbf{D}} - \mathbf{D})\bs{\Sigma}_1^{1/2}\right)\nonumber\\
&=&(\bs{z} - \widehat{\bs{\mu}}_1)^{\top}(\widehat{\mathbf{D}} - \mathbf{D})(\bs{z} - \widehat{\bs{\mu}}_1) - tr \left( \bs{\Sigma}_1^{1/2} (\widehat{\mathbf{D}} - \mathbf{D})\bs{\Sigma}_1^{1/2}\right)\nonumber\\
&=& (1+\frac{1}{n})\bs{z}_0^{\top} (\bs{\Sigma}_1^{1/2} (\widehat{\mathbf{D}} - \mathbf{D})\bs{\Sigma}_1^{1/2}) \bs{z}_0 - tr \left( \bs{\Sigma}_1^{1/2} (\widehat{\mathbf{D}} - \mathbf{D})\bs{\Sigma}_1^{1/2}\right)\nonumber\\
&=& (1+\frac{1}{n})\sum_{j=1}^J \rho_j (z_{0j}^2-1) + \frac{1}{n}\sum_{j=1}^J \rho_j \nonumber\\
& \lesssim &  (1+\frac{1}{n})\sqrt{\frac{J\log n}{n}} + \frac{1}{n} \|\widehat{\mathbf{D}} - \mathbf{D}\|_F \|\bs{\Sigma}_1 \|_2 \nonumber\\
& \lesssim &  \sqrt{\frac{J\log n}{n}} (\| \mathbf{D}\bs{\Sigma}_1\|_F + \| \mathbf{D}\bs{\Sigma}_2\|_F) \nonumber
\end{eqnarray}
with probability at least $1 - O(1/n)$, where the second last inequality comes from Gaussian Chaos and the last inequality comes from Lemma \ref{LEM:Dandbeta}.
Combining (\ref{EQ:I2}), (\ref{EQ:muD}) and (\ref{EQ:1}),
\begin{eqnarray}\label{EQ:1}
\Xi_1 \lesssim \sqrt{\frac{J\log n}{n}} (\| \mathbf{D}\bs{\Sigma}_1\|_F + \| \mathbf{D}\bs{\Sigma}_2\|_F).
\end{eqnarray}


Secondly, by Lemma \ref{LEM:Dandbeta}, with probability at least $1 - O(1/n)$, we have
\begin{eqnarray}\label{EQ:2}
\Xi_2= |2\widehat{\bs{\beta}}(\bs{z} - \widehat{\bar{\bs{\mu}}}) - 2\bs{\beta}(\bs{z} - \bar{\bs{\mu}}) |
& = &  |2(\widehat{\bs{\beta}} - \bs{\beta})(\bs{z} - \widehat{\bar{\bs{\mu}}}) - \bs{\beta}(\widehat{\bar{\bs{\mu}}} -\bar{\bs{\mu}})|\nonumber\\
&\leq & 2 \| \widehat{\bs{\beta}} - \bs{\beta}\|_2 \| \bs{z} - \widehat{\bar{\bs{\mu}}}\|_2 + \| \bs{\beta}\|_2\| \widehat{\bar{\bs{\mu}}} -\bar{\bs{\mu}}\|_2
\lesssim \sqrt{\frac{J\log n}{n}}\| \bs{\beta}\|_2.
\end{eqnarray}

Thirdly, we have
\begin{eqnarray*}
&& \log|\mathbf{D}\bs{\Sigma}_1 + \mathbf{I}_J | - \log|\widehat{\mathbf{D}}\widehat{\bs{\Sigma}}_1 + \mathbf{I}_J | \\
& \leq & tr\left\{(\mathbf{D}\bs{\Sigma}_1 + \mathbf{I}_J)^{-1}(\mathbf{D}\bs{\Sigma}_1 - \widehat{\mathbf{D}}\widehat{\bs{\Sigma}}_1) \right\} \\
&=& tr\left\{ -\mathbf{D}\bs{\Sigma}_2(\mathbf{D}\bs{\Sigma}_1 - \widehat{\mathbf{D}}\widehat{\bs{\Sigma}}_1) \right\} + tr(\mathbf{D}\bs{\Sigma}_1 - \widehat{\mathbf{D}}\widehat{\bs{\Sigma}}_1) \\
& \leq &  \| \mathbf{D} \bs{\Sigma}_2\|_F \| \mathbf{D}\bs{\Sigma}_1 - \widehat{\mathbf{D}}\widehat{\bs{\Sigma}}_1\|_F + tr(\widehat{\mathbf{D}}\bs{\Sigma}_1 - \widehat{\mathbf{D}}\widehat{\bs{\Sigma}}_1) + tr(\mathbf{D}\bs{\Sigma}_1 - \widehat{\mathbf{D}}\bs{\Sigma}_1).
\end{eqnarray*}
Since with probability at least $1 - O(1/n)$,
\begin{eqnarray*}
&&  \| \mathbf{D}\bs{\Sigma}_2\|_F \| \mathbf{D}\bs{\Sigma}_1 - \widehat{\mathbf{D}}\widehat{\bs{\Sigma}}_1\|_F  \\
& \leq & \| \mathbf{D} \bs{\Sigma}_2\|_F \left( \| (\bs{\Sigma}_1 - \widehat{\bs{\Sigma}}_1)\widehat{\mathbf{D}}\|_F + \| \bs{\Sigma}_1(\widehat{\mathbf{D}} - \mathbf{D}) \|_F\right) \\
& \leq & \| \mathbf{D} \bs{\Sigma}_2\|_F \left( \| (\bs{\Sigma}_1 - \widehat{\bs{\Sigma}}_1)(\widehat{\mathbf{D}} - \mathbf{D})\|_F +  \| (\bs{\Sigma}_1 - \widehat{\bs{\Sigma}}_1)\mathbf{D}\|_F + \| \bs{\Sigma}_1(\widehat{\mathbf{D}} - \mathbf{D}) \|_F\right) \\
& \lesssim & \sqrt{\frac{J\log n}{n}} ( \| \mathbf{D} \bs{\Sigma}_1\|_F\|\mathbf{D} \bs{\Sigma}_2\|_F +\|\mathbf{D} \bs{\Sigma}_2\|_F^2 )
\end{eqnarray*}
and with probability at least $1 - O(1/n)$,
\begin{eqnarray*}
  tr(\widehat{\mathbf{D}}\bs{\Sigma}_1 - \widehat{\mathbf{D}}\widehat{\bs{\Sigma}}_1)
& \leq & \| \widehat{\mathbf{D}}\bs{\Sigma}_1\|_F \| \mathbf{I}_J - \widehat{\bs{\Sigma}}_1\bs{\Sigma}_1^{-1} \|_F \\
& \leq & \|( \widehat{\mathbf{D}} - \mathbf{D})\bs{\Sigma}_1\|_F \| \bs{\Sigma}_1 - \widehat{\bs{\Sigma}}_1 \|_F  + \|\mathbf{D}\bs{\Sigma}_1 \|_F \|\bs{\Sigma}_1 - \widehat{\bs{\Sigma}}_1 \|_F\\
& \leq & \| \widehat{\mathbf{D}} - \mathbf{D}\|_F\|\bs{\Sigma}_1 \|_{\infty} \| \bs{\Sigma}_1 - \widehat{\bs{\Sigma}}_1 \|_F  + \|\mathbf{D}\bs{\Sigma}_1 \|_F \|\bs{\Sigma}_1 - \widehat{\bs{\Sigma}}_1 \|_F\\
& \lesssim & \sqrt{\frac{J\log n}{n}} \left\{\sqrt{\frac{J\log n}{n}} \left(\|\mathbf{D}\bs{\Sigma}_1 \|_F +\|\mathbf{D}\bs{\Sigma}_2 \|_F \right)  +\|\mathbf{D}\bs{\Sigma}_1 \|_F  \right\}\\
& \lesssim & \sqrt{\frac{J\log n}{n}}\|\mathbf{D}\bs{\Sigma}_1 \|_F + \frac{J\log n}{n}\|\mathbf{D}\bs{\Sigma}_2 \|_F
\end{eqnarray*}
and
\begin{eqnarray*}
  tr(\mathbf{D}\bs{\Sigma}_1 - \widehat{\mathbf{D}}\bs{\Sigma}_1)
  \leq   \|\mathbf{D} - \widehat{\mathbf{D}}\|_F \|\bs{\Sigma}_1  \|_F
 \lesssim   \sqrt{\frac{J\log n}{n}}\left( \|\mathbf{D}\bs{\Sigma}_1 \|_F + \|\mathbf{D}\bs{\Sigma}_2 \|_F \right).
\end{eqnarray*}
Note that $tr\left(\bs{\Sigma}_1^{1/2}\left(\mathbf{D} - \widehat{\mathbf{D}}\right)\bs{\Sigma}_1^{1/2}\right) = tr(\mathbf{D}\bs{\Sigma}_1 - \widehat{\mathbf{D}}\bs{\Sigma}_1) $, hence we have with  probability at least $1-O(1/n)$
\begin{equation}\label{EQ:3}
 \Xi_3 \lesssim \sqrt{\frac{J\log n}{n}}\left( \|\mathbf{D}\bs{\Sigma}_1 \|_F + \|\mathbf{D}\bs{\Sigma}_2 \|_F \right)
\end{equation}

Lastly, by Hoeffding inequality, we have $\widehat{\pi}_k - \pi_k \lesssim \sqrt{\frac{\log n}{n}}$  with  probability at least  $1 - O(1/n)$, $k = 1, 2$. Thus
\begin{eqnarray}\label{EQ:4}
\Xi_4= \left| \log \left( \frac{\pi_1}{\pi_2}\right) - \log \left( \frac{\widehat{\pi}_1}{\widehat{\pi}_2}\right) \right|
\lesssim   |\log (\widehat{\pi}_1 - \pi_1 )| +  |\log (\widehat{\pi}_2 - \pi_2 )|
 \lesssim  \sqrt{\frac{\log n}{n}}
\end{eqnarray}
in probability $1 - O(1/n)$.

Combining (\ref{EQ:1}) to (\ref{EQ:4}), the lemma has been proved.
\end{proof}

\begin{lemma}\label{LEM:fq}
Denote $f_{Q(\bs{z})}(t)$ the probability density function of $Q(\bs{z})$. When $\Delta_{M(\bs{z})}= o(1)$,  we have
\begin{eqnarray*}
 \int_0^{C\Delta_{M(\bs{z})}}  (1 - e^{-t})f_{Q(\bs{z})}(t) dt \lesssim & \Delta_{M(\bs{z})}\int_0^{C\Delta_{M(\bs{z})}} f_{Q(\bs{z})}(t) dt
 \lesssim\Delta_{M(\bs{z})}^2.
\end{eqnarray*}
 When $\Delta_{M(\bs{z})}=O(1)$ or $\Delta_{M(\bs{z})}=\infty$, we have
\begin{eqnarray*}
\int_0^{C\Delta_{M(\bs{z})}}  (1 - e^{-t})f_{Q(\bs{z})}(t) dt \lesssim & \int_0^{C\Delta_{M(\bs{z})}} f_{Q(\bs{z})}(t) dt
 \lesssim  \Delta_{M(\bs{z})}.
\end{eqnarray*}
\end{lemma}

\begin{proof}
Without loss of generality, assuming $\bs{z}\sim {N}(\bs{\mu}_1, \bs{\Sigma}_1)$. By simple calculation, we have
\begin{equation} \label{EQ:Qz}
Q(\bs{z}) = \sum_{j=1}^J (\epsilon_j-1)\chi^2(\delta_j^2) - \sum_{j=1}^J \frac{(\mu_{1j}-\mu_{2j})^2}{\lambda_{j}^{(2)}}(\epsilon_j-1)^{-1}-  \sum_{j=1}^J \log \epsilon_j,
\end{equation}
where $\chi^2(\delta_j^2) = h_j^2$, such that $h_j \sim {N}\left(\delta_j, 1\right)$, where $\delta_j=\frac{\mu_{1j}-\mu_{2j}}{(\lambda_{j}^{(2)})^{1/2}}\cdot\frac{\epsilon_j^{1/2}}{\epsilon_j-1}$.
Denote $q_z= \sum_{j=1}^J (\epsilon_j-1)\chi_j^2(\delta_j)$. {To estimate the density of $q_z$, without loss of generality, we assume that $\epsilon_1-1 \geq \epsilon_2-1 \geq \cdots \geq 0 $, otherwise, $q_z$ can always be represented as the subtraction of two linear combinations of non-central chi-square random variables with positive coefficients, whose density function can be derived by convolution, thus, the boundedness of the density can be obtained by this simple case.} As suggested in  \cite{Liu:etal:09} ,  define
$$A_k = \sum_{j=1}^J \left\{(\epsilon_j-1)^k + k (\epsilon_j-1)^{k-1}\frac{(\mu_{1j}-\mu_{2j})^2}{\lambda_{j}^{(2)}} +  (\epsilon_j-1)^{k-2}\frac{(\mu_{1j}-\mu_{2j})^2}{\lambda_{j}^{(2)}}\right\},
$$
 then  $\mu_{q} := E(q_z)= \sum_{j=1}^J \left\{ \epsilon_j-1 + \frac{(\mu_{1j}-\mu_{2j})^2}{\lambda_{j}^{(2)}} + \frac{(\mu_{1j}-\mu_{2j})^2}{\lambda_{j}^{(2)}}(\epsilon_j-1)^{-1}  \right\}$ and 
 $$\sigma_{q} :=SD(q_z)=  \sqrt{ 2\sum_{j=1}^J \left\{(\epsilon_j-1)^2 +2 \frac{(\mu_{1j}-\mu_{2j})^2}{\lambda_{j}^{(2)}} \epsilon_j\right\}}. $$
Define
$\omega_1^2=A_3^2A_2^{-3}$
and
$\omega_2= A_4A_2^{-2}$.

i) When $\omega_1^2>\omega_2$, define $\mu_{\chi}=\frac{\omega_1 - 2\sqrt{\omega_1^2 - \omega_2}}{(\omega_1 - \sqrt{\omega_1^2 - \omega_2})^3}$, $\sigma_{\chi}=\frac{\sqrt{2}}{\omega_1 - \sqrt{\omega_1^2 - \omega_2}}$, $D_{q} = \frac{\mu_{\chi}}{\sigma_{\chi}} \sigma_{q} - \mu_{q}$  and  $U_{q} = \frac{\sqrt{\omega_1^2 - \omega_2}}{(\omega_1 - \sqrt{\omega_1^2 - \omega_2})^3}$. Denote $\gamma   = \frac{\omega_1-3\sqrt{\omega_1^2 - \omega_2}}{(\omega_1 - \sqrt{\omega_1^2 - \omega_2})^3}$. Then, the probability density function of $q_z$ is approximately by noncentral chi-square distribution with noncentrality parameter $U_q$ and degree of freedom $\gamma$ as following
$$f_{q}(t) = \frac{\sigma_{\chi}}{2\sigma_{q}}\exp \left(-\frac{K_q+ U_{q}}{2}\right)\left(\frac{K_q}{U_{q}}\right)^{\gamma/4-1/2}I_{\frac{\gamma}{2}-\gamma}\left(\sqrt{U_{q}K_q}\right), $$ where $K_q=\frac{\sigma_{\chi}}{\sigma_{q}}(t+D_{q} )$, $I_{\gamma/2-\gamma}(\cdot)$ is a modified Bessel function.
Note that
\begin{eqnarray*}
f_{q}(t)
& \leq & \frac{\sigma_{\chi}}{2\sigma_{q}} {U_{q}}^{1/2-\gamma/4}\exp \left(-\frac{U_q+K_q}{2} \right)K_q^{\gamma/4-1/2}\\
&&\times\left( {U_{q}K_q}\right)^{{\gamma}/{4}-1/2}\exp \left( \sqrt{U_{q}K_q}\right)2^{1-{\gamma}/{2}}\Gamma^{-1}({\gamma}/{2}-1)\\
& \leq & \frac{\sigma_{\chi}}{2^{\frac{\gamma}{2}}\sigma_{q}}\Gamma^{-1}\left({\gamma}/{2}-1\right)\exp \left\{-(U_q+K_q)/2 +\sqrt{U_{q}K_q}\right\}.
\end{eqnarray*}
The second last inequality is obtained from \cite{Luke:72} and Bessel function  has the inequality $1<\Gamma(\nu+1)\left(\frac{2}{x}\right)^{\nu}I_{\nu}(x) < \cosh x \leq e^x$, where $\nu > -1/2$, $x>0$. Here we can still find that $f_{q}(t)$ has lower bound $\frac{\sigma_{\chi}}{2^{\frac{\gamma}{2}}\sigma_{q}}\Gamma^{-1}({\gamma}/{2}-1)\exp \left\{-(U_{q}+K_q)/2\right\} $, which will be used later. Then we have
\begin{eqnarray*}
f_{Q({\bs{z}})}(t)
 \leq  \frac{\sigma_{\chi}}{2^{\frac{\gamma}{2}}\sigma_{q}}\Gamma^{-1}({\gamma}/{2}-1)\exp \left\{-\frac{\sigma_{\chi}L_{J}}{2\sigma_{q}}-\frac{K_q+U_q}{2} +\sqrt{U_{q}\left(K_q+\frac{\sigma_{\chi}}{\sigma_{q}}L_J\right)}\right\},
\end{eqnarray*}
where $L_J= \sum_{j=1}^J \frac{(\mu_{1j}-\mu_{2j})^2}{\lambda_{j}^{(2)}}(\epsilon_j-1)^{-1} + \sum_{j=1}^J \log \epsilon_j$. The upper bound is a function increasing when $t < \frac{\sigma_{q}}{\sigma_{\chi}}U_{q} - D_{q} - L_J $, and decreasing otherwise. When $t = \frac{\sigma_{q}}{\sigma_{\chi}}U_{q} - D_{q} - L_J$, the global maximal value of density is $\frac{\sigma_{\chi}}{2^{\frac{\gamma}{2}}\sigma_{q}}\Gamma^{-1}({\gamma}/{2}-1)$.


When $\Delta_{M(\bs{z})}= o(1)$, then
\begin{eqnarray*}
 \int_0^{C\Delta_{M(\bs{z})}}  (1 - e^{-t})f_{Q(\bs{z})}(t) dt \lesssim & \Delta_{M(\bs{z})}\int_0^{C\Delta_{M(\bs{z})}} f_{Q(\bs{z})}(t) dt
 \lesssim\Delta_{M(\bs{z})}^2 \frac{\sigma_{\chi}}{\sigma_{q}}\frac{1}{2^{\frac{\gamma}{2}}\Gamma(\frac{\gamma}{2}-1)},
\end{eqnarray*}
otherwise
\begin{eqnarray*}
\int_0^{C\Delta_{M(\bs{z})}}  (1 - e^{-t})f_{Q(\bs{z})}(t) dt \lesssim & \int_0^{C\Delta_{M(\bs{z})}} f_{Q(\bs{z})}(t) dt
 \lesssim  \Delta_{M(\bs{z})} \frac{\sigma_{\chi}}{\sigma_{q}}\frac{1}{2^{\frac{\gamma}{2}}\Gamma(\frac{\gamma}{2}-1)}.
\end{eqnarray*}


ii) When  $\omega_1^2 \leq \omega_2$, noncentrality parameter $U_{q} = 0$, $\gamma = \omega_1^{-2}$. The probability density function of $q_z$ is
$$f_{q}(t) = \frac{\sigma_{\chi}}{\sigma_{q}}\frac{1}{2^{\gamma/2}\Gamma(\gamma/2)}\exp \left(-K_q/2\right)K_q^{\gamma/2-1}.$$
Note that $\frac{\sigma_{\chi}}{\sigma_{q}}(q_z+D_{q})\sim \chi^2_{\gamma}$. Given the parameter space $\Theta$, scale parameter $\sigma_{\chi}/\sigma_q <\infty$, thus $f_q(t) <\infty$  when $\omega_1^2 \leq \omega_2$.


Hence, when $\Delta_{M(\bs{z})} \rightarrow 0$, we have
\begin{eqnarray*}
 \int_0^{C\Delta_{M(\bs{z})}}  (1 - e^{-t})f_{Q(\bs{z})}(t) dt
 \lesssim   \Delta_{M(\bs{z})}\int_0^{C\Delta_{M(\bs{z})}} f_{Q(\bs{z})}(t) dt
  \lesssim  \Delta_{M(\bs{z})}^2\frac{\sigma_{\chi}}{\sigma_{q}}\frac{1}{\sqrt{4\pi \gamma}}
  \lesssim   \Delta_{M(\bs{z})}^2,
\end{eqnarray*}
otherwise  we have
\begin{eqnarray*}
 \int_0^{C\Delta_{M(\bs{z})}}  (1 - e^{-t})f_{Q(\bs{z})}(t) dt
\lesssim  \int_0^{C\Delta_{M(\bs{z})}} f_{Q(\bs{z})}(t) dt
 \lesssim  \Delta_{M(\bs{z})} \frac{\sigma_{\chi}}{\sigma_{q}}\frac{1}{\sqrt{4\pi \gamma}}
 \lesssim \Delta_{M(\bs{z})}.
\end{eqnarray*}
\end{proof}


\begin{lemma} \label{LEM:GJ}
Consider parameter space $\Theta$, we have
$$ \sup_{\bs{\theta} \in \Theta} E\left[ R_{\bs\theta}(\widehat{G}_J) - R_{\bs\theta}(G^{\ast}_{\bs\theta})\right]  \lesssim \frac{J\log n}{n} +{ g(J; \Theta)} , $$
  where $\widehat{G}_J$ is the proposed classifier for first $J$ scores. Thus we can easily conclude the least upper bound
$$ \inf_{\widehat{G}}\sup_{\bs{\theta} \in \Theta} E\left[ R_{\bs\theta}(\widehat{G}) - R_{\bs\theta}(G^{\ast}_{\bs\theta})\right]  \lesssim \frac{J^{\ast}\log n}{n}  , $$
where $J^{\ast}$ satisfies $\frac{J^{\ast}\log n}{n} = {g(J^{\ast}; \Theta)}$.
\end{lemma}

\begin{proof}
 Without loss of generality, assuming $\bs{z}\sim {N}(\bs{\mu}_1, \bs{\Sigma}_1)$. \begin{eqnarray*}
&& R_{\bs\theta}(\widehat{ {G}}_J) - R_{\bs\theta}( {G}_{J}^{\ast}) \\
&=&\frac{1}{2}\int_{Q(\bs{z})>0}\frac{\pi_1}{(2\pi)^{p/2} |\bs\Sigma_1 |^{1/2}}e^{-1/2(\bs{z}-\bs{\mu}_1)^\top\bs\Sigma^{-1}_1 (\bs{z} -\bs{\mu}_1 )}d\bs{z} \\
 &&+ \frac{1}{2}\int_{Q(\bs{z})\leq 0}\frac{\pi_2}{(2\pi)^{p/2} |\bs\Sigma_2 |^{1/2}}e^{-1/2(\bs{z}-\bs{\mu}_2)^\top\bs\Sigma^{-1}_2 (\bs{z} -\bs{\mu}_2 )}d\bs{z} \\
&& -  \frac{1}{2}\int_{\widehat{Q}(\bs{z})>0}\frac{\pi_1}{(2\pi)^{p/2} |\bs\Sigma_1 |^{1/2}}e^{-1/2(\bs{z}-\bs{\mu}_1)^\top\bs\Sigma^{-1}_1 (\bs{z} -\bs{\mu}_1 )}d\bs{z} \\
&&- \frac{1}{2}\int_{\widehat{Q}(\bs{z})\leq 0}\frac{\pi_2}{(2\pi)^{p/2} |\bs\Sigma_2 |^{1/2}}e^{-1/2(\bs{z}-\bs{\mu}_2)^\top\bs\Sigma^{-1}_2 (\bs{z} -\bs{\mu}_2 )}d\bs{z} \\
&=&\frac{1}{2}\int_{Q(\bs{z})>0}\frac{1}{(2\pi)^{p/2} |\bs\Sigma_1 |^{1/2}}e^{-1/2(\bs{z}-\bs{\mu}_1)^\top\bs\Sigma^{-1}_1 (\bs{z} -\bs{\mu}_1 ) - \log|\bs\Sigma_1 |/2} \left( 1- e^{-Q(\bs{z})}\right)d\bs{z} \\
&&-\frac{1}{2}\int_{\widehat{Q}(\bs{z})>0}\frac{1}{(2\pi)^{p/2} |\bs\Sigma_1 |^{1/2}}e^{-1/2(\bs{z}-\bs{\mu}_1)^\top\bs\Sigma^{-1}_1 (\bs{z} -\bs{\mu}_1 ) - \log|\bs\Sigma_1 |/2} \left( 1- e^{-Q(\bs{z})}\right)d\bs{z} \\
&\leq & \frac{1}{2}\int_{Q(\bs{z})>0,\widehat{Q}(\bs{z})\leq 0 }\frac{1}{(2\pi)^{p/2} |\bs\Sigma_1 |^{1/2}}e^{-1/2(\bs{z}-\bs{\mu}_1)^\top\bs\Sigma^{-1}_1 (\bs{z} -\bs{\mu}_1 ) - \log|\Sigma_1 |/2} \left( 1- e^{-Q(\bs{z})}\right)d\bs{z} \\
&=& \frac{1}{2}\int_{Q(\bs{z})>0,Q(\bs{z}) \leq Q(\bs{z}) - \widehat{Q}(\bs{z}) }\frac{1}{(2\pi)^{p/2} |\bs\Sigma_1 |^{1/2}}e^{-1/2(\bs{z}-\bs{\mu}_1)^\top\bs\Sigma^{-1}_1 (\bs{z} -\bs{\mu}_1 ) - \log|\bs\Sigma_1 |/2} \\
&& \times\left( 1- e^{-Q(\bs{z})}\right)d\bs{z} \\
&=& \frac{1}{2} E_{\bs{z} \sim {N}(\bs{\mu_1}, \bs{\Sigma}_1)}\left\{ (1 - e^{-Q(\bs{z})})\mathbb{I}\left\{ 0 < Q(\bs{z}) \leq M(\bs{z})\right\}\mathbb{I}\left( M(\bs{z}) \lesssim \Delta_{M(\bs{z})}\right) \right\} \\
& & + \frac{1}{2} E_{\bs{z} \sim {N}(\bs{\mu_1}, \bs{\Sigma}_1)}\left\{ (1 - e^{-Q(\bs{z})})\mathbb{I}\left( 0 < Q(\bs{z}) \leq M(\bs{z})\right)\mathbb{I}\left( M(\bs{z}) \gtrsim \Delta_{M(\bs{z})}\right) \right\} \\
&\lesssim & E_{\bs{z} \sim {N}(\bs{\mu_1}, \bs{\Sigma}_1)}\left\{ (1 - e^{-Q(\bs{z})})\mathbb{I}\left( 0 < Q(\bs{z})   \lesssim \Delta_{M(\bs{z})}\right) \right\} + n^{-1}.
\end{eqnarray*}
By Lemma \ref{LEM:Mz},
we have $P\left\{ M(\bs{z}) \lesssim \Delta_{M(\bs{z})}\right\} = 1-O(1/n)$. Hence,
$R_{\bs\theta}(\widehat{ {G}}_J) - R_{\bs\theta}( {G}_{J}^{\ast}) \lesssim \int_0^{C\Delta_{M(\bs{z})}}  (1 - e^{-t})f_{Q(\bs{z})}(t) dt + n^{-1}$.
By Lemma \ref{LEM:fq} we have
 \begin{equation}\label{EQ:Rj}
 R_{\bs\theta}(\widehat{ {G}}_J) - R_{\bs\theta}( {G}_{J}^{\ast}) \lesssim \frac{J\log n}{n}.
 \end{equation}

Next, we approximate the first $J$ scores and the whole process.  Note that
\begin{eqnarray*}
  R_{\bs\theta}( {G}_{J}^{\ast}) -R_{\bs\theta}( {G}_{\bs\theta}^{\ast})  \asymp  P(Q_{\infty}(Z)>0)-P(Q(\bs{z})>0)
= \int_{0}^{\infty}( f_{{\infty}}(t)-f_{Q(\bs{z})}(t)) dt,
\end{eqnarray*}
where $f_{\infty}$ is the density function of $Q_{\infty}$. Denote
\begin{eqnarray*}
R_J = Q_{\infty} - Q_J
= \sum_{j=J+1}^{\infty} (\epsilon_j-1)\chi^2(\delta_j^2) - \sum_{j=J+1}^{\infty} \frac{(\mu_{1j}-\mu_{2j})^2}{\lambda_{j}^{(2)}}(\epsilon_j-1)^{-1}-  \sum_{j=J+1}^{\infty} \log \epsilon_j .
\end{eqnarray*}
Then we have
\begin{eqnarray*}
&& P(Q_{\infty} > 0)-P(Q_J > 0)\\
&=& P(Q_{\infty} > 0)-P(Q_{\infty} > R_J)  =P(Q_{\infty} > 0) - E\left[P\left(Q_{\infty} > R_J|R_J \right)\right]\\
&=&  \left( 1 - \int_{-\infty}^0 f_{\infty}(t)dt\right) - E\left(1-\int_{-\infty}^{R_J} f_{\infty}(t)dt \right) = E\left(\int_{R_J}^0 f_{\infty}(t)dt \right).
\end{eqnarray*}
When  $0<f_{\infty}(t)<\infty$, we have
$$ E\left(\int_0^{R_J} f_{\infty}(t)dt \right) \asymp E\left(R_J\right).$$
Recall in (\ref{EQ:Qz}),  $\chi^2(\delta_j^2) = h_j^2$ and $h_j \sim {N}\left(\delta_j, 1\right)$, we have $E\left(R_J\right) = \sum_{j=J+1}^{\infty} \left(\epsilon_j-1-\log \epsilon_j \right)+\sum_{j=J+1}^{\infty}\frac{(\mu_{1j}-\mu_{2j})^2}{\lambda_{j}^{(2)}}$. Given parameter space ${\Theta}$,  $\epsilon_j$'s are constants for any $j \in \mathbb{N}$, we have  $E\left(R_J\right) \asymp \sum_{j=J+1}^{\infty} \left(\epsilon_j-1 \right)^2+\sum_{j=J+1}^{\infty}\frac{(\mu_{1j}-\mu_{2j})^2}{\lambda_{j}^{(2)}}$. Hence, we have
\begin{eqnarray} \label{EQ:Rtheta}
 R_{\bs\theta}( {G}_{J}^{\ast}) -R_{\bs\theta}( {G}_{\bs\theta}^{\ast}) \asymp {\sum_{j=J+1}^{\infty} \left(\epsilon_j-1 \right)^2+\sum_{j=J+1}^{\infty}\frac{(\mu_{1j}-\mu_{2j})^2}{\lambda_{j}^{(2)}} }.
\end{eqnarray}
Hence, by (\ref{EQ:Rj}) and (\ref{EQ:Rtheta}), we have
\begin{eqnarray*}
&&  \inf_{\widehat{G}}\sup_{\bs{\theta} \in \Theta} E\left[ R_{\bs\theta}(\widehat{G}) - R_{\bs\theta}(G^{\ast}_{\bs\theta})\right] \\
&=& \inf_{\widehat{G}}\sup_{\bs{\theta} \in \Theta} E\left[ R_{\bs\theta}(\widehat{G}) -  R_{\bs\theta}(G^{\ast}_{J}) \right]+ \sup_{\bs{\theta} \in \Theta}\left[ R_{\bs\theta}(G^{\ast}_{J}) - R_{\bs\theta}(G^{\ast}_{\bs\theta})\right]
\lesssim \frac{J\log n}{n} + {g(J; \Theta)}.
\end{eqnarray*}
\end{proof}
\subsection{Proof of Theorems \ref{THM:full:qda} (upper bound)}
\begin{proof}
Theorem \ref{THM:full:qda} can be derived easily from Lemmas \ref{LEM:GJ} and   \ref{LEM:self-similar full}. 
\end{proof}


\subsection{Proof of Theorem \ref{THM:full} (lower bound)}
\begin{proof} 
Let $\bs{e}_j$ be the vector of length $J$ with $j$-th element $1$ and $0$ otherwise, where $J>4J^{\ast}$, $J^\ast$ is defined in Lemma \ref{LEM:self-similar full}. For some constants $\tau_1$ and $\tau_2$ which are specified later, we consider the parameter space
\begin{eqnarray*}
\Theta_{J^\ast}^{(1)} =  \left\{\right.\bs\theta_u =&& \left(1/2, 1/2,  \bs\mu_1, \bs\mu_2, \mathbf{\Sigma}, \mathbf{\Sigma} \right): \bs{\Sigma} = \text{diag}\left( \lambda_1,\ldots,  \lambda_J\right), \\ 
&&\left.  \bs\mu_1 = \tau_1 \bs{e}_1 +  \tau_2 \sqrt{\frac{\log n}{n }}  \sum_{j=3}^{J-1}u_j \bs{e}_{j}, \bs{u}=(u_1,\ldots,u_J) \in \mathcal{A}_{J, J^{\ast}} \right\}
\end{eqnarray*}
where $\mathcal{A}_{J, J^{\ast}}$ is defined in Lemma \ref{LEM:A}. Thus, we have $\Theta_{J^\ast}^{(1)} \subseteq \Theta$.
For any $\bs u \in \mathcal{A}_{J, J^{\ast}}$, define $\bs \mu_u = \tau_1 \bs{e}_1 +  \tau_2 \sqrt{\frac{\log n}{n }}$, then the KL divergence between ${P}_{\bs\theta_{\bs u}}=N(\bs\mu_{\bs u}, \mathbf{\Sigma})$ and ${P}_{\bs\theta_{\bs u'}}=N(\bs\mu_{\bs u'}, \mathbf{\Sigma})$ is $\frac{1}{2}\|\bs\mu_{\bs u}-\bs\mu_{\bs u'} \|_2^2$, which is bounded by $\tau_2^2\frac{J^\ast \log n}{n}$. By Lemma \ref{LEM:LL}, since $L_{\theta}(\widehat{G}) \leq 1$, given any classifier $\widehat{G}$ in the considered parameter space $\Theta_{J^\ast}^{(1)}$, we have
\begin{eqnarray*}
L^2_{\bs{\theta}_u}(\widehat{G}) + L^2_{\bs{\theta}_{u'}}(\widehat{G})\geq \frac{1}{2}\left(L_{\bs{\theta}_u}(\widehat{G}) + L_{\bs{\theta}_{u'}}(\widehat{G}) \right)^2
\geq   \frac{1}{2}\left( L_{\bs{\theta}}(G^{\ast}_{\bs{\theta}_{u'}}) - \sqrt{\frac{KL\left(P_{\bs{\theta}_{u}},  P_{\bs{\theta}_{u'}}\right)}{2}}\right)^2.
\end{eqnarray*} 
then it's sufficient to show that $ L_{\bs{\theta}}(G^{\ast}_{\bs{\theta}_{u'}}) \geq c\sqrt{\frac{J^\ast\log n}{n}}$ for some $c>\tau_2/\sqrt{2}$. Without loss of generality, we assume that $\bs{z}\sim {N}(\bs{\mu}_1, \bs{\Sigma}_1)$ and  $u_j = u'_j = 1$ when $i = 2, \ldots, m_1$, $u_j = 1 - u'_j = 1$ when $j = m_1+1, \ldots, m_2$, $u_j = 1 - u'_j = 0$ when $j = m_2+1, \ldots, m_3$ and $u_j = u'_j = 0$ when $j = m_3+1, \ldots, J$. Since the optimal decision rule for any $\bs\theta \in \Theta_{J^\ast}^{(1)}$ is given by $G^{\ast}_{\bs{\theta}}(\bs z) = 1 + \mathbb{I}\left\{ -\bs\mu_1^\top(\bs z - \bs\mu_1/2)>0\right\}$, we have 
$$G^{\ast}_{\bs{\theta}_u} = 1+\mathbb{I}\left\{-\tau_2\sqrt{\frac{\log n}{n}}\left(\sum_{j=2}^{m_1}z_j  + \sum_{j=m_1+1}^{m_2}z_j \right)  - \tau_2z_1 + \frac{1}{2}\tau_1^2 + \frac{J^\ast\log n}{n}\right\},$$
and 
$$G^{\ast}_{\bs{\theta}_{u'}} = 1+\mathbb{I}\left\{-\tau_2\sqrt{\frac{\log n}{n}}\left(\sum_{j=2}^{m_1}z_j  + \sum_{j=m_2+1}^{m_3}z_j \right)  - \tau_2z_1 + \frac{1}{2}\tau_1^2 + \frac{J^\ast\log n}{n}\right\}.$$
We can simplify $G^{\ast}_{\bs{\theta}_{u}}$ and $G^{\ast}_{\bs{\theta}_{u'}}$ as $G^{\ast}_{\bs{\theta}_{u}} = \mathbb{I}\left\{T_1>T_2 \right\}$ and $G^{\ast}_{\bs{\theta}_{u'}} = \mathbb{I}\left\{T_1>T_3 \right\}$, where  $T_1 = -\tau_1z_1 - \tau_2\sqrt{\frac{\log n}{n}}\sum_{i=2}^{m_1}z_i + \frac{1}{2}\tau_2^2 + \frac{J^\ast\log n}{2n}$, $T_2=\tau_2\sqrt{\frac{\log n}{n}}\sum_{i=m_1+1}^{m_2}z_i$, $T_3=\tau_2\sqrt{\frac{\log n}{n}}\sum_{i=m_2+1}^{m_3}z_i$. Therefore, 
$\left\{ G^{\ast}_{\bs{\theta}_u} \neq G^{\ast}_{\bs{\theta}_{u'}}\right\} \supseteq  \left\{ T_2\leq T_1 \leq T_3\right\}$, i.e. 
\begin{eqnarray*}
&&{P}_{\bs\theta_{\bs u}}\left( G^{\ast}_{\bs{\theta}_u} \neq G^{\ast}_{\bs{\theta}_{u'}}\right)\geq {P}_{\bs\theta_{\bs u}}\left( T_2\leq T_1 \leq T_3\right)\geq  \frac{1}{2}{P}_{\bs z \sim N(\bs 0, \mathbf{\Sigma})}\left( T_2\leq T_1 \leq T_3\right)\\
&&\geq \frac{1}{2}{P}_{\bs z \sim N(\bs 0, \mathbf{\Sigma})}\left( T_2\leq T_1 \leq T_3, -\tau_2 \sqrt{\frac{J^\ast\log n}{n}} \leq T_2 < T_3 \leq \tau_2 \sqrt{\frac{J^\ast\log n}{n}} \right)\\
&&\geq c(\tau_1, \tau_2)E\left[ (T_3 - T_2)\mathbb{I} \left\{ -\tau_2 \sqrt{\frac{J^{\ast}\log n}{n}} \leq T_2 < T_3 \leq \tau_2 \sqrt{\frac{J^{\ast}\log n}{n}}\right\}\right], 
\end{eqnarray*}
where $c(\tau_1, \tau_2)$ is the lower bound of $f_{T_1}(t)$, for all $t \in \left(-\tau_2 \sqrt{\frac{J^\ast\log n}{n}},  \tau_2 \sqrt{\frac{J^\ast\log n}{n}}\right)$. Since $\sum_{j=1}^\infty \lambda_j < \infty$, we have $T_1 \sim {N}\left(\frac{\tau_1^2}{2} + \frac{\tau_2^2J^\ast\log n}{2n}, \tau_1^2 + \frac{\tau_2^2 \widetilde{c}\log n}{n} \right)$, where $\widetilde{c}=\sum_{j=2}^{m_1} \lambda_j$ is an absolute constant, therefore, we have 
$$c(\tau_1, \tau_2)=\left( 2\pi(\tau_1^2 + \tau_2^2\log n/n)\right)^{-1/2}\exp\left\{ -\frac{\tau_2 \sqrt{{J^\ast\log n}/{n}} +{\tau_1^2}/{2} + {\tau_2^2J^\ast\log n}/{2n} }{2(\tau_1^2 + \tau_2^2\log n/n)}\right\},$$ which goes to infinity for sufficiently small $\tau_1$ and constant $\tau_2$. 

Let $$\mathcal{H}_2 = \left\{T_2:  -\tau_2 \sqrt{\frac{J^{\ast}\log n}{n}} \leq T_2 < 0\right\}  \quad \mbox{and} \quad \mathcal{H}_3 = \left\{T_3: 0 \leq T_3 < \tau_2 \sqrt{\frac{J^{\ast}\log n}{n}} \right\}.$$ 
We have 
\begin{eqnarray*}
&& E\left[ (T_3 - T_2)\mathbb{I} \left\{ -\tau_2 \sqrt{\frac{J^{\ast}\log n}{n}} \leq T_2 < T_3 \leq \tau_2 \sqrt{\frac{J^{\ast}\log n}{n}}\right\}\right]\\
& \geq &  E\left[ (T_3 - T_2)\mathbb{I} \left\{T_2\in \mathcal{H}_2 , T_3\in \mathcal{H}_3\right\}\right] \\
& \geq & 2\tau_2 \sqrt{\frac{J^{\ast}\log n}{n}} E\left( \mathbb{I} \left\{ T_2\in \mathcal{H}_2\right\}\right) E\left( \mathbb{I} \left\{ T_3\in \mathcal{H}_3\right\}\right) \\
& \geq & 2\tau_2 \sqrt{\frac{J^{\ast}\log n}{n}}\times C_{\tau_2} \times C_{\tau_2}\\
& = & 2C_{\tau_2}^2{\tau_2 }\sqrt{\frac{J^{\ast}\log n}{n}}, 
\end{eqnarray*}
where $C_{\tau_2} = \Phi\left(\tau_2\sqrt{\frac{J^{\ast}\log n}{\widetilde{c}n}}\right) - \Phi(0)$, and $\Phi(\cdot)$ is the cumulative density function of standard normal distribution. 
Therefore, $${P}_{\bs\theta_{\bs u}}\left( G^{\ast}_{\bs{\theta}_u}\right) \geq {2\sqrt 2}{C_{\tau_2}^2c(\tau_1, \tau_2)}\frac{\tau_2 }{\sqrt 2}\sqrt{\frac{J^{\ast}\log n}{n}}\geq {\tau_2 }/{\sqrt 2}\sqrt{\frac{J^{\ast}\log n}{n}}$$ for sufficiently small $\tau_1$. 
Finally, by Lemma \ref{LEM:Etheta}, we can conclude that for any $J>4J^{\ast}$,
$$ \inf_{\widehat{G}}\sup_{\bs{\theta} \in \Theta_{J^\ast}^{(1)}} E\left[ R_{\bs\theta}(\widehat{G}) - R_{\bs\theta}(G^{\ast}_{\bs\theta})\right]  \gtrsim \frac{J^{\ast}\log n}{n}.  $$
Finally, plug in the specific $J^\ast$ and the result can be easily derived.
\end{proof}



\subsection{Proof of Theorem \ref{THM:sampling:qda} (upper bound).}

\begin{proof}
  The main proof can be easily derived from the proof of Lemma \ref{LEM:sample} and Theorem \ref{THM:full}, thus it is omitted. {We now proof that under Sobolev balls condition, the rate is determined by $M^\ast$, $J_1^\ast$ and $J_2^\ast$. We denote $f(\cdot) = f(\cdot; \Theta(\nu_1,\nu_2))$ and $g(\cdot) = g(\cdot; \Theta(\nu_1,\nu_2))$ for simplicity.
  
 First, since $f^\ast(M^\ast)\asymp f(M^\ast)$, we have when $M\leq M^\ast$, $\log n/n + f(M)\asymp \log n/n + f^\ast(M) \asymp f^\ast(M)$. For any fixed $M$ such that $M\leq M^\ast$, since there exists an unique $J_1^\ast \leq M$ with $J_1^\ast f^\ast(M)= g^\ast(J_1^\ast) $, we have for all $J \leq M$, 
  $$ J_1^\ast f(M) + g(J_1^\ast) \asymp J_1^\ast f^\ast(M) + g^\ast(J_1^\ast) \leq  J f^\ast(M) + g^\ast(J)\asymp Jf(M) + g(J), $$
  i.e. $J_1^\ast f(M) + g(J_1^\ast) \lesssim  Jf(M) + g(J). $
  
  Next, we have when $M\geq M^\ast$, $\log n/n + f(M)\asymp \log n/n + f^\ast(M) \asymp \log n/n$. For any fixed $M$ such that $M\geq M^\ast$, since there exists an unique $J_2^\ast \leq M$ with $J_2^\ast \log n/n= g^\ast(J_2^\ast) $, we have for all $J \leq M$, 
  $$ J_2^\ast\log n/n  + g(J_2^\ast) \asymp J_2^\ast \log n/n + g^\ast(J_2^\ast) \leq  J \log n/n + g^\ast(J)\asymp J\log n/n + g(J), $$
  i.e. $J_2^\ast \log n/n + g(J_2^\ast) \lesssim  J\log n/n + g(J). $
  By the definition of $J_1^\ast$, $J_2^\ast$ and $M^\ast$ in Lemma \ref{LEM:self-similar}, let $J^\ast = J_1^\ast\mathbb{I}(M\leq M^\ast)  +  J_2^\ast\mathbb{I}(M> M^\ast) $. The result is obtained.
  }

\end{proof}

\subsection{Proof of Theorem \ref{THM:sampling} (lower bound)}
\begin{proof}
In the first case, we consider the difference only on means.
 Let $\bs{e}_j$ be the vector of length $J$ with $j$-th element $1$ and $0$ otherwise, where $J>4J_1^{\ast}$, $J_1^{\ast}$ is defined in Lemma \ref{LEM:self-similar}. For some constants $\tau_1$ and $\tau_2$ which will be specified later, we consider the parameter space
\begin{eqnarray*}
\Theta_{J^\ast}^{(2)} =  \left\{\right.\bs\theta_u =&& \left(1/2, 1/2,  \bs\mu_1, \bs\mu_2, \mathbf{\Sigma}, \mathbf{\Sigma} \right): \bs{\Sigma} = \text{diag}\left( \lambda_1,\ldots,  \lambda_J\right), \\ 
&&\left.  \bs\mu_1 = \tau_1 \bs{e}_1 +  \tau_2 \sqrt{ f_1^2(M)}  \sum_{j=3}^{J-1}u_j \bs{e}_{j}, \bs{u}=(u_1,\ldots,u_J) \in \mathcal{A}_{J, J_1^{\ast}} \right\}
\end{eqnarray*}
where $\mathcal{A}_{J, J^{\ast}}$ is defined in Lemma \ref{LEM:A}. Thus, we have $\Theta_{J_1^\ast}^{(2)} \subseteq \Theta$.
For any $\bs u \in \mathcal{A}_{J, J^{\ast}}$, define $\bs \mu_u = \tau_1 \bs{e}_1 +  \tau_2 \sqrt{ f_1^2(M)}$, then the KL divergence between ${P}_{\bs\theta_{\bs u}}=N(\bs\mu_{\bs u}, \mathbf{\Sigma})$ and ${P}_{\bs\theta_{\bs u'}}=N(\bs\mu_{\bs u'}, \mathbf{\Sigma})$ is $\frac{1}{2}\|\bs\mu_{\bs u}-\bs\mu_{\bs u'} \|_2^2$, which is bounded by $\tau_2^2 J_1^\ast f_1^2(M)$. By Lemma \ref{LEM:LL}, given any classifier $\widehat{G}$ in the considered parameter space $\Theta_{J^\ast_1}^{(2)}$, we have
\begin{eqnarray*}
L^2_{\bs{\theta}_u}(\widehat{G}) + L^2_{\bs{\theta}_{u'}}(\widehat{G})\geq \frac{1}{2}\left(L_{\bs{\theta}_u}(\widehat{G}) + L_{\bs{\theta}_{u'}}(\widehat{G}) \right)^2
\geq   \frac{1}{2}\left( L_{\bs{\theta}}(G^{\ast}_{\bs{\theta}_{u'}}) - \sqrt{\frac{KL\left(P_{\bs{\theta}_{u}},  P_{\bs{\theta}_{u'}}\right)}{2}}\right)^2.
\end{eqnarray*} 
then it's sufficient to show that $ L_{\bs{\theta}}(G^{\ast}_{\bs{\theta}_{u'}}) \geq c\sqrt{J_1^\ast  f_1^2(M)}$ for some $c>\tau_2/\sqrt{2}$. Without loss of generality, we assume that $\bs{z}\sim {N}(\bs{\mu}_1, \bs{\Sigma}_1)$ and  $u_j = u'_j = 1$ when $i = 2, \ldots, m_1$, $u_j = 1 - u'_j = 1$ when $j = m_1+1, \ldots, m_2$, $u_j = 1 - u'_j = 0$ when $j = m_2+1, \ldots, m_3$ and $u_j = u'_j = 0$ when $j = m_3+1, \ldots, J$. Since the optimal decision rule for any $\bs\theta \in \Theta_{J_1^\ast }^{(2)}$ is given by $G^{\ast}_{\bs{\theta}}(\bs z) = 1 + \mathbb{I}\left\{ -\bs\mu_1^\top(\bs z - \bs\mu_1/2)>0\right\}$, we have 
$$G^{\ast}_{\bs{\theta}_u} = 1+\mathbb{I}\left\{-\tau_2\sqrt{f_1^2(M)}\left(\sum_{j=2}^{m_1}z_j  + \sum_{j=m_1+1}^{m_2}z_j \right)  - \tau_2z_1 + \frac{1}{2}\tau_1^2 + J_1^\ast  f_1^2(M)\right\},$$
and 
$$G^{\ast}_{\bs{\theta}_{u'}} = 1+\mathbb{I}\left\{-\tau_2\sqrt{f_1^2(M)}\left(\sum_{j=2}^{m_1}z_j  + \sum_{j=m_2+1}^{m_3}z_j \right)  - \tau_2z_1 + \frac{1}{2}\tau_1^2 + J_1^\ast  f_1^2(M)\right\}.$$
We can simplify $G^{\ast}_{\bs{\theta}_{u}}$ and $G^{\ast}_{\bs{\theta}_{u'}}$ as $G^{\ast}_{\bs{\theta}_{u}} = \mathbb{I}\left\{T_1>T_2 \right\}$ and $G^{\ast}_{\bs{\theta}_{u'}} = \mathbb{I}\left\{T_1>T_3 \right\}$, where we let $T_1 = -\tau_1z_1 - \tau_2\sqrt{f_1^2(M)}\sum_{i=2}^{m_1}z_i + \frac{1}{2}\tau_2^2 + \frac{J_1^\ast  f_1^2(M)}{2}$, $T_2=\tau_2\sqrt{f_1^2(M)}\sum_{i=m_1+1}^{m_2}z_i$, $T_3=\tau_2\sqrt{f_1^2(M)}\sum_{i=m_2+1}^{m_3}z_i$. Therefore, 
$\left\{ G^{\ast}_{\bs{\theta}_u} \neq G^{\ast}_{\bs{\theta}_{u'}}\right\} \supseteq  \left\{ T_2\leq T_1 \leq T_3\right\}$, i.e. 
\begin{eqnarray*}
&&{P}_{\bs\theta_{\bs u}}\left( G^{\ast}_{\bs{\theta}_u} \neq G^{\ast}_{\bs{\theta}_{u'}}\right)\geq {P}_{\bs\theta_{\bs u}}\left( T_2\leq T_1 \leq T_3\right)\geq  \frac{1}{2}{P}_{\bs z \sim N(\bs 0, \mathbf{\Sigma})}\left( T_2\leq T_1 \leq T_3\right)\\
&&\geq \frac{1}{2}{P}_{\bs z \sim N(\bs 0, \mathbf{\Sigma})}\left( T_2\leq T_1 \leq T_3, -\tau_2 \sqrt{J_1^\ast  f_1^2(M)} \leq T_2 < T_3 \leq \tau_2 \sqrt{J_1^\ast  f_1^2(M)} \right)\\
&&\geq c(\tau_1, \tau_2)E\left[ (T_3 - T_2)\mathbb{I} \left\{ -\tau_2 \sqrt{J_1^\ast  f_1^2(M)} \leq T_2 < T_3 \leq \tau_2 \sqrt{J_1^\ast  f_1^2(M)}\right\}\right], 
\end{eqnarray*}
where $c(\tau_1, \tau_2)$ is the lower bound of $f_{T_1}(t)$, for all $t \in \left(-\tau_2 \sqrt{J_1^\ast  f_1^2(M)},  \tau_2 \sqrt{J_1^\ast  f_1^2(M)}\right)$. Since $\sum_{j=1}^\infty \lambda_j < \infty$, we have $T_1 \sim N\left(\frac{\tau_1^2}{2} + \frac{\tau_2^2J_1^\ast  f_1^2(M)}{2}, \tau_1^2 + {\tau_2^2 \widetilde{c}f_1^2(M)} \right)$, where $\widetilde{c}=\sum_{j=2}^{m_1} \lambda_j$ is an absolute constant, therefore, we have 
$$c(\tau_1, \tau_2)=\left( 2\pi(\tau_1^2 + \tau_2^2f_1^2(M))\right)^{-1/2}\exp\left\{ -\frac{\tau_2 \sqrt{J_1^\ast  f_1^2(M)} +{\tau_1^2}/{2} + {\tau_2^2J_1^\ast  f_1^2(M)}/{2} }{2(\tau_1^2 + \tau_2^2f_1^2(M))}\right\},$$ which goes to infinity for sufficiently small $\tau_1$ and constant $\tau_2$. 

With a  slight abuse of notations, let $$\mathcal{H}_2 = \left\{T_2:  -\tau_2 \sqrt{J_1^\ast  f_1^2(M)} \leq T_2 <  0\right\} \quad \mbox{and} \quad  \mathcal{H}_3 = \left\{T_3: 0 \leq T_3 < \tau_2 \sqrt{J_1^\ast  f_1^2(M)} \right\}.$$ 
We have 
\begin{eqnarray*}
&& E\left[ (T_3 - T_2)\mathbb{I} \left\{ -\tau_2 \sqrt{J_1^\ast  f_1^2(M)} \leq T_2 < T_3 \leq \tau_2 \sqrt{J_1^\ast  f_1^2(M)}\right\}\right]\\
& \geq &  E\left[ (T_3 - T_2)\mathbb{I} \left\{T_2\in \mathcal{H}_2 , T_3\in \mathcal{H}_3\right\}\right] \\
& \geq & 2\tau_2 \sqrt{J_1^\ast  f_1^2(M)} E\left( \mathbb{I} \left\{ T_2\in \mathcal{H}_2\right\}\right) E\left( \mathbb{I} \left\{ T_3\in \mathcal{H}_3\right\}\right) \\
& \geq & 2\tau_2 \sqrt{J_1^\ast  f_1^2(M)}\times C_{\tau_2} \times C_{\tau_2}\\
& = & 2C_{\tau_2}^2{\tau_2 }\sqrt{J_1^\ast  f_1^2(M)}. 
\end{eqnarray*}
Therefore, ${P}_{\bs\theta_{\bs u}}\left( G^{\ast}_{\bs{\theta}_u}\right) \geq {2\sqrt 2}{C_{\tau_2}^2c(\tau_1, \tau_2)}\ \tau_2 /\sqrt {2}\sqrt{J_1^\ast  f_1^2(M)}\geq {\tau_2 }/{\sqrt 2}\sqrt{J_1^\ast  f_1^2(M)}$ for sufficiently small $\tau_1$. 
Finally, by Lemma \ref{LEM:Etheta}, we can conclude that for any $J>4J_1^{\ast}$,
$$ \inf_{\widehat{G}}\sup_{\bs{\theta} \in \Theta_{J_1^\ast }^{(2)}} E\left[ R_{\bs\theta}(\widehat{G}) - R_{\bs\theta}(G^{\ast}_{\bs\theta})\right]  \gtrsim J_1^\ast  f_1^2(M).  $$

In the second case, we consider the difference only on variances.
For any $3\leq \ell_1<\ell_2<\ell_3<J$, define the following partial summations based on the definitions in the proof of Lemma \ref{LEM:fq}.
$$A_k(\ell_1,\ell_2) = \sum_{j=\ell_1}^{\ell_2} \left[(\epsilon_j-1)^k + k (\epsilon_j-1)^{k-1}\frac{(\mu_{1j}-\mu_{2j})^2}{\lambda_{j}^{(2)}} +  (\epsilon_j-1)^{k-2}\frac{(\mu_{1j}-\mu_{2j})^2}{\lambda_{j}^{(2)}}\right],
$$
$\omega_1^2(\ell_1,\ell_2) =A_3^2(\ell_1,\ell_2) A_2^{-3}(\ell_1,\ell_2) $
,
$\omega_2(\ell_1,\ell_2) = A_4(\ell_1,\ell_2) A_2^{-2}(\ell_1,\ell_2) $,\\
 $\sigma_{\chi}(\ell_1,\ell_2)=\frac{\sqrt{2}}{\omega_1(\ell_1,\ell_2) - \sqrt{\omega_1^2(\ell_1,\ell_2) - \omega_2(\ell_1,\ell_2)}}$,
 $ \sigma_{q} (\ell_1,\ell_2) =  \sqrt{ 2\sum_{j=\ell_1}^{\ell_2} \left[(\epsilon_j-1)^2 +2 \frac{\Delta\mu_j^2}{\lambda_{j}^{(2)}} \epsilon_j\right]}$,\\
 $D_{q}(\ell_1,\ell_2) = \frac{\mu_{\chi}(\ell_1,\ell_2)}{\sigma_{\chi}(\ell_1,\ell_2)} \sigma_{q} (\ell_1,\ell_2)- \mu_{q}(\ell_1,\ell_2)$,
  $U_{q}(\ell_1,\ell_2) = \frac{\sqrt{\omega_1^2 (\ell_1,\ell_2)- \omega_2(\ell_1,\ell_2)}}{(\omega_1(\ell_1,\ell_2) - \sqrt{\omega_1^2(\ell_1,\ell_2) - \omega_2(\ell_1,\ell_2)})^3}$,
 $\gamma (\ell_1,\ell_2)= \frac{\omega_1(\ell_1,\ell_2)-3\sqrt{\omega_1^2(\ell_1,\ell_2)- \omega_2(\ell_1,\ell_2)}}{(\omega_1(\ell_1,\ell_2) - \sqrt{\omega_1^2(\ell_1,\ell_2) - \omega_2(\ell_1,\ell_2)})^3} \mathbb{I}(\omega_1(\ell_1,\ell_2)^2>\omega_2(\ell_1,\ell_2))+\omega_1^{-2}(\ell_1,\ell_2))$\\ $\times\mathbb{I}(\omega_1^2(\ell_1,\ell_2) \leq \omega_2(\ell_1,\ell_2))$.

Define  a particular parameter space for  $ \{ \lambda_j\}_{j=1}^J $,
\begin{eqnarray*}
\mathcal{B}_J &=& \left\{ \left\{ \lambda_j \right\}_{j=1}^J :   R_{3,\ell_1}\exp \left\{\frac{- U_{q}( 3, \ell_1) - \sum_{j=3}^{\ell_1} \lambda_j}{2} \right\}  > c_0^{-1},  \ell_1< J_1^\ast  /2, \right.\\
&&\left. \ell_2-\ell_1,  \ell_3-\ell_2\geq  J_1^\ast  /2 \right\} ,
\end{eqnarray*}
where
$R _{\ell_1,\ell_2}=\frac{\sigma_{\chi} (\ell_1, \ell_2) \sigma_{q}^{-1} (\ell_1, \ell_2)}{2^{{\gamma(\ell_1, \ell_2)}/{2}-1}\Gamma({\gamma(\ell_1, \ell_2)}/{2}-1)}$, $H_{\ell_1,\ell_2}= U_{q} (\ell_1, \ell_2)   +\sigma_{\chi} (\ell_1, \ell_2)\sigma_{q}(\ell_1, \ell_2)^{-1}(D_{q}( \ell_1, \ell_2) $ $+ \sum_{j=\ell_1}^{\ell_2} \lambda_j) $ and
\begin{eqnarray}\label{EQ:c_0}
{
c_0=  \inf_{3\leq \ell_1, \ell_2, \ell_3\leq J }R_{\ell_1,\ell_2}R_{\ell_2,\ell_3} \exp \left\{\frac{-H_{\ell_1,\ell_2} - H_{\ell_2,\ell_3}- \sigma_{\chi} (\ell_2, \ell_3)\sigma_{q}^{-1}( \ell_2, \ell_3)( \sum_{j=\ell_2}^{\ell_3} \lambda_j + 1)  }{2} \right\} }\nonumber\\
\end{eqnarray}
is a nonzero constant. 
Let $\bs{E}_{j}$ be the diagonal matrix with $j$-th diagonal element $1$ and $0$ otherwise.  Let  $\tau_1$, $\widetilde{\tau}_1$, $\widetilde{\tau}_2$ be some constants small enough, which are specified later.
For any $J>4J_1^{\ast}$, where $J_1^{\ast}$ is defined in Lemma \ref{LEM:self-similar}, we consider the parameter space
\begin{eqnarray*}
\Theta_{J_1^\ast }^{(3)} = \left\{\bs\theta_u = \left(1/2, 1/2,  \tau_1 \bs{e}_1 + \widetilde{\tau}_1 \bs{e}_2 ,  - \tau_1 \bs{e}_1 - \widetilde{\tau}_1 \bs{e}_2, \bs{\Sigma}_1^u, \bs{\Sigma}_2\right):  \bs{\Sigma}_2 = \text{diag}\left( \lambda_1,\ldots,  \lambda_J\right),  \right.\\
\left.  \left\{ \lambda_j \right\}_{1\leq j\leq J} \in \mathcal{B}_J, \left( \bs{\Sigma}_1^u\right)^{-1} = \bs{\Sigma}_2^{-1} +  \tau_2 { f_2(M)}  \sum_{j=3}^{J-1}u_j \bs{E}_{j} + \widetilde{\tau}_2\bs{E}_{J}, \bs{u}=(u_1,\ldots,u_J) \in \mathcal{A}_{J, J^{\ast}_1}\right\} ,
\end{eqnarray*}
where $\mathcal{A}_{J, J^{\ast}}$ is defined in Lemma \ref{LEM:A}. Thus, we have $\Theta_{J_1^\ast }^{(3)} \subseteq \Theta$.


We first calculate KL divergence between two distributions for $u$ and $u'$. By Lemma \ref{LEM:A}, we have
\begin{eqnarray*}
KL( {P}_{\bs{\theta}_{u}},  {P}_{\bs{\theta}_{u'}}) &=& \frac{1}{2} \left\{ \log \frac{|\bs{\Sigma}_1^{u'}|}{|\bs{\Sigma}_1^{u}|} - J +tr\left(\left( \bs{\Sigma}_1^{u'}\right)^{-1}\bs{\Sigma}_1^{u}\right)\right\}\\
 &\leq & \frac{1}{4} \sum_{j=3}^{J-1}  \left\{\tau_2^2  f_2^2(M)(u_j - u'_j)^2 + o\left( f_2^2(M)\right)\right\} \\
&\leq & \frac{\tau_2^2}{4} J_1^\ast  f_2^2(M) + o\left(J_1^\ast  f_2^2(M)\right)
\leq  \frac{\tau_2^2}{2} J_1^\ast  f_2^2(M). \\
\end{eqnarray*}
By Lemmas \ref{LEM:LL}, given any classifier $\widehat{G}$ in the considered parameter space $\Theta_{J^\ast_1}^{(3)}$, we have
\begin{eqnarray*}
L^2_{\bs{\theta}_u}(\widehat{G}) + L^2_{\bs{\theta}_{u'}}(\widehat{G})\geq \frac{1}{2}\left(L_{\bs{\theta}_u}(\widehat{G}) + L_{\bs{\theta}_{u'}}(\widehat{G}) \right)^2
\geq   \frac{1}{2}\left( L_{\bs{\theta}}(G^{\ast}_{\bs{\theta}_{u'}}) - \sqrt{\frac{KL\left(P_{\bs{\theta}_{u}},  P_{\bs{\theta}_{u'}}\right)}{2}}\right)^2.
\end{eqnarray*} 
Since $KL( {P}_{\bs{\theta}_{u}},  {P}_{\bs{\theta}_{u'}}) \leq \frac{\tau_2^2}{2} J_1^\ast  f_2^2(M)$, it's sufficient to show that $L_{\bs{\theta}}(G^{\ast}_{\bs{\theta}_{u'}}) \geq c\sqrt{J_1^\ast}  f_2(M)$  for some $c>\tau_2/2$.
Now, we need to calculate
\begin{eqnarray*}
 &&(\bs{z} - \bs{\mu}_1)^{\top}\mathbf{D}^u(\bs{z} - \bs{\mu}_1) - 2\bs{\delta}^{\top}\mathbf{\Omega}_2^u(\bs{z} - \bar{\bs{\mu}}) - \log(|\mathbf{\Sigma}_1^u|/|\mathbf{\Sigma}_2|) \\
 &=& -\tau_2 {f_2(M)} \sum_{j=3}^{J-1}  u_j z_j^2+  \frac{4\tau_1}{\lambda_1}z_1 + \frac{4\widetilde{\tau}_1}{\lambda_2}z_2 -\widetilde{\tau}_2z_J^2 + \log \left(1+\lambda_{J}\widetilde{\tau}_2\right) \\
 &&+\sum_{j=3}^{J-1} \log \left(1+\lambda_{j}\tau_2{f_2(M)} u_j\right) \\
 & = & -\tau_2{f_2(M)} \sum_{j=3}^{J-1}  u_j (z_j^2 - \lambda_j) +  \frac{4\tau_1}{\lambda_1}z_1 + \frac{4\widetilde{\tau}_1}{\lambda_2}z_2 -\widetilde{\tau}_2(z_J^2-\lambda_J) - \frac{1}{2}\widetilde{\tau}_2^2 \lambda_J^2 \\
  &&-\frac{1}{2}\sum_{j=3}^{J-1}  \lambda_{j}^2\tau_2^2 f_2^2(M) u_j +  o\left( f_2^2(M)\right)\\
  & = & -\tau_2{f_2(M)} \sum_{j=3}^{J-1}  u_j (z_j^2 - \lambda_j) +  \frac{4\tau_1}{\lambda_1}z_1 + \frac{4\widetilde{\tau}_1}{\lambda_2}z_2 -\widetilde{\tau}_2(z_J^2-\lambda_J) - \frac{1}{2}\widetilde{\tau}_2^2\lambda_J^2 \\
   &&+O\left(\tau_2^2J_1^\ast  f_2^2(M)\right) +  o\left(J_1^\ast  f_2^2(M)\right).
\end{eqnarray*}
Without loss of generality, we assume that $\bs{z}\sim N(\bs{\mu}_1, \bs{\Sigma}_1)$ and  $u_j = u'_j = 1$ when $i = 3, \ldots, m_1$, $u_j = 1 - u'_j = 1$ when $j = m_1+1, \ldots, m_2$, $u_j = 1 - u'_j = 0$ when $j = m_2+1, \ldots, m_3$ and $u_j = u'_j = 0$ when $j = m_3+1, \ldots, J$.
Then for $G^{\ast}_{\bs{\theta}_u}$, we have decision function
\begin{eqnarray*}
  Q_u(\bs{z}) &=&-\tau_2{f_2(M)} \left( \sum_{j=3}^{m_1} (z_j^2 - \lambda_j) + \sum_{j=m_1+1}^{m_2} (z_j^2 - \lambda_j) \right) +  \frac{4\tau_1}{\lambda_1}z_1 \\
 &&+ \frac{4\widetilde{\tau}_1}{\lambda_2}z_2 -\widetilde{\tau}_2(z_J^2-\lambda_J) - \frac{1}{2}\widetilde{\tau}_2^2\lambda_J^2 +O\left(\tau_2^2J_1^\ast  f_2^2(M)\right) +  o\left(J_1^\ast  f_2^2(M)\right),
 \end{eqnarray*}
and for $G^{\ast}_{\bs{\theta}_{u'}}$, we have decision function
\begin{eqnarray*}
  Q_{u'}(\bs{z}) &=& -\tau_2{f_2(M)} \left[\sum_{j=3}^{m_1} (z_j^2 - \lambda_j) + \sum_{j=m_2+1}^{m_3} (z_j^2 - \lambda_j) \right] +  \frac{4\tau_1}{\lambda_1}z_1\\
 && + \frac{4\widetilde{\tau}_1}{\lambda_2}z_2 -\widetilde{\tau}_2(z_J^2-\lambda_J) - \frac{1}{2}\widetilde{\tau}_2^2\lambda_J^2+O\left(\tau_2^2J_1^\ast  f_2^2(M)\right) +  o\left(J_1^\ast  f_2^2(M)\right).
 \end{eqnarray*}
Let $T_1 = -\tau_2{f_2(M)} \sum_{j=3}^{m_1}  (z_j^2 - \lambda_j)  +  \frac{4\tau_1}{\lambda_1}z_1 + \frac{4\widetilde{\tau}_1}{\lambda_2}z_2 -\widetilde{\tau}_2(z_J^2-\lambda_J) - \frac{1}{2}\widetilde{\tau}_2^2\lambda_J^2 +O\left(\tau_2^2 J_1^\ast  f_2^2(M)\right)$, $T_2 = -\tau_2{f_2(M)} \sum_{j=m_1+1}^{m_2} (z_j^2 - \lambda_j) $, $T_3 = -\tau_2{f_2(M)} \sum_{j=m_2+1}^{m_3} (z_j^2 - \lambda_j) $, then
\begin{eqnarray*}
\left\{ G^{\ast}_{\bs{\theta}_u} \neq G^{\ast}_{\bs{\theta}_{u'}}\right\} &\supseteq &  \left\{ T_3+o\left(J_1^\ast  f_2^2(M)\right)< T_1 \leq T_2+ o\left(J_1^\ast  f_2^2(M)\right), \right.\\
 && \left.T_2+ o\left(J_1^\ast  f_2^2(M)\right)< T_1 \leq T_3 + o\left(J_1^\ast  f_2^2(M)\right)\right\}.
\end{eqnarray*}
Thus, $ L_{\bs{\theta}}(G^{\ast}_{\bs{\theta}_{u'}})  \geq \frac{1}{2} P_{\bs{z}\sim N(\bs{\mu}_1, \bs{\Sigma}_2)}\left( T_2 < T_1 \leq T_3\right)+o\left(J_1^\ast  f_2^2(M)\right)$.

Notice that $Ez_j = 0$ and $Ez_j^2 = \lambda_j$ for all $j \geq 3$.  Since $\sum_{j=m_1+1}^{m_2} z_j^2 = \sum_{j=m_1+1}^{m_2} \lambda_j h_j^2 $ and $\sum_{j=m_2+1}^{m_3} z_j^2 = \sum_{j=m_2+1}^{m_3} \lambda_j h_j^2 $, where $h_j\sim N(0, 1)$, following the same procedure in the proof of Lemma \ref{LEM:fq}, we have $A_k^{\prime}=\sum_{j=m_1+1}^{m_2}\lambda_j^k$, $\mu_q^{\prime}=\sum_{j=m_1+1}^{m_2}\lambda_j$ and $\sigma_q^{\prime}=\sqrt{2\sum_{j=m_1+1}^{m_2}\lambda_j^2}$, which are all constants. Without loss of generality, we have $\omega_1^2>\omega_2$. Using the lower bound of probability density, with constants $\sigma_{\chi}/\sigma_q$, $D_q$ and $U_q$, 
we have
\begin{eqnarray*}
&&E\left[ (T_3 - T_2)\mathbb{I} \left\{ -\tau_2 \sqrt{J_1^\ast  f_2^2(M)} \leq T_2 < T_3 < \tau_2 \sqrt{J_1^\ast  f_2^2(M)}\right\}\right] \\
& \geq &  E\left[ (T_3 - T_2)\mathbb{I} \left\{ -\tau_2 \sqrt{J_1^\ast  f_2^2(M)} \leq T_2 < 0, 0< T_3 < \tau_2 \sqrt{J_1^\ast  f_2^2(M)}\right\}\right] \\
& \geq & 2\tau_2 \sqrt{J_1^\ast  f_2^2(M)} P\left(  -\tau_2 \sqrt{J_1^\ast  f_2^2(M)} \leq T_2 < 0\right) P\left(   0< T_3 < \tau_2 \sqrt{J_1^\ast  f_2^2(M)}\right) \\
& = & 2\tau_2 \sqrt{J_1^\ast  f_2^2(M)} P\left(  \sum_{j=m_1+1}^{m_2} \lambda_j \leq \sum_{j=m_1+1}^{m_2} z_j^2  < \sqrt{J_1^\ast }+ \sum_{j=m_1+1}^{m_2} \lambda_j\right)\\
&&\times P\left(  \left(-\sqrt{J_1^\ast } + \sum_{j=m_2+1}^{m_3} \lambda_j \right)\vee 0 < \sum_{j=m_2+1}^{m_3} z_j^2  < \sum_{j=m_2+1}^{m_3} \lambda_j\right) \\
& \geq & 2\tau_2 \sqrt{J_1^\ast  f_2^2(M)} P\left(  \sum_{j=m_1+1}^{m_2} \lambda_j \leq \sum_{j=m_1+1}^{m_2} z_j^2  < 1+ \sum_{j=m_1+1}^{m_2} \lambda_j\right)\\
&&\times P\left(  \left(-1+ \sum_{j=m_2+1}^{m_3} \lambda_j \right)\vee 0 < \sum_{j=m_2+1}^{m_3} z_j^2  < \sum_{j=m_2+1}^{m_3} \lambda_j\right) \\
& \geq & 2c_0\tau_2 \sqrt{J_1^\ast  f_2^2(M)}.
\end{eqnarray*}
Similar as $T_2$ and $T_3$,  the density function for $ -\tau_2{f_2(M)} \sum_{j=3}^{m_1}  (z_j^2 - \lambda_j)$ has the lower bound
$$-\tau_2{f_2(M)}  \frac{\sigma_{\chi}}{2^{\frac{\gamma}{2}}\sigma_{q}\Gamma({\gamma}/{2}-1)}\exp \left(\frac{\tau_2\sqrt{ f_2^2(M)} K_q - U_{q} - \sum_{j=3}^{m_1} \lambda_j}{2} \right),$$
 and $\frac{4\tau_1}{\lambda_1}z_1 \sim N\left(\frac{4\tau_1^2}{\lambda_1}, \frac{16\tau_1^2}{\lambda_1}\right)$, $\frac{4\widetilde{\tau}_1}{\lambda_2}z_2  \sim N\left(\frac{4\widetilde{\tau}_1^2}{\lambda_2}, \frac{16\widetilde{\tau}_1^2}{\lambda_2}\right)$.  Since $\lambda_J$ is decreasing with $J$, $-\widetilde{\tau}_2(z_J^2-\lambda_J) - \frac{1}{2}\widetilde{\tau}_2^2\lambda_J^2=o_p(1)$ and $O\left(\tau_2^2J_1^\ast  f_2^2(M)\right)=o\left(J_1^\ast  f_2^2(M)\right)$ when  $\widetilde{\tau}_2$ and $\tau_2$ are sufficiently small,  $T_1$ is the sum of $-\tau_2{f_2(M)} \sum_{j=3}^{m_1}  (z_j^2 - \lambda_j)$ and  a normal random variable. By simple calculation, we have
\begin{eqnarray*}
f_{T_1}(t) &\geq & c_1c_2\exp\left[ -\tau_2\sqrt{ f_2^2(M)} \frac{\sigma_{\chi}}{\sigma_{q}}\left\{ t -  \left(\frac{12\tau_1^2}{\lambda_1}+\frac{12\widetilde{\tau}_1^2}{\lambda_2}\right) - D_q \right\}\right],
\end{eqnarray*}
where $c_1 = \frac{\sigma_\chi \sigma_q^{-1}}{2^{\frac{\gamma}{2}}\Gamma({\gamma}/{2}-1)}$ and $c_2=\exp \left(\frac{- U_{q} - \sum_{j=3}^{m_1} \lambda_j}{2} \right) .$
When $|t| < \tau_2\sqrt{J_1^\ast  f_2^2(M)} $, and $\widetilde{\tau}_2$ is sufficiently small,  we have
\begin{eqnarray*}
&&\inf _{t:|t| < \tau_2\sqrt{J_1^\ast  f_2^2(M)} } f(t)  \\
& \geq & c_1c_2\exp\left\{ -\tau_2^2 f_2^2(M) \frac{\sigma_{\chi}}{\sigma_{q}} \sqrt{J_1^\ast }  - \tau_2\sqrt{ f_2^2(M)} \frac{\sigma_{\chi}}{\sigma_{q}}\left(\frac{12\tau_1^2}{\lambda_1}+\frac{12\widetilde{\tau}_1^2}{\lambda_2}\right) -\tau_2\sqrt{ f_2^2(M)}\frac{\sigma_{\chi}}{\sigma_q} D_q \right\}\\
 & \geq & c_1c_2 C_n\left(\tau_1 , \widetilde{\tau}_1, {\tau}_2\right).
\end{eqnarray*}
When $n\rightarrow \infty$ and $\tau_1$, $\widetilde{\tau}_1$, $ {\tau}_2 \rightarrow 0$, $C_n\left(\tau_1 , \widetilde{\tau}_1, {\tau}_2\right)\rightarrow 1$, and given parameter space $\Theta_{J_1^\ast }^{(3)}$, we have $2c_1c_2 \geq c_0^{-1}$, where $c_0$ is given in (\ref{EQ:c_0}).
 Note that
\begin{eqnarray*}
&& P_{\bs{z}\sim N(\bs{\mu}_1, \bs{\Sigma}_2)} \left( T_2 < T_1 \leq T_3 \right)  \\
&\geq &P_{\bs{z}\sim N(\bs{\mu}_1, \bs{\Sigma}_2)}\left( T_2 < T_1 \leq T_3, -\tau_2 \sqrt{J_1^\ast  f_2^2(M)} \leq T_2 < T_3 < \tau_2 \sqrt{J_1^\ast  f_2^2(M)} \right)  \\
&=& E_{T_2}\left[ \mathbb{I} \left\{ -\sqrt{J_1^\ast  f_2^2(M)} \leq T_2 < T_3 < \tau_2 \sqrt{J_1^\ast  f_2^2(M)}\right\} \int_{T_2}^{T_3}f_{T_1}(t)dt \right] \\
&\geq &\frac{1}{2c_0} E\left[ (T_3 - T_2)\mathbb{I} \left\{ -\tau_2 \sqrt{J_1^\ast  f_2^2(M)}\leq T_2 < T_3 < \tau_2 \sqrt{J_1^\ast  f_2^2(M)}\right\}\right]   \\
& \geq & \frac{c_0}{2c_0}\tau_2 \sqrt{J_1^\ast  f_2^2(M)}  =  \frac{\tau_2}{2} \sqrt{J_1^\ast  f_2^2(M)}.
\end{eqnarray*}

By Lemma \ref{LEM:Etheta}, we can conclude that for any $J>4J_1^\ast $,
$$ \inf_{\widehat{G}}\sup_{\bs{\theta} \in \Theta_{J_1^\ast }^{(3)}} E\left[ R_{\bs\theta}(\widehat{G}) - R_{\bs\theta}(G^{\ast}_{\bs\theta})\right]  \gtrsim J_1^\ast  f_2^2(M).  $$

Therefore, we conclude that 
$$ \inf_{\widehat{G}}\sup_{\bs{\theta} \in \Theta_{J_1^\ast }^{(2)}\cup \Theta_{J_1^\ast }^{(3)}} E\left[ R_{\bs\theta}(\widehat{G}) - R_{\bs\theta}(G^{\ast}_{\bs\theta})\right]  \gtrsim J_1^\ast  \left(f_1^2(M) + f_2^2(M)\right).  $$

Finally, follow the proof of Theorem \ref{THM:full} (lower bound), we have
$$ \inf_{\widehat{G}}\sup_{\bs{\theta} \in \Theta_{J^\ast}^{(1)}} E\left[ R_{\bs\theta}(\widehat{G}) - R_{\bs\theta}(G^{\ast}_{\bs\theta})\right]  \gtrsim \frac{J^{\ast}\log n}{n}.  $$
Plug in the specific $J_2^\ast$ and the lower bound can be derived.
By the definition of $\Theta$, the result is easily obtained.  
\end{proof}


\subsection{Technical lemmas for deep neural networks}\label{SEC: DNN}

 We introduce the complexity measures of a given function class. Let $\mathcal{F}$ be a given class of real valued functions on $\mathcal{C}$. 
 \subsubsection{Covering number}
 Let $\kappa>0$ and $\| f\|_\infty = \sup_{\bs z \in \mathcal{C} }|f(\bs z)|$. A subset $\left\{ f_k \in \mathcal{F}\right\}_{k\geq 1}$ is called a $\kappa$-covering set of $\mathcal{F}$ with respect to $\| f\|_\infty$, if for all $f\in \mathcal{F}$, there exists an $f_k$ such that $\| f_k - f\|_\infty \leq \kappa$. The $\kappa$-covering number of $\mathcal{F}$ with respect to $\| f\|_\infty$ is defined by
 $$\mathcal{N}(\kappa, \mathcal{F}, \|\cdot \|_\infty) = \inf\left\{ N\in \mathbb{N}: \exists f_1, \ldots, f_N, \textit{s.t.} \mathcal{F} \subset \bigcup_{k=1}^{N}\left\{ f\in \mathcal{F} : \| f_k-f\|_\infty \leq \kappa\right\}  \right\}.$$
 \subsubsection{Bracketing entropy}
 A collection of pairs $\left\{\left( f_k^L, f_k^U\right)\in \mathcal{F}\times \mathcal{F}\right\}_{k\geq 1}$ is called a $\kappa$-bracketing set of $\mathcal{F}$ with respect to $\| f\|_\infty$, if $\| f_k^L- f_k^U\|_\infty\leq \kappa$ and for all $f\in \mathcal{F}$, there exists a pair $\left( f_k^L, f_k^U\right)$ such that $f_k^L\leq f\leq f_k^U$. The cardinality of the minimal $\kappa$-bracketing set with respect to $\| f\|_\infty$ is called the $\kappa$-bracketing number, which is denoted by $\mathcal{N}_B(\kappa, \mathcal{F}, \|\cdot \|_\infty)$. Define $\kappa$-bracketing entropy as  ${H}_B(\kappa, \mathcal{F}, \|\cdot \|_\infty) = \log \mathcal{N}_B(\kappa, \mathcal{F}, \|\cdot \|_\infty)$. 
 
 Given any $\kappa>0$, it is known that 
 $$\log \mathcal{N}(\kappa, \mathcal{F}, \|\cdot \|_\infty) \leq {H}_B(\kappa, \mathcal{F}, \|\cdot \|_\infty) \leq \log \mathcal{N}(\kappa/2, \mathcal{F}, \|\cdot \|_\infty) .$$
 
 We first give the following lemma on the upper bound of the $\kappa$-entropy of DNN space.

 \begin{lemma}\label{entropy}(Proposition 1 of \cite{Kim:etal:21})
For any $\kappa>0$, 
$$\log \mathcal{N}(\kappa, \mathcal{F}\left( L, J, \bs p, s, B\right), \|\cdot \|_\infty) \leq 2L(s+1)\log\left( \kappa^{-1}(L+1)(N+1)(B\vee 1)\right),$$
where $B\vee 1 = \max\left\{B, 1\right\}$, $N = \max_{0\leq \ell \leq L} p_\ell$. 
\end{lemma}

\subsection{Approximation of density ratio function}

Hereafter, for technical convenience, we relabel the data as $Y(Z)=-1$ or $+1$ corresponding to the original label $1$ or $2$.

Since normally distributed random variable can be taken any values on $\mathbb{R}$, for any $z_j \sim N(\mu_{1j}, \lambda_{1j})$, $j=1, \ldots, J$, we have $(z_{1j}-\mu_{1j})/\sqrt{\lambda_{1j}}\in \left[-M_n, M_n \right]$ in probability $1-2\Phi(-M_n)$, where $M_n\rightarrow \infty$ as $n\rightarrow \infty$, and $\Phi(-M_n)=o(1)$ when $n$ is large. Therefore with probability $1-O(J\Phi(-M_n))$, we can restrict the domain for $(z_j-\mu_{1j})/\sqrt{\lambda_{1j}}$ by $ \left[-M_n, M_n \right]^J$. Let $M_n\asymp \sqrt{\log Jn}$, then $\Phi(-M_n)\asymp (Jn)^{-1}$. For $J \leq n$, $M_n\lesssim \sqrt{\log n}$, $J\Phi(-M_n)=O(n^{-1})$. Assume $\lambda_{1j}/\lambda_{2j}$ is an absolute constant for any $j$. Define excess $\phi$-risk $\mathcal{E}_\phi\left( f, f^\ast_\phi \right) = {E}\left[\phi(Y(Z)f(\bs Z)) - \phi(Y(Z)f^\ast_\phi(\bs Z)) \right]$, where $f^\ast_\phi = \arg\min_f {E}\left[\phi(Y(Z)f(\bs Z))\right]$, $\bs Z$ is a random vector identically distributed as the first $J$ projection scores of $Z$.

\begin{lemma}\label{LEM: Q approx}
With probability greater than $1-O(J/n)$, there exists a network $\widetilde{Q}(\bs{z})$ satisfying $||\widetilde{Q}(\bs{z})-Q(\bs{z})||_\infty \leq \epsilon_{nJ}$  and $\widetilde{Q}(\bs{z}) \in \mathcal{F}(L, J, \bs{p}, s, B)$, where network class $\mathcal{F}(L, J, \bs{p}, s, B)$ satisfying $L\lesssim \log(p_{nJ}/\epsilon_{nJ})$, $\bs{p}=(J, p_{nJ}, 4p_{nJ}, \ldots, 4p_{nJ}, J, 1)$, $s\lesssim \log(p_{nJ}/\epsilon_{nJ})p_{nJ}$, $B=1$, $p_{nJ}\asymp J\log n$.
\end{lemma}
\begin{proof}
W first constrain the absolute value of all parameters   no more than one.

Without loss of generality, we assume $\bs z \sim N(\bs\mu_1, \Sigma_1)$. By (\ref{EQ:Qz}), we can rewrite $Q(\bs{z}) = \sum_{j=1}^J g_{0j}(z_j)$, where $g_{0j}(z_j) = (\epsilon_j-1)\left( (z_j-\mu_{1j})/{\sqrt{\lambda_j^{(1)}}} + \delta_j\right)^2+c_j = (\epsilon_j-1)\left( d_j + \delta_j\right)^2+c_j$, such that $d_j=(z_j-\mu_{1j})/{\sqrt{\lambda_j^{(1)}}} \sim N(0, 1)$, 
\[
\delta_j =  \frac{\mu_{1j} - \mu_{2j}}{\sqrt{ \lambda_j^{(2)}}}\times\frac{\epsilon_j^{1/2}}{\epsilon_j-1},\,\,\,\,
c_j = - \frac{(\mu_{1j} - \mu_{2j})^2  }{  \lambda_j^{(2)} (\epsilon_j-1)} -\log \epsilon_j .
\]
When $\Phi(-M_n) = o(1)$, with probability greater than $\left[ 1-2\Phi(-M_n)\right]^J \approx 1-O(J\Phi(-M_n))$, $d_j\in \left[-M_n, M_n\right]$ for all $1\leq j\leq J$. Let $M_n\asymp \log ^{1/2} n$, since $\Phi^{-1}(n^{-1})\approx \sqrt{2\log n}$, we have $\Phi(-M_n)\asymp n^{-1}$. By this procedure, we confine all $d_j$'s in $\left[-\log ^{1/2}n, \log ^{1/2}n \right]$, with probability greater than $1-O(J/n)$, for all $J \ll n$. {To apply Lemma A.1 in \cite{Schmidt:19}}, we need to transform $t_j$ into $\left[0, 1\right]$, such that we can rewrite $g_{0j}(z_j) = 4M_n^2(\epsilon_j-1)\left\{\left[s_j - \left(\frac{1}{2}-\frac{\delta_j}{M_n} \right)\right]^2+\frac{c_j}{4M_n^2(\epsilon_j-1)}\right\}$, where $s_j = (2M_n)^{-1}d_j + \frac{1}{2} \in \left[ 0, 1\right]$. Define
\begin{eqnarray*}
r_j(s_j)&=&\left[s_j - \left(\frac{1}{2}-\frac{\delta_j}{M_n} \right)\right]^2+\frac{c_j}{4M_n^2(\epsilon_j-1)}\\ 
& =&-s_j(1-s_j) + \frac{2\delta_j}{M_n}s_j + \left(\frac{1}{2}-\frac{\delta_j}{M_n}\right)^2 +  \frac{c_j}{4M_n^2(\epsilon_j-1)}\\
&=&r_0(s_j) + h(s_j) + \pi_j,
\end{eqnarray*}
where $r_0(s_j)=-s_j(1-s_j)$, $h(s_j)=\frac{2\delta_j}{M_n}s_j$, $\pi_j=\left(\frac{1}{2}-\frac{\delta_j}{M_n}\right)^2 +  \frac{c_j}{4M_n^2(\epsilon_j-1)}\in \left(0, 1\right)$. Therefore,  there exists a network $\widetilde{r}_j \in \mathcal{F}\left( m+1, (1, 4, 4, \ldots, 4, 1)\right)$ that computes the function $y \mapsto \sum_{k=1}^m R^k(y) + h(y) + \pi_j$, such that $\|\sum_{k=1}^m R^k(y)-r_0(y) \|_\infty \leq 2^{-m}$, i.e. $\|\widetilde{r}_j-r_0(y) \|_\infty \leq 2^{-m}$. Note that for $\widetilde{r}_j$, the first layer computes $y \mapsto\left(T_+(y), T^1_-(y), h(y), \pi_j \right)$, where $T_+(y)=(y/2)_+$ and $T^1_-(y) = (y-1/2)_+$.

Now we have $\|g_{0j}(x_j) - 4M_n^2(\epsilon_j-1)\widetilde{r}_j\|_\infty \leq 4M_n^2(\epsilon_j-1)2^{-m}$, and we need to explore the structure of $4M_n^2(\epsilon_j-1)\widetilde{r}_j$. As a matter of fact, since all parameters are bounded by one, define $\widetilde{g}_{0j}\in \mathcal{F}\left( m+2, (1, p_{nj},  4p_{nj}, \ldots, 4p_{nj}, 1)\right)$, such that the first layer computes $y \mapsto \underbrace{(y, \ldots, y)}_{p_{nj}} $, and each element in the second layer connects a network $\widetilde{r}_j$, where $p_{nj} = \lceil 4M_n^2(\epsilon_j-1)\rceil\asymp \log n$. Therefore, we have $\|g_{0j}(x_j) - \widetilde{g}_{0j}(x_j)\|_\infty \leq p_{nj}2^{-m}$.

Next, we define the network $\widetilde{Q}$ with additive structure $\widetilde{Q} = \sum_{j=1}^J \widetilde{g}_{0j}$, such that $\widetilde{Q}\in \mathcal{F}\left( m+3, (J, p_{nJ},  4p_{nJ}, \ldots, 4p_{nJ}, J, 1)\right)$, where $p_{nJ}=\sum_{j=1}^J p_{nj}\asymp J\log n$.  Define $\epsilon_{nJ} = p_{nJ}2^{-(m+2)}$, we have $\widetilde{Q}\in \mathcal{F}\left( \log(p_{nJ}\epsilon^{-1}_{nJ}), (J, p_{nJ},  4p_{nJ}, \ldots, 4p_{nJ}, J, 1), s, 1, 1\right)$, such that $\|\widetilde{Q}-Q \|_\infty \leq \epsilon_{nJ}$ and $s\lesssim \log(p_{nJ}\epsilon^{-1}_{nJ})p_{nJ}$. The upper bound for $s$ is obtained by the fact that there are no active nodes between $\widetilde{g}_{0j}$ and $\widetilde{g}_{0j'}$ for any $1\leq j \neq j' \leq J$.
\end{proof}

\begin{lemma}\label{LEM: phi norm approx}
There exists an $f_n(\bs{x}) \in \mathcal{F}(L, J, \bs{p}, s, B)$ satisfying $\mathcal{E}_\phi\left( f_n, f^\ast_\phi \right) \lesssim \epsilon_{nJ}^2$, such that $L\lesssim \log(p_{nJ}/\epsilon_{nJ})$, $\bs{p}=(J, 4 p_{nJ},  16 p_{nJ}, \ldots, 16 p_{nJ}, 4J, 4 ,5, 1)$, $s\lesssim \log(p_{nJ}/\epsilon_{nJ})p_{nJ}$, $B \lesssim \epsilon_{nJ}^{-1}$, where $p_{nJ}\asymp J\log n$.
\end{lemma}
\begin{proof}

We use the network $\widetilde{Q}$ obtained in Lemma \ref{LEM: Q approx} and construct a network $$\widetilde{f} = 2\left(\sigma\left(\epsilon_{nJ}^{-1}\widetilde{Q} \right) - \sigma\left( \epsilon_{nJ}^{-1}\widetilde{Q}-1\right)\right)- 1 .$$
We need two more layers from $\widetilde{Q}$ to $\widetilde{f}$, which can be obtained by
$$\widetilde{Q} \mapsto \sigma\left(\epsilon_{nJ}^{-1}\widetilde{Q} \right),$$
$$\widetilde{Q} \mapsto  \sigma\left(\epsilon_{nJ}^{-1}\widetilde{Q}-1 \right)$$ with the maximal value of weights is bounded above by $\epsilon_{nJ}^{-1}$.
Since the subtraction is multiplied by two, we need double the width, and the last layer of additive structure with bias term $-1$, we have $$\widetilde{f} \in \mathcal{F}\left(m+4, (J, 4 p_{nJ},  16 p_{nJ}, \ldots, 16 p_{nJ}, 4J, 4 ,5, 1), s, \lceil \epsilon_{nJ}^{-1} \rceil, 1\right).$$
Define $D=\left\{\bs z: | {Q}(\bs z) |> 2\epsilon_{nJ} \right\}$, then $\widetilde{f}(\bs z) = G^\ast_J(\bs z)$ when $\bs z\in D$. This is because when $Q(\bs z) > 2\epsilon_{nJ}$, we have
$$\widetilde{Q}(\bs z) = \widetilde{Q}(\bs z) -  {Q}(\bs z) + {Q}(\bs z) > -\epsilon_{nJ} + 2\epsilon_{nJ} =\epsilon_{nJ}$$
and when $Q(\bs z) < -2\epsilon_{nJ}$,
$$\widetilde{Q}(\bs z) = \widetilde{Q}(\bs z) -  {Q}(\bs z) + {Q}(\bs z) < \epsilon_{nJ} - 2\epsilon_{nJ} <0. $$
The set $D$ is equivalent to $\left\{\bs z: | P(Y( Z)=1 | \bs Z =  \bs z) - \frac{1}{2}| > \frac{1}{2}\epsilon_{nJ} \right\}$. To see this,  note that we have
\begin{eqnarray*}
\left\{ \left| P(Y( Z)=1 | \bs Z = \bs z) - \frac{1}{2}\right| \leq \frac{1}{2}\epsilon_{nJ}\right\} =&& \left\{ \left|\frac{f_1(\bs z)}{f_1(\bs z) + f_2(\bs z)} - \frac{1}{2}\right| \leq \frac{1}{2}\epsilon_{nJ}\right\}\\
 =&& \left\{ \left|\frac{1-e^{-Q(\bs{z};\bs{\theta})}}{1 +e^ {-Q(\bs{z};\bs{\theta})}}\right| \leq \epsilon_{nJ}\right\} \\
 =&& \left\{\frac{1-\epsilon_{nJ}}{1+\epsilon_{nJ}}\leq e^ {-Q(\bs{z};\bs{\theta})}\leq \frac{1+\epsilon_{nJ}}{1-\epsilon_{nJ}} \right\} \\
 \asymp && \left\{-2\epsilon_{nJ} \leq Q(\bs{z};\bs{\theta}) \leq 2\epsilon_{nJ} \right\}, 
\end{eqnarray*}
where $f_k$ are the density function for the $k$-th class.
By the proof of Lemma \ref{LEM:fq}, we have the density for $Q(\bs{z};\bs{\theta})$ is bounded above, therefore, when $\epsilon_{nJ}$ is sufficiently small around $0$, we have $ P_{\bs\theta}\left(-2\epsilon_{nJ} \leq Q(\bs{z};\bs{\theta}) \leq 2\epsilon_{nJ} \right) \leq C \epsilon_{nJ}$, where $C$ is some positive constant.

Therefore,
\begin{eqnarray*}
&&E\left[\phi\left(Y( Z)\widetilde{f}(\bs Z)\right) -  \phi\left(Y( Z)G^\ast_J(\bs Z)\right)\right] \\
&=& \int |\widetilde{f}( \bs z) - G^\ast_J( \bs z)| | 2 P(Y( Z)=1 | \bs Z =  \bs z)-1 | dP_{\bs Z}( \bs z)\\
&\asymp & \int_{D^c} |\widetilde{f}( \bs z) - G^\ast_J( \bs z)| | 2P(Y( Z)=1 | \bs Z =  \bs z)-1 | dP_{\bs Z}( \bs z)\\
&\leq & 2\epsilon_{nJ} P\left(\left\{ Z: | 2P(Y( Z)=1 | \bs Z =  \bs z)-1 | \leq  \epsilon_{nJ} \right\} \right)\\
&\leq & 4\epsilon_{nJ}^2.
\end{eqnarray*}

\end{proof}

For simplicity, define the network class in Lemma \ref{LEM: phi norm approx} as $\mathcal{F}_n$. By Lemma \ref{entropy}, we have 
\begin{eqnarray*}
\log \mathcal{N}(\delta, \mathcal{F}_n, \|\cdot \|_\infty) && \lesssim \log^2(p_{nJ}/\epsilon_{nJ})p_{nJ}\log (\delta^{-1}\log(p_{nJ}/\epsilon_{nJ})p_{nJ}/\epsilon_{nJ})\\
&& \lesssim \log^2(p_{nJ}/\epsilon_{nJ})\log (\delta^{-1}p_{nJ}/\epsilon_{nJ})p_{nJ}.
\end{eqnarray*}

We also claim two similar lemmas for approximation of $\widehat{Q}(\bs z)$ when data are discretely observed.
\begin{lemma}\label{LEM: Q approx sampling}
With probability more than $1-O(J/n)$, there exists a network $\widetilde{Q}^{(s)}(\bs{z}) \in \mathcal{F}(L, J,   \bs{p}, s, B)$ satisfying $\|\widetilde{Q}^{(s)}(\bs{z})-\widehat{Q}(\bs{z})\|_\infty \leq \epsilon_{nJ}$, such that $L\lesssim \log(p_{nJ}/\epsilon_{nJ})$, $\bs{p}=(J, p_{nJ}, 4p_{nJ}, \ldots, 4p_{nJ}, J, 1)$, $s\lesssim \log(p_{nJ}/\epsilon_{nJ})p_{nJ}$, $B=1$, where $p_{nJ}\asymp J\log n$.
\end{lemma}
\begin{proof}
$\widehat{Q}(\bs z)$ and ${Q}(\bs z)$ essentially have the same quadratic form with different parameters, which can be approximated similarly, the proof is similar to the proof of Lemma \ref{LEM: Q approx}.
\end{proof}

\begin{lemma}\label{LEM: phi norm approx sampling}
There exists a network $f_n^{(s)}(\bs{x}) \in \mathcal{F}(L, J, \bs{p}, s, B)$ satisfying $\mathcal{E}_\phi\left( f_n^{(s)}, \widehat{f}^\ast_\phi \right) \lesssim \epsilon_{nJ}^2 + \Delta_M^2 + n^{-1}$, such that $L\lesssim \log(p_{nJ}/\epsilon_{nJ})$, $\bs{p}=(J, 4 p_{nJ},  16 p_{nJ}, \ldots, 16 p_{nJ}, 4J, 4 ,5, 1)$, $s\lesssim \log(p_{nJ}/\epsilon_{nJ})p_{nJ}$, $B \lesssim \epsilon_{nJ}^{-1}$, where $p_{nJ}\asymp J\log n$.
\end{lemma}
\begin{proof}
We use the network $\widetilde{Q}^{(s)}$ obtained in Lemma \ref{LEM: Q approx sampling} and construct a network $$\widetilde{f} = 2\left(\sigma\left(\epsilon_{nJ}^{-1}\widetilde{Q}^{(s)} \right) - \sigma\left( \epsilon_{nJ}^{-1}\widetilde{Q}^{(s)}-1\right)\right)- 1 .$$
We need two more layers from $\widetilde{Q}$ to $\widetilde{f}$, which can be obtained by
$$\widetilde{Q}^{(s)} \mapsto \sigma\left(\epsilon_{nJ}^{-1}\widetilde{Q}^{(s)} \right),$$
$$\widetilde{Q}^{(s)} \mapsto  \sigma\left(\epsilon_{nJ}^{-1}\widetilde{Q}^{(s)}-1 \right)$$ with the maximal value of weights is bounded above by $\epsilon_{nJ}^{-1}$, where $\sigma(\cdot)$ is ReLU activation function.
Since the subtraction is multiplied by two, we need double the width, and the last layer of additive structure with bias term $-1$, we have $$\widetilde{f}^{(s)} \in \mathcal{F}\left(m+4, (J, 4 p_{nJ},  16 p_{nJ}, \ldots, 16 p_{nJ}, 4J, 4 ,5, 1), s, \lceil \epsilon_{nJ}^{-1} \rceil, 1\right).$$
Define $D^{(s)}=\left\{\bs z: | {Q}(\bs z) |> 2\epsilon_{nJ} + \Delta_M \right\}$, then $\widetilde{f}(\bs z) = G^\ast_J(\bs z)$ when $\bs z\in D$. This is because when $Q(\bs z) > 2\epsilon_{nJ} + \Delta_M$, with probability $1-O(n^{-1})$, 
$$\widetilde{Q}(\bs z) = \widetilde{Q}(\bs z) - \widehat{Q}(\bs z) + \widehat{Q}(\bs z) - {Q}(\bs z) + {Q}(\bs z) > -\epsilon_{nJ} - \Delta_M + 2\epsilon_{nJ} + \Delta_M  =\epsilon_{nJ}$$
and when $Q(\bs z) < -2\epsilon_{nJ}- \Delta_M$, with probability $1-O(n^{-1})$,
$$\widetilde{Q}(\bs z) = \widetilde{Q}(\bs z) - \widehat{Q}(\bs z) + \widehat{Q}(\bs z) -  {Q}(\bs z) + {Q}(\bs z) < \epsilon_{nJ} + \Delta_M - 2\epsilon_{nJ}- \Delta_M <0, $$
where we use the fact as discussed in Lemmas \ref{LEM:sample} and \ref{LEM:Mz} that $\| \widehat{Q} - Q\|_\infty = \| M(\bs z)\|_\infty \lesssim \Delta_{M}$ in probability $1-O(n^{-1})$.

As discussed in Lemma \ref{LEM: phi norm approx}, $D^{(s)}=\left\{\bs z: |  P(Y( Z)=1 | \bs Z = \bs z)-1/2 |> 2\epsilon_{nJ}+\frac{\Delta_M}{4} \right\}$, we have
\begin{eqnarray*}
&& E\left[\phi\left(Y( Z)\widetilde{f}^{(s)}(\bs Z)\right) -  \phi\left(Y( Z)\widehat{G}^\ast_J(\bs Z)\right)\right] \\
=&&  E_{M( \bs z)}\left[\int\left. |\widetilde{f}^{(s)}( \bs z) - \widehat{G}^\ast_J( \bs z)| | 2 P(Y( Z)=1 |\bs Z = \bs z)-1 | dP_{\bs Z}(\bs z)\right| M(\bs z) \right]\\
=&&  E_{M(\bs z)}\left[\int\left. |\widetilde{f}^{(s)}(\bs z) - \widehat{G}^\ast_J(\bs z)| | 2 P(Y(Z)=1 |\bs Z = \bs z)-1 | dP_{\bs Z}(\bs z)\right| M(\bs z)\lesssim \Delta_M \right]\\
&&+ E_{M(\bs z)}\left[\int\left. |\widetilde{f}^{(s)}(\bs z) - \widehat{G}^\ast_J(\bs z)| | 2 P(Y(Z)=1 |\bs Z = \bs z)-1 | dP_{\bs Z}(\bs z)\right| M(\bs z)\gtrsim \Delta_M \right]\\
\asymp && E_{M(\bs z)}\left[\left.\int_{{(D^{(s)})}^c} |\widetilde{f}^{(s)}(\bs z) - \widehat{G}^\ast_J(\bs z)| | 2 P(Y(Z)=1 |\bs Z = \bs z)-1 | dP_{\bs Z}(\bs z)\right| M(\bs z)\lesssim \Delta_M\right]\\
&&+ n^{-1}\\
\leq && (8\epsilon_{nJ} + \Delta_M) P\left(\left\{ \bs Z: | 2 P(Y(Z)=1 |\bs Z = \bs z)-1 | \leq 4 \epsilon_{nJ}+\frac{1}{2}\Delta_m \right\} \right)(1-o(1)) + n^{-1}\\
\lesssim && \epsilon_{nJ}^2 + \Delta_M^2 + n^{-1}.
\end{eqnarray*}

\end{proof}

For simplicity, define the network class in Lemma \ref{LEM: phi norm approx sampling} as $\mathcal{F}_n^{(s)}$. By Lemma \ref{entropy}, we have 
\begin{eqnarray*}
\log \mathcal{N}(\delta, \mathcal{F}_n^{(s)}, \|\cdot \|_\infty) &&\lesssim \log^2(p_{nJ}/\epsilon_{nJ})p_{nJ}\log (\delta^{-1}\log(p_{nJ}/\epsilon_{nJ})p_{nJ}/\epsilon_{nJ})\\
&&\lesssim \log^2(p_{nJ}/\epsilon_{nJ})\log (\delta^{-1}p_{nJ}/\epsilon_{nJ})p_{nJ}.
\end{eqnarray*}

\subsection{Excess risk of FDNN classifier for full observed case}\label{Excess risk Full}
Let $Y_i=Y(Z_i)$.
Define $\widehat{f}_{\phi} = \arg\min_{f\in \mathcal{F}_n} \frac{1}{n_1+n_2}\sum_{i=1}^{n_1+n_2} \phi(Y_if(\mathbf{Z}_i))$, where $\bs Z_i$ is the random vector of the first $J$ projection scores for the $i$-th subject.
By {Theorem 2.13 of \cite{Steinwart:Christmann:08}}, $\mathcal{E}\left( \widehat{f}_{\phi}, G^\ast_J \right)\leq \mathcal{E}_\phi\left( \widehat{f}_{\phi}, f^\ast_\phi \right)$. where $G^\ast_J$ is the Bayes classifier of first $J$ scores. Consider $\Theta=\Theta_S(\nu_1, \nu_2)$, we have the following two lemmas for the excess risk with respect to the first $J$ scores. 

\begin{lemma}\label{LEM: phi excess risk J}
Under Gaussian assumption, assume $\mathcal{F}_n$ satisfies $H_B(\delta_n, \mathcal{F}_n, \|\cdot \|_2)\lesssim n\delta_n$, where $\| f\|_\infty \leq 1$ for all $f\in \mathcal{F}_n$. If there exists an $f_n\in \mathcal{F}_n$ satisfying $\mathcal{E}_{\phi}\left( f_n, f^\ast_\phi\right) \lesssim \theta_n^2$, the empirical $\phi$-risk minimizer $\widehat{f}_{\phi}$ over $\mathcal{F}_n$ satisfies $P\left( \mathcal{E}_\phi\left( \widehat{f}_{\phi}, f^\ast_\phi \right) \geq  \theta_n^2\right)\lesssim \exp\left( -10^{-4} n\theta_n^2\right)$.
\end{lemma}
\begin{proof}
In the following, we use $\mathcal{F}_n$ to denote a network space such that there exists an $f_n\in \mathcal{F}_n$ satisfying $\mathcal{E}_{\phi}\left( f_n, f^\ast_\phi\right) \lesssim \theta_n^2$.  Define the following empirical process:
\begin{eqnarray*}
E_n(f) &&=\frac{1}{n_1+n_2}\sum_{i=1}^{n_1+n_2}\left[\phi(Y_if_n(\bs Z_i)) - \phi(Y_if(\bs Z_i)) \right] - E\left[ \phi(Y( Z)f_n(\bs Z)) -\phi(Y( Z)f(\bs Z)) \right]\\
&&=\left[\mathcal{E}_{\phi, n}(f_n)- \mathcal{E}_{\phi, n}(f)\right] - \left[ \mathcal{E}_{\phi}(f_n)- \mathcal{E}_{\phi}(f) \right],
\end{eqnarray*}
where $f_n\in \mathcal{F}_n$ satisfies $\mathcal{E}_{\phi}\left( f_n, f^\ast_\phi\right)\leq \theta_n^2/2$. Since $\widehat{f}$ satisfies $\mathcal{E}_{\phi, n}(\widehat{f}_n)- \mathcal{E}_{\phi, n}(f_n) \leq 0$, if $\widehat{f}_n\in \mathcal{F}_n$, such that  $\mathcal{E}_{\phi}(\widehat{f}_n,f^\ast_\phi)\geq \theta_n^2$, then  $\widehat{f}_n\in \left\{ f\in \mathcal{F}_n | \mathcal{E}_{\phi}(f,f^\ast_\phi)\geq \theta_n^2, \mathcal{E}_{\phi, n}(f)- \mathcal{E}_{\phi, n}(f_n) \leq 0\right\}$, i.e.
\begin{equation}
P\left( \mathcal{E}_{\phi}(\widehat{f}_n,f^\ast_\phi)\geq \theta_n^2\right) \leq \mathbf{P}^\ast\left( \sup_{f\in \mathcal{F}_n, \mathcal{E}_{\phi}(f,f^\ast_\phi)\geq \theta_n^2} \mathcal{E}_{\phi, n}(f_n)- \mathcal{E}_{\phi, n}(f)\geq 0\right),
\end{equation}
where $\mathbf{P}^\ast$ is the outer measure.
Let $\mathcal{F}_{n, i} = \left\{ f\in \mathcal{F}_n: 2^{i-1}\theta_n^2 \leq \mathcal{E}_{\phi}(f,f^\ast_\phi) < 2^i\theta_n^2\right\}$ be the partition of the space $\left\{ f\in \mathcal{F}_n: \mathcal{E}_{\phi}(f,f^\ast_\phi) \geq \theta_n^2\right\}$, since
\begin{eqnarray*}
\mathcal{E}_{\phi}(f,f^\ast_\phi) \leq E|\phi(Y(Z)f(\bs Z)) -\phi(Y(Z)f^\ast_\phi(\bs Z)) | \leq E |f(\bs Z) - f^\ast_\phi(\bs Z)| \leq 2,
\end{eqnarray*}
there exists an $i^\ast=\lceil \log \theta_n^{-2}/\log 2\rceil$, such that $\left\{ f\in \mathcal{F}_n: \mathcal{E}_{\phi}(f,f^\ast_\phi) \geq \theta_n^2\right\} \subseteq \cup_{i=1}^{i^\ast}\mathcal{F}_{n, i}$.

To apply Theorem 3 of \cite{Shen:Wong:94}, we will investigate some conditions. For any $\mathcal{F}_{n, i} $, we have
\begin{eqnarray*}
\inf_{f\in \mathcal{F}_{n, i}}E\left[\phi(Y( Z)f(\bs Z)) -\phi(Y( Z)f_n(\bs Z)) \right] = \inf_{f\in \mathcal{F}_{n, i}}\left\{ \mathcal{E}_{\phi}(f,f^\ast_\phi) -  \mathcal{E}_{\phi}(f_n,f^\ast_\phi)\right\} \geq 2^{i-2}\theta_n^2
\end{eqnarray*}
and
\begin{eqnarray*}
&&\sup_{f\in \mathcal{F}_{n, i}}E\left[\phi(Y( Z)f(\bs Z)) -\phi(Y( Z)f_n(\bs Z)) \right]^2\\
&\leq& 2\sup_{f\in \mathcal{F}_{n, i}}E|\phi(Y( Z)f(\bs Z)) -\phi(Y( Z)f_n(\bs Z)) | \\
&\leq& 2\times2^i\theta_n^2 + 2\times \left(\theta_n^2/2\right)  \leq 8\left( 2^{i-2}\theta_n^2\right). 
\end{eqnarray*}
For simplicity, define $M_{n, i}=2^{i-2}\theta_n^2$, and we have
\begin{equation}\label{EQ:Partition}
P\left( \mathcal{E}_{\phi}(\widehat{f}_n,f^\ast_\phi)\geq \theta_n^2\right) \leq \sum_{i=1}^{i^\ast}\mathbf{P}^\ast\left( \sup_{f\in \mathcal{F}_n} E_n(f)\geq M_{n, i}\right).
\end{equation}
To apply Theorem 3 of \cite{Shen:Wong:94}, define $\mathcal{H}_{n, i} = \left\{ (\bs z, y) \mapsto \phi(yf_n(\bs z))- \phi(yf(\bs z)) \right\}$, such that $\| h\|_\infty = \| f_n-f\|_\infty \leq 2$ for $h\in \mathcal{H}_{n, i}$. We have $\sup_{h\in \mathcal{H}_{n, i}}Var\left(h( Z, Y) \right) \leq 8 M_{n, i}$. Define $\nu_{n, i} = 24 M_{n, i}$, and $\nu_{n, i} \leq 24 $ by the fact that $\theta_n^2 \leq 2^{2-i}$. Let $F_{n, i} = \sqrt{24}$, then $\nu_{n, i}^{1/2} \leq F_{n, i}$, and ${4F_{n, i}M_{n, i}}{\nu_{n, i}}^{-1} = 4\times 24^{-1/2} < 1$. Define $\xi = \sqrt{2/3}$.

Since $H_B\left( u, \mathcal{H}_{n, i}, \|\cdot \|_2\right)\leq H_B\left( u, \mathcal{F}_{n}, \|\cdot \|_2\right)$, by the Lipschitz continuity w.r.t hinge loss with constant $1$, when $H_B\left( u, \mathcal{H}_{n, i}, \|\cdot \|_2\right)\leq H_B\left( u, \mathcal{F}_{n}, \|\cdot \|_2\right) \leq Cnu$, we have
\begin{eqnarray*}
  M_{n, i}^{-1}\int_{\xi M_{n, i}/32}^{\nu_{n, i}^{1/2}} H_B\left( u, \mathcal{H}_{n, i}, \|\cdot \|_2\right)^{1/2} du  
&  \leq& M_{n, i}^{-1} \nu_{n, i}^{1/2} H_B\left( \xi M_{n, i}/32, \mathcal{H}_{n, i}, \|\cdot \|_2\right)^{1/2} \\
&& = 24^{1/2}M_{n, i}^{-1/2}C^{1/2}n^{1/2}\xi^{1/2}M_{n, i}^{1/2}32^{-1/2} \\
&& = (3C/2)^{1/2}\xi^{3/2} n^{1/2} \leq \frac{\xi^{3/2} n^{1/2}}{2^{10}},  \\
\end{eqnarray*}
for $C = 2^{-21}$.
Furthermore, the derivation above implies
\begin{eqnarray*}
&&H_B\left( \nu_{n, i}, \mathcal{H}_{n, i}, \|\cdot \|_2\right)^{1/2}\\
&& \leq \frac{M_{n,i}}{\nu_{n, i}^{1/2}-M_{n,i}/64} M_{n,i}^{-1}\int_{\xi M_{n,i}/32}^{\nu_{n, i}^{1/2}}H_B\left( u, \mathcal{H}_{n, i}, \|\cdot \|_2\right)^{1/2} du\\
&&\leq \frac{1}{\nu_{n, i}^{1/2}-M_{n,i}/64}\frac{\sqrt{2/3}}{2^{10}}{M_{n,i}\xi^{1/2}n^{1/2}} \\
&&= \left(\frac{1}{3\times 2^{19}}\right)^{1/2}\left(\frac{1}{\nu_{n, i} + M_{n,i}^2/4096 - \nu_{n, i}^{1/2}M_{n,i}/32} \right)^{1/2}\left(\xi n{M_{n,i}^2}\right)^{1/2}\\
&&\leq \left(\frac{1}{8\left( 4\nu_{n, i} + M_{n,i}F_{n, i}/3\right)}\right)^{1/2}\left(\xi n{M_{n,i}^2}\right)^{1/2}.
\end{eqnarray*}
Hence all conditions are satisfied, and we apply Theorem 3 in \cite{Shen:Wong:94} to $\mathcal{H}_{n,i}$, and equation (\ref{EQ:Partition}) is bounded as
\begin{eqnarray*}
P\left( \mathcal{E}_{\phi}(\widehat{f}_n,f^\ast_\phi)\geq \theta_n^2\right) &&\leq \sum_{i=1}^{i^\ast}\mathbf{P}^\ast\left( \sup_{f\in \mathcal{F}_n} E_n(f)\geq M_{n, i}\right)\\
&&\leq \sum_{i=1}^{i^\ast}3\exp\left\{ -(1-\xi) \frac{nM_{n,i}^2}{2(4\nu_{n, i} + M_{n, i}F_{n, i}/3)}\right\} \\
&&= \sum_{i=1}^{i^\ast}3\exp\left\{ -(1-\sqrt{2/3}) \frac{nM_{n,i}}{2(96 + \sqrt{24}/3)}\right\} \\
&&\leq \sum_{i=1}^{i^\ast}3\exp\left\{ - 10^{-4}nM_{n,i} \right\}  = \sum_{i=1}^{i^\ast}3\exp\left\{ - 10^{-4} 2^{i-2} n\theta_n^2 \right\} \\
&&\lesssim \exp\left\{ - 10^{-4} n\theta_n^2 \right\}
\end{eqnarray*}
\end{proof}

Let $\widehat{G}^{FDNN}$ be the proposed classifier $\widehat{f}_{\phi}\in \mathcal{F}_n$ . As a result, we have the following lemma.
\begin{lemma}\label{LEM: excess risk all}The proposed FDNN classifier over $\Theta$ satisfies
$\sup_{\theta \in \Theta} \mathcal{E}\left( \widehat{G}^{FDNN}, G^\ast_{\bs\theta} \right) \lesssim J^\ast n^{-1}\log^4 n $, where $\widehat{f}_{\phi}\in \mathcal{F}_n$ with $J\asymp J^\ast$.
\end{lemma}
\begin{proof}
Let $\mathcal{F}_n$ discussed in Lemma \ref{LEM: phi excess risk J} be the network space mentioned in Lemma \ref{LEM: phi norm approx}, and $\epsilon_{nJ}^2=\theta_n^2$, then $\sup_{\theta \in \Theta_J} \mathcal{E}\left( \widehat{f}_{\phi}, G^\ast_J \right) \lesssim \epsilon_{nJ}^2$, where $\epsilon_{nJ}^2 \gtrsim Jn^{-1}\log^4 n$, $G^\ast_J$ is the Bayes classifier for the first $J$ scores. The asymptotic lower bound for $\epsilon_{nJ}^2$ is obtained by the bound of $\epsilon_{nJ}^2$-bracketing entropy. The rest of proof is the same as Lemma \ref{LEM:GJ}.
\end{proof}

\subsubsection{Proof of Theorem \ref{THM:full:dnn}}
\begin{proof}
This proof of this theorem directly follows Lemmas \ref{LEM:self-similar full} and \ref{LEM: excess risk all}. The specific choice of $(L, J, \bs p, s, B)$ can be derived from the network structure in Lemma \ref{LEM: phi norm approx}, with $\epsilon_{nJ}^2 = J^\ast n^{-1}\log^4 n$, such that $J^\ast=  \left({n}/{\log^4 n}\right)^{1/(1+\nu_2)}$.

\end{proof}

\subsection{Excess risk of sFDNN classifier for discretely observed case}\label{Excess risk Sampling}
\begin{lemma}\label{LEM: phi excess risk J sampling}
Under Gaussian assumption, assume $\mathcal{F}_n^{(s)}$ satisfies $H_B(\delta_n, \mathcal{F}_n^{(s)}, \|\cdot \|_2)\lesssim \iota\delta_n$, such that $\iota\lesssim n$, if there exists an $f^{(s)}_n\in \mathcal{F}_n$ satisfying $\mathcal{E}_{\phi}\left( f^{(s)}_n, f^\ast_\phi\right) \lesssim \theta_n^2$, the empirical $\phi$-risk minimizer $\widehat{f}_{\phi}^{(s)}$ over $\mathcal{F}_n^{(s)}$ satisfies $P\left( \mathcal{E}_\phi\left( \widehat{f}_{\phi}^{(s)}, \widehat{f}^\ast_\phi \right) \geq  \theta_n^2\right)\lesssim \exp\left( -10^{-4} n\theta_n^2\right)$.
\end{lemma}
\begin{proof}
The condition for $\delta_n$-bracketing entropy can be easily bounded by the same $n$-order by the fact that $\iota \lesssim n$, and the rest part can simply follow the proof of Lemma \ref{LEM: phi excess risk J}.
\end{proof}

Let $\widehat{G}^{sFDNN}$ be the proposed classifier $\widehat{f}_{\phi}^{(s)} \in \mathcal{F}^{(s)}_n$. As a result, we have the following lemma.
\begin{lemma}\label{LEM: excess risk all sampling}Consider the parameter space $\Theta$. Then the sFDNN classifier satisfies
$\sup_{\theta \in \Theta} \mathcal{E}\left( \widehat{G}^{sFDNN}, G^\ast_{\bs\theta} \right) \lesssim J^\ast_1f(M) + J^\ast_2 n^{-1}\log^4 n $, where $J^\ast_1=M^{\nu_1/(1+\nu_2)}$ and $ J^\ast_2 =  \left({n}/{\log^4 n}\right)^{1/(1+\nu_2)}$, and $\widehat{f}_{\phi}^{(s)} \in \mathcal{F}^{(s)}_n$ with $J\asymp J^\ast_1(M \leq M^\ast) + J^\ast_2\mathbb{I}(M \geq M^\ast)$.
\end{lemma}

\begin{proof}
It holds for $\sup_{\theta \in \Theta_J} \mathcal{E}\left( \widehat{f}_{\phi}^{(s)}, \widehat{G}^\ast_J \right) \lesssim \epsilon_{nJ}^2 + \Delta^2_M$, where $\epsilon_{nJ}^2 \gtrsim Jn^{-1}\log^4 n$, $\Delta^2_M \gtrsim J f(M)$, and $\widehat{G}^\ast_J$ is the estimated Bayes classifier for the first $J$ scores by FQDA method. To see this, consider two situations. When $\epsilon_{nJ} \gtrsim \Delta_M$, let $\iota = n$ and $\theta_n^2=\epsilon_{nJ}^2$, then the result can be simply derived from Lemma \ref{LEM: phi excess risk J} and the proof in Lemma \ref{LEM: excess risk all}; Otherwise, when $\epsilon_{nJ} \lesssim \Delta_M$, and we only need to consider $\epsilon_{nJ}^2 = J/n\log^4 n$, let $\iota = (f(M))^{-1}\log^4 n$ and $\theta_n^2=\Delta_M^2$, since $\epsilon_{nJ} \lesssim \Delta_M$, $\iota \lesssim n$ is satisfied, and $n\Delta_M^2 = J\frac{\log^4 n \Delta_M^2}{\epsilon_{nJ}^2} \lesssim J\log^4 n $ diverges, which makes the upper bound of probability of excess $\phi$-risk exceeding $\theta_n^2$ goes to $0$, therefore, $\sup_{\theta \in \Theta_J} \mathcal{E}\left( \widehat{f}_{\phi}^{(s)}, \widehat{G}^\ast_J \right) \lesssim \Delta^2_M$, such that $\Delta^2_M \gtrsim J f(M)$ is derived by bound of $\delta_n$-bracketing entropy. Since $\sup_{\theta \in \Theta_J} \mathcal{E}\left( \widehat{f}_{\phi}^{(s)},{G}^\ast_J \right) = \sup_{\theta \in \Theta_J} \mathcal{E}\left( \widehat{f}_{\phi}^{(s)},\widehat{G}^\ast_J \right) + \sup_{\theta \in \Theta_J} \mathcal{E}\left( \widehat{G}^\ast_J,{G}^\ast_J \right)$, and the second part is bounded by $Jn^{-1}\log n+J f(M)$ (see Lemma \ref{LEM:sample} and Theorem \ref{THM:full} (upper bound) ), we have $\sup_{\theta \in \Theta_J} \mathcal{E}\left( \widehat{f}_{\phi}^{(s)},{G}^\ast_J \right) \lesssim \epsilon_{nJ}^2 + \Delta^2_M$, where $\epsilon_{nJ}^2 \gtrsim Jn^{-1}\log^4 n$ and $\Delta^2_M \gtrsim J f(M)$. By the conclusions in Lemmas \ref{LEM:sample} and \ref{LEM:GJ}, the result is straightforward. 

\end{proof}

\subsubsection{Proof of Theorem \ref{THM:sampling:dnn}}
\begin{proof}
This proof of this theorem directly follows from Lemmas \ref{LEM:self-similar} and \ref{LEM: excess risk all sampling}. The specific choice of $L, J, \bs p, s, B$ can be derived from the network structure in Lemma \ref{LEM: phi norm approx sampling}, with $\epsilon_{nJ}^2 = J^\ast n^{-1}\log^4 n$ when $M\ge M^\ast$ and  $\epsilon_{nJ}^2=J^\ast f(M)$ when $M<M^\ast$, where $J^\ast=M^{\nu_1/(1+\nu_2)} \mathbb{I}(M \leq M^\ast) + \left({n}/{\log^4 n}\right)^{1/(1+\nu_2)}\mathbb{I}(M \geq M^\ast)$. 
\end{proof}

\bibliographystyle{imsart-number} 
\bibliography{Ref}       

\end{document}